\documentclass[aps,pra,twocolumn,superscriptaddress,10pt,longbibliography]{revtex4-2}
\usepackage{graphicx,xcolor}
\usepackage{amssymb,amsmath,amsthm}
\usepackage{amsfonts,microtype,bbm,mathtools}
\usepackage{braket}
\usepackage[title,titletoc]{appendix}
\usepackage[english]{babel}
\usepackage{booktabs}
\usepackage[bookmarksnumbered,linktocpage,hypertexnames=false,colorlinks=true,linkcolor=blue,urlcolor=blue,citecolor=blue,anchorcolor=green,breaklinks=true,pdfusetitle]{hyperref}
\usepackage[capitalize]{cleveref}

\newtheorem{thm}{Theorem}
\newtheorem{lem}[thm]{Lemma}
\newtheorem{prp}[thm]{Proposition}
\newtheorem{cor}[thm]{Corollary}
\newtheorem{hyp}[thm]{Hypothesis}
\newtheorem{dfn}[thm]{Definition}

\crefname{thm}{Theorem}{Theorems}
\crefname{lem}{Lemma}{Lemmas}
\crefname{prp}{Proposition}{Propositions}
\crefname{cor}{Corollary}{Corollaries}
\crefname{hyp}{Hypothesis}{Hypotheses}
\crefname{dfn}{Definition}{Definitions}

\DeclareMathOperator{\sgn}{sign}
\DeclareMathOperator{\alg}{Alg}
\DeclareMathOperator{\tr}{tr}
\DeclareMathOperator{\spn}{span}

\newcommand{\id}{\mathbbm{1}}

\newcommand{\D}{\mathcal D}

\newcommand{\eps}{\varepsilon}

\def\superopnorm#1{ {\left|\hspace{-.015in}\left|\hspace{-.015in}\left|#1\right|\hspace{-.015in}\right|\hspace{-.015in}\right|} }

\DeclarePairedDelimiter{\norm}{\lVert}{\rVert}
\DeclarePairedDelimiter{\abs}{\lvert}{\rvert}

\makeatletter
\newcommand{\prlsection}{%
  \@startsection{section}{1}{1em}{0.5ex}{-0.3ex}{\normalfont\slshape}%
}
\def\@runin@tos@section#1#2{#1#2\textnormal{.\textemdash}}
\AtBeginDocument{%
  \let\standardsection\section
  \let\section\prlsection
}
\makeatother

\begin{document}
\title{Contractivity of time-dependent driven-dissipative systems}

\author{Lasse H. Wolff}
\email{lhw@math.ku.dk}
\affiliation{Department of Mathematical Sciences, University of Copenhagen, Universitetsparken 5, 2100 Copenhagen, Denmark}

\author{Daniel Malz}
\email{daniel.malz@unibas.ch}
\affiliation{Department of Mathematical Sciences, University of Copenhagen, Universitetsparken 5, 2100 Copenhagen, Denmark}
\affiliation{Department of Physics, University of Basel, Switzerland}
\author{Rahul Trivedi}
\email{rahul.trivedi@mpq.mpg.de}
\affiliation{Max Planck Institute of Quantum Optics, Hans Kopfermann Str. 1, 85748 Garching, Germany}
\affiliation{Munich Center for Quantum Science and Technology, Schellingstr. 4, 80799 Munich, Germany}

\date{\today}
\raggedbottom
\begin{abstract}
	Experimentally relevant quantum systems are frequently described by a time-dependent Lindblad master equation in which the dissipation due to interaction with an environment is fixed, but the experimenter has control over the Hamiltonian, which therefore can change with time. A fundamental question that is also of practical relevance is to determine when such systems invariably lose information about their initial state (are exponentially contractive independent of the Hamiltonian) or when the control can be used to prevent this loss of information to some degree. Here, we identify the necessary and sufficient criterion on the dissipator for Hamiltonian-invariant contractivity. If the dissipator does not fulfil this criterion, but by itself is exponentially contractive, we also establish exponential contractivity for sufficiently weak or slow drives and construct examples where strong or fast driving prevents exponential contraction.
\end{abstract}

\maketitle
When quantum systems evolve in isolation, the resulting dynamics modeled by the Schrödinger equation is reversible. However, experimentally realistic quantum systems invariably couple to an environment \cite{Breuer2007theory, rivas2012open, carmichael1993open} and, under the Born–Markov and secular approximations, the dynamics of the reduced system state is modeled instead by the Lindblad master equation~\cite{lindblad1976generators, gorini1976completely, wolf2008dividing}. In general, the Lindblad equation does not generate reversible dynamics and exhibits markedly different properties from closed systems \cite{kessler2012dissipative, horstmann2013noise, benedict2018super, masson2022universality, asenjo2017exponential}. A physically important consequence of the irreversibility of Lindbladian dynamics is that it can be contractive, i.e., unlike with Hamiltonian dynamics, different initial states eventually converge to the same state and the system loses information about its initial state \cite{ruskai1994beyond}. This is well understood to be the case for finite-dimensional, time-independent Lindbladians with a unique fixed point $\sigma_\infty$: Starting from any initial state, the state at time $t$ converges to $\sigma_\infty$ as $t\to \infty$ \cite{wolf2012quantum}.

Contractivity has important physical implications: It limits quantum simulators and quantum memories \cite{aharonov1996limitations, alicki2009thermalization, temme2015fast, temme2017thermalization, mishra2024classically, de2023limitations, stilck2021limitations}, it is behind loss of quantum-advantage in sensing and interferometry \cite{abbasgholinejad2025theory, yang2023efficient, gammelmark2014fisher}, and is expected to constrain the controllability of open quantum systems \cite{wu2007controllability, kurniawan2012controllability, koch2016controlling, kallush2022controlling}. Consequently, considerable effort has gone into characterizing when Lindbladian dynamics is contractive and obtaining concrete bounds on the time scales of such contractivity. Most of this effort has been dedicated to time-independent Lindbladians, for which uniqueness of the fixed point can often be checked by examining the algebraic properties of the Hamiltonian and jump operators \cite{yoshida2024uniqueness, frigerio1977quantum, frigerio1978stationary, spohn1976approach, Spohn1977, evans1977irreducible} and the contraction time can be estimated from the spectral properties of the Lindbladian \cite{Temme2010,Kastoryano2013-ot, capel2018quantum, capel2020modified, muller2018sandwiched, temme2014hypercontractivity, montanaro2012some, beigi2020quantum, cubitt2015quantum,fang2025mixing}.

However, in many physically relevant settings, while the jump operators describing the system-environment interaction are static, the experimentalist can still tune the system Hamiltonian as a function of time. Such “driven Lindbladians” do not drive the system to a fixed steady state, but their dynamics may still be exponentially contractive. The extensive framework developed for time-independent Lindbladians cannot be applied to characterize contractive time-dependent Lindbladians, and other than results addressing the restricted settings of dissipators that are either unital and primitive \cite{aharonov1996limitations} or whose jumps generate a self-adjoint algebra \cite{Plenio2025}, the characterization of contractivity under time-dependent driving largely remains an open problem.

In this paper, we exhaustively address the problem of contractivity in time-dependent Lindbladians. We first present a number of “surprises” in which seemingly natural conjectures about the contractivity of driven Lindbladians turn out to be incorrect, thus illustrating the subtlety of the problem. We then show that adding any sufficiently small or slow driving Hamiltonian, to a contractive Lindbladian preserves exponential contractivity. We then completely characterize the dissipators which remain contractive irrespective of an added time-dependent Hamiltonian: We show that a dissipator is Hamiltonian-independently contractive if and only if there do not exist two non-trivial orthogonal subspaces that are invariant under all its jump operators.

\section*{Preliminaries}
Throughout this paper, we will consider time-dependent or driven Lindbladians on a finite-dimensional Hilbert space $\mathcal{H} = \mathbb{C}^d$, where the jump operators $\{L_1, L_2 \dots L_m\}$ are time-independent, but the Hamiltonian $H(t)$ depends on time. The corresponding master equation would then be
\begin{align}\label{eq:time_dep_lindblad}
\frac{d}{dt}\rho(t) = \mathcal{L}_t \rho(t) = -i[H(t), \rho(t)] + \sum_{i = 1}^m \mathcal{D}_{L_i}\rho(t),
\end{align}
where $\mathcal{D}_L(\cdot) = L \cdot L^\dagger - \{L^\dagger L, \cdot\}/2$ is the dissipator associated with $L$. Formally, given the state $\rho(s)$ at time $s$, the state at time $t$ can be expressed as \footnote{For the solution to the time-dependent Lindblad equation to exist, a sufficient condition is for $\mathcal{L}_t$ to be locally integrable, i.e. $\int_s^t  \mathcal{L}_{\tau} d \tau$ is well-defined for all finite intervals $[s,t] \subset \mathbb{R}$. We shall always consider this to be the case, hence we only consider driven Lindbladians whose time-dependent Hamiltonians $H(t)$ have components that are locally integrable. Importantly, when we make statements about some result concerning $\mathcal{L}_t$ holding \text{``independently"} of the Hamiltonian or holding for \text{``all"} Hamiltonians $H(t)$ that can appear in \cref{eq:time_dep_lindblad}, what we really mean is that the statement holds for all Hamiltonians with locally integrable components.}
\begin{align}\label{eq:lindblad_solution}
    \rho(t) = \mathcal{E}_{t, s}(\rho(s)) \text{  where  }\mathcal{E}_{t, s}=\mathcal{T}\exp\bigg(\int_s^t \mathcal{L}_\tau d\tau\bigg),
\end{align}
with $\mathcal{T}$ being the time-ordering operator \cite{Baake2011-xi, Hale2009-sf}. To avoid arbitrariness in the decomposition of the Lindbladian into Hamiltonian and dissipators, we assume that the jump operators are traceless, $\tr(L_j) = 0$.
If the jump operators are not traceless, then we can always obtain traceless jumps by subtracting a multiple of identity and appropriately adjusting the Hamiltonian \footnote{In particular, if $\tr(L_i) \neq 0$, then do the transformation
\begin{align*}
L_j  &\to L_j - \text{Tr}(L_j) \tau \ \text{  and} \\
H(t) &\to H(t) + \frac{i}{2d} \sum_j ( \tr (L_j)^{*} L_j - \tr (L_j) L_j^\dagger),
\end{align*}
where $\tau = I/d$ is the maximally mixed state. This transformation leaves the Lindbladian $\mathcal{L}_\tau$ unchanged while enforcing $\tr(L_j) = 0$.}.

We will analyze the contractivity under the time-evolution channels $\mathcal{E}_{t,s}$: The main results of this paper are summarized in \cref{tab:contractivity_summary}. First, we remark that since $\mathcal{E}_{t, s}$ is a completely-positive trace preserving map, it is automatically contractive in the trace norm i.e., for any two states $\rho, \sigma \in \text{D}_1(\mathcal{H})$, where $\text{D}_1(\mathcal{H})$ is the set of valid quantum states on $\mathcal{H}$, $\norm{\mathcal{E}_{t, s} (\rho - \sigma)}_1 \leq \norm{\rho - \sigma}_1$.
\begin{table*}[t]
\centering
\begingroup
\setlength{\tabcolsep}{8pt}
\renewcommand{\arraystretch}{1.2}
\newcommand{\settingwidth}{0.30\textwidth}
\newcommand{\resultwidth}{0.55\textwidth}
\newcommand{\summarycell}[2]{\parbox[c]{#1}{\raggedright #2\par}}
\newcommand{\settingcell}[1]{\summarycell{\settingwidth}{#1}}
\newcommand{\resultcell}[1]{\summarycell{\resultwidth}{#1}}
\newcommand{\referencecell}[1]{#1}
\begin{tabular}{@{}lll@{}}
\toprule
\settingcell{\textbf{Setting}}
& \resultcell{\textbf{Result}}
& \referencecell{\textbf{Reference}} \\
\midrule
\settingcell{Instantaneous gap}
& \resultcell{Lindbladian may not be exponentially contractive (EC)}
& \referencecell{\cref{thm:hyp_1_false}} \\
\settingcell{Instantaneous gap and slow driving}
& \resultcell{Lindbladian is exponentially contractive}
& \referencecell{\cref{thm:perturbation:slow H}}\\
\midrule
\settingcell{EC dissipator}
& \multicolumn{2}{l@{}}{%
\begin{tabular}[c]{@{}ll@{}}
\resultcell{Lindbladian may not be EC after adding a bounded Hamiltonian}
& \referencecell{\cref{thm:hyp_2_false}}\\[5pt]
\resultcell{Necessary and sufficient conditions on the dissipator for Hamiltonian-independent EC}
& \referencecell{\cref{thm:complete_characterization}} \\
\end{tabular}}\\
\midrule
\settingcell{EC dissipator with small Hamiltonian}
& \resultcell{Lindbladian is exponentially contractive}
& \referencecell{\cref{thm:perturbation:small H}}\\
\bottomrule
\end{tabular}
\endgroup
\caption{Summary of the results of this paper dealing with exponential contractivity for time-dependent Lindbladians.}
\label{tab:contractivity_summary}
\end{table*}

Our interest will be in a stronger notion of contractivity, stated precisely below, where the trace-norm distance between two states decreases exponentially with evolution time. 
\begin{dfn}[Exponentially contractive Lindbladian] \label{def:exp contraction}
	A time-dependent Lindbladian $\mathcal{L}_t$ on Hilbert space $\mathcal{H}$ is exponentially contractive if $\exists K, \gamma >0 : \forall \rho, \sigma \in \textnormal{D}_1(\mathcal{H})$ and $t \geq s \geq 0$,
	\begin{equation}
		\norm{\mathcal{E}_{t, s} (\rho - \sigma)}_1\leq K e^{-\gamma \abs{t - s}} \norm{\rho - \sigma}_1,
		\label{eq:contractivity}
	\end{equation}
	where $\mathcal{E}_{t, s} = \mathcal{T}\exp(\int_s^t \mathcal{L}_\tau d\tau)$ is the time-evolution channel for $\mathcal{L}_t$. The parameter $\gamma$ in \cref{eq:contractivity} will be called a contraction rate associated with $\mathcal{L}_t$.
\end{dfn}

When the Hamiltonian is time-independent i.e., $H(t) = H$, then the Lindbladian is exponentially contractive if and only if it has a unique fixed point.
As its contraction rate $\gamma$, we can pick any positive number smaller than its spectral gap
\begin{align}
    0<\gamma < \min_{\lambda \in \text{spec}(\mathcal{L})\setminus \{0\}} \text{Re}(-\lambda).
\end{align}
This fact follows straightforwardly from viewing the Lindblad master equation as a linear dynamical system and using its spectral decomposition \cite{Baumgartner2008-nb}. Generally, the uniqueness of the fixed-point can only be explicitly checked by solving the fixed-point equation $\mathcal{L}(\sigma_\infty) = 0$. Furthermore, whether the Lindbladian has a unique \emph{full rank} fixed-point (in which case they are called primitive or irreducible Lindbladians), can be checked with a simple algebraic condition: A Lindbladian is primitive (i.e., it has a unique and full-rank fixed point) if and only if the algebra generated by $\{L_1, L_2 \dots L_m, Q\}$, where $Q = -iH-\frac{1}{2}\sum_{i = 1}^m L_i^\dagger L_i$, under operator multiplication and addition is the full matrix algebra on $\mathcal{H}$ \cite{evans1977irreducible}. Physically, irreducible Lindbladians describe open systems where all possible noise-operators are eventually applied on the system. Simple examples of irreducible qubit Lindbladians include the depolarizing Lindbladian, which has jump operators $\{\sigma^x, \sigma^y, \sigma^z\}$ (that clearly generate the entire two-dimensional matrix algebra) or the Hamiltonian $H = \sigma^x$ and jump operator $L = \sigma^z$ both of which have the maximally mixed state $I/2$ as their unique fixed point. We re-emphasize that a Lindbladian can have a unique fixed-point without being irreducible: The qubit Lindbladian $\mathcal{L} = \mathcal{D}_{\ket{0}\!\bra{1}}$ is not irreducible while having a unique fixed-point $ \ket{0}\!\bra{0}$.

\section*{Contractivity with a uniform Lindbladian gap}
Since contractivity in the time-independent case reduces to checking if the Lindbladian has a unique fixed-point and thus a non-zero spectral gap, one might guess that a sufficient condition for exponential contractivity in the time-dependent case would be for the instantaneous Lindbladian $\mathcal{L}_t$, at every $t$, to have a unique fixed point and a spectral gap that remains uniformly non-zero.
\begin{hyp}\label{hyp:uniform_gap_contractivity}
    Given $\mathcal{D} = \sum_{i = 1}^m \mathcal{D}_{L_i}$ and consider the time-dependent Lindbladian $\mathcal{L}_t =-i[H(t), \cdot] + \mathcal{D}$ such that $\forall t, \mathcal{L}_t$ has a unique fixed point and $\exists \kappa > 0:$
	\begin{equation}
		\forall t,  \min_{\lambda \in \textnormal{spec}(\mathcal{L}_t)\setminus \{0\}} \textnormal{Re}(-\lambda) \geq \kappa,
		\label{eq:hyp2}
	\end{equation}
    then $\mathcal{L}_t$ generates exponentially contractive dynamics.
\end{hyp}
\noindent Our first result is to show, with an explicit counter-example, that this hypothesis is false even under more stringent restrictions on $\mathcal{D}$ and $H(t)$.
\begin{thm}[Uniform gap does not imply contractivity]\label{thm:hyp_1_false}
     \cref{hyp:uniform_gap_contractivity} is false. Furthermore, it remains false even when $H(t)$ is periodic with a single harmonic i.e., $H(t) = H_0 + H_1 \cos \omega t+H_2 \sin \omega t$ for some $H_0, H_1, H_2, \omega$, and when $\mathcal{L}_t$ (for any $t$), $\mathcal{D}$ and the time-averaged Lindbladian $\bar{\mathcal{L}}= -i[H_0, \cdot] +\mathcal{D}$ are irreducible.
\end{thm}
We provide a counter-example and its analysis proving \cref{thm:hyp_1_false} in \cref{app:subsec:counterexample hyp2}. Here, to convey the key idea, we present a simpler counter-example which satisfies the conditions of \cref{hyp:uniform_gap_contractivity} but without satisfying the irreducibility constraints in \cref{thm:hyp_1_false}. Consider a qutrit with computational basis states labelled by $\ket{0}, \ket{1}, \ket{e}$ evolving with the time-dependent Lindbladian
\begin{align}
    \mathcal{L}_t = -i[H(t),  \cdot \, ] + \mathcal{D}_{\ket{e}\!\bra{e}},
\end{align}
with
\begin{align}
    H(t) = e^{i\omega Gt}(hG +  V )e^{-i\omega Gt},
\end{align}
where $h, \omega \in \mathbb{R}$, $G =\ket{1}\!\bra{0}+ \text{h.c.}$ and $V=\ket{e}\!\bra{1}+\text{h.c.}$. Since $G$ has eigenvalues $\pm 1$, $H(t)$ is periodic with a single harmonic. The jump operator in $\mathcal{L}_t$ simply dephases the coherences between the subspaces $\text{span}(\{\ket{e}\})$ and $\text{span}(\{\ket{0}, \ket{1}\})$, while the Hamiltonian $H(t)$ coherently couples the two subspaces.

It can be seen easily that $\mathcal{L}_0$ has the maximally mixed state $I/3$ as its unique fixed point. Furthermore, the fixed point and the spectral gap of $\mathcal{L}_t$ is independent of $t$. This follows by noting that $\mathcal{L}_t$ is unitarily related to $\mathcal{L}_{t=0}$, i.e., $\mathcal{L}_t = \mathcal{U}_t \mathcal{L}_{t =0} \mathcal{U}_t^\dagger$ where $\mathcal{U}_t(\cdot) =e^{i\omega Gt}(\cdot)e^{-i\omega Gt}$, and this unitary conjugation preserves the spectrum and the identity fix point. Therefore, $\mathcal{L}_t$ satisfies the conditions of \cref{hyp:uniform_gap_contractivity}. 

Nevertheless, the dynamics generated by $\mathcal{L}_t$ is not contractive: This can be most easily seen using a rotating frame transformation. Suppose $\rho(t)$ satisfies the master equation with the Lindbladian $\mathcal{L}_t$, then $\tilde{\rho}(t)=e^{-i\omega Gt} \rho(t)e^{i\omega Gt}$ satisfies 
\begin{align}
    \frac{d}{dt}\tilde{\rho}(t)=-i[(h-\omega) G+V, \tilde{\rho}(t)] + \mathcal{D}_{\ket{e}\!\bra{e}}\tilde{\rho}(t),
\end{align}
with the dissipator remaining unchanged since $[G, \ket{e}\!\bra{e}]=0$. Now, if $\omega=h$, then it can be seen that if $\tilde{\rho}(0)=\ket{0}\!\bra{0}$, then $\tilde{\rho}(t)=\ket{0}\!\bra{0}$ for all $t$ while if $\tilde{\rho}(0)$ is supported on $\text{span}(\{\ket{1}, \ket{e}\})$ then $\tilde{\rho}(t)\to (\ket{1}\!\bra{1}+\ket{e}\!\bra{e})/2$ as $t\to \infty$, thereby showing that $\mathcal{L}_t$ yields dynamics that is not exponentially contractive.

It can still be asked if there are conditions when \cref{hyp:uniform_gap_contractivity} can still be true. To address this question, we show that if the time-dependent Hamiltonian is sufficiently slowly varying, then contractivity of $\mathcal{L}_t$ at every time-instant is indeed sufficient to guarantee contractivity of the time-dependent master equation. A proof of this result is provided in \cref{app:small_slow_results}.
\begin{thm}[Contractivity for sufficiently slow Hamiltonians] \label{thm:perturbation:slow H}
Consider the driven Lindbladian $\mathcal{L}_t = - i[H(t), \ \cdot \ ] + \mathcal{L}_0$ and suppose that for all $t$, $\mathcal{L}_t$ is exponentially contractive with constant $K_0$ and contraction rate $\gamma_0$  (\cref{def:exp contraction}).
Then, the original time-dependent Lindbladian $\mathcal{L}_t$ will be contractive provided
\begin{align} \label{eq:bounded dH dt condition}
\sup_{t \geq 0} \left\| \frac{d}{dt} H(t) \right\|_\infty < \frac{2\gamma_0^2 }{3} \frac{1}{1 + \frac{2}{3} \ln(K_0) + \frac{1}{3} \ln(K_0)^2}.
\end{align}
\end{thm}

\section*{Hamiltonian-independent contractivity}
Consider now the following question: Suppose the Lindbladian in \cref{eq:time_dep_lindblad} is exponentially contractive in the absence of a Hamiltonian ($H(t)=0$).
Then, is the dynamics still contractive for an arbitrary Hamiltonian $H(t)$? Physically, one might expect this to indeed be the case: Exponential contractivity in the absence of the Hamiltonian suggests that the dissipators acting on the system eventually remove information about the initial system state. Application of a Hamiltonian cannot remove or cancel these dissipators: While it might decrease the contractivity rate, the system dynamics should eventually become independent of its initial state. This reasoning then suggests the following Hypothesis.
\begin{hyp}\label{hyp:contractivity_preservation}
    If the Lindbladian $\mathcal{D} = \sum_{i = 1}^m \mathcal{D}_{L_i}$ is exponentially contractive, then so is $\mathcal{L}_t = -i[H(t), \cdot]+\mathcal{D}$ for any time-dependent Hamiltonian $H(t)$ which is bounded (i.e. $\sup_{t\geq 0}\norm{H(t)} < \infty$).
\end{hyp}
\noindent There are several examples of known dissipators where this hypothesis is indeed true: Most prominently, this is true if the dissipators are unital and irreducible \cite{aharonov1996limitations}, for instance in the case of depolarizing noise acting on a set of qubits. Other examples include dissipators whose jump operators span the space of traceless operators \cite{abbasgholinejad2025theory}. Furthermore, this hypothesis always holds for the analogous problem for classical Markov processes: Consider the case where the density matrix $\rho(t)$ remains diagonal at all times i.e., $\rho(t) = \text{diag}(p(t))$ for some time-dependent probability distribution $p(t)$, and the time-independent dissipator $\mathcal{D}$ generates  contractive Markov process described by a rate matrix $A$ i.e., $\mathcal{D}\rho(t) = \text{diag}(Ap(t))$. Then, for any bounded time-dependent rate matrix $\tilde{A}(t)$, $A + \tilde{A}(t)$ still generates a contractive Markov process.

Surprisingly, this hypothesis fails to be true in general for quantum dynamics, as we show with an explicit counterexample in the next result.
\begin{thm}[Bounded Hamiltonians can halt dissipative contractivity]\label{thm:hyp_2_false}
     \cref{hyp:contractivity_preservation} is false. Furthermore, it remains false even when $\mathcal{D}$ is irreducible.
\end{thm}
\noindent The failure of \cref{hyp:contractivity_preservation} can be seen again for a qutrit (in fact, due to next result we present, the hypothesis is true for qubits). Label the computational basis states for the qutrit by $\ket{0}, \ket{1}, \ket{e}$ (in order) and consider $\mathcal{D} = \mathcal{D}_L$, where
\begin{align}\label{eq:main_qutrit_jump}
    \mathcal{L}_0 =\mathcal{D}_L \quad \text{where}\quad L = \begin{bmatrix}
        \ell_0 & 0 & \beta_0 \\
        0 & \ell_1 & \beta_1 \\
        0 & 0 & -(\ell_0 + \ell_1)
    \end{bmatrix},
\end{align}
where we will assume $\ell_0, \ell_1 \neq 0, \ell_0 \neq \pm \ell_1$ and $\beta_0, \beta_1 \neq 0$. It can be shown by analyzing the algebra generated by $L$ and $L^\dagger L$ that $\mathcal{D}$ is irreducible (see \cref{app:subsec:counterexample hyp5}) and thus generates exponentially contractive dynamics. Now, consider applying the time-independent Hamiltonian $H$
\begin{align}\label{eq:main_qutrit_hamil}
    H = \frac{1}{2} \begin{bmatrix}
        0 & 0 & -i\ell_0^* \beta_0 \\
        0 & 0 & -i\ell_1^* \beta_1 \\
        i\ell_0 \beta_0^* & i\ell_1\beta_1^* & 0  
    \end{bmatrix}.
\end{align}
The application of this Hamiltonian results in the full Lindbladian being reducible: This is easily appreciated by noting that $L$ is, by design, upper-triangular. Furthermore, $Q = -iH-\frac{1}{2}L^\dagger L$ is given by
\begin{align}
    Q= -\frac{1}{2}\begin{bmatrix}
            \abs{\ell_0}^2 & 0 & 2\ell_0^* \beta_0 \\
            0 & \abs{\ell_1}^2 & 2\ell_1^*\beta_1 \\
            0 & 0 & \abs{\ell_0+\ell_1}^2 + \abs{\beta_0}^2 + \abs{\beta_1}^2 
        \end{bmatrix},
\end{align}
and is also upper-triangular even though $L^\dagger L$ is not. Thus, the algebra generated by $L$ and $Q$ is only the space of upper-triangular matrices which is strictly smaller than the full space of qutrit operators. Furthermore, one can also easily see that for both $b \in \{0, 1\}$
\begin{align}
    L\ket{b} = \ell_b \ket{b} \quad \text{and}\quad Q \ket{b} =-\frac{\abs{\ell_b}^2}{2}\ket{b},
\end{align}
and thus both $\ket{0}$ and $\ket{1}$ are fixed points of the full Lindbladian (see \cref{app:subsec:counterexample hyp5}). This shows that, in the presence of the Hamiltonian, we indeed have a failure of exponential contractivity even though the jump-only dissipator is irreducible.

Next, we consider settings where the jump-only dissipator is in fact contractive in the presence of any (possibly time-dependent) Hamiltonian. Our first result provides a necessary and sufficient algebraic condition that completely characterizes this problem.
\begin{thm}[Complete characterization of Hamiltonian-independent contractivity] \label{thm:complete_characterization}
Consider a dissipator $\mathcal{D} = \sum_{j = 1}^m \mathcal{D}_{L_j}$. Then, $-i[H(t), \cdot] + \mathcal{D}$ is exponentially contractive for any bounded $H(t)$, with a contraction rate independent of $H(t)$, if and only if there do not exist two non-trivial orthogonal subspaces $V_1 \perp V_2 \subset \mathcal{H}$ that are left invariant under the jump operators $L_1,L_2 \dots L_m$, i.e. for which we have $L_j V_1 \subseteq V_1$ and $L_j V_2 \subseteq V_2$ for all $j$. 
\end{thm}
\noindent A complete proof of this theorem is provided in \cref{app:hamil_ind_contr}. This theorem implies several important sufficient conditions for Hamiltonian-independent contractivity. First and foremost, it follows from  \cref{thm:complete_characterization} as well as the Burnside theorem that a set of jump operators $L_1, L_2 \dots L_m$ is Hamiltonian-independently contractive if the algebra generated by the jump operators is the full matrix algebra. Note that such a dissipator is irreducible but not every irreducible dissipator satisfies this condition. In fact, this shows that for irreducible dissipators, exponential contractivity can only be removed, by addition of a Hamiltonian, if the operator $Q = -\frac{1}{2}\sum_{i = 1}^m L_i^\dagger L_i$ is important in completing the algebra: This is indeed the case in the counterexample to \cref{hyp:contractivity_preservation} presented above.

Furthermore, \cref{thm:complete_characterization} also shows that several reducible dissipators can be Hamiltonian-independently contractive. In particular, for the case of qubits, this theorem implies that exponential contractivity of the dissipator also implies Hamiltonian-independent contractivity. Furthermore, an important class of reducible dissipators are those with lowering jump operators i.e.,  $L = \sum_{k = 1}^{d - 1}\sqrt{\gamma_k}\ket{k - 1}\!\bra{k} $ on the Hilbert space $\mathbb{C}^d$. These dissipators arise in modeling emission processes in atomic systems \cite{abbasgholinejad2025theory, asenjo2017exponential}, and \cref{thm:complete_characterization} implies that if $\gamma_k \neq 0$, then these are also Hamiltonian-independently contractive. In \cref{app:hamil_ind_contr} and \cref{app:Hamiltonian indpendent rates}, we provide a detailed analysis of these specific settings, and also provide estimates for the corresponding contraction rates.

Finally, we also show that \cref{hyp:contractivity_preservation} is indeed true, without any further constraints on the dissipator $\mathcal{D}$, if the added Hamiltonian is sufficiently small (in comparison to the contractivity rate of the dissipator $\mathcal{D}$). 
\begin{thm}[Contractivity for sufficiently small Hamiltonian] \label{thm:perturbation:small H}
Consider the driven Lindbladian $\mathcal{L}_t = - i[H(t), \ \cdot \ ] + \mathcal{D}$ where  $\mathcal{D}$ is exponentially contractive with constant $K$ and contraction rate $\gamma$ as defined in \cref{def:exp contraction}.
Then, the full Lindbladian $\mathcal{L}_t$ will also be exponentially contractive provided
\begin{align} \label{eq:bounded H condition}
    \sup_{t \geq 0} \left\| H (t) \right\|_\infty < \frac{\gamma}{2 + 2\ln(K)}.
\end{align}
\end{thm}
\noindent A proof of this theorem is provided in \cref{app:small_slow_results}.

\section*{Conclusion and Outlook}
In \cref{thm:complete_characterization}, we gave a complete algebraic characterization of Hamiltonian-independent contractivity: A dissipator remains exponentially contractive under arbitrary Hamiltonian driving if and only if no two non-trivial orthogonal subspaces are invariant under all its jump operators. From a mathematical perspective, it is surprising that this notion admits necessary and sufficient algebraic conditions, given the difficulty of obtaining such conditions even for time-independent Lindbladians. It remains an open problem to determine optimal Hamiltonian-independent contraction rates. Although the proof of \cref{thm:complete_characterization} provides universal rates, these decrease very rapidly with the Hilbert space dimension. We can obtain stronger bounds for restricted classes, including unital dissipators, but obtaining comparable bounds more generally remains challenging. More sophisticated estimates may improve the rates in our proof, although avoiding their poor dimensional scaling likely requires a substantially different argument.

Furthermore, we also showed that, for a given time-dependent drive,  instantaneous contractivity and a uniform spectral gap do not guarantee exponential contraction. However, we established sufficient conditions for sufficiently weak or slow driving. It would be interesting to obtain sharper necessary and sufficient conditions in this setting, accounting for the interplay between the driving Hamiltonian and the dissipator.

\section*{Acknowledgments}

The authors acknowledge the use of OpenAI’s language model GPT-5 (Sol) for generating candidate counterexamples based on suggestions from the author. These counterexamples were then checked, simplified and rigorously analyzed by the authors. The authors also acknowledge that the proof of \cref{thm:complete_characterization} was partly inspired by a proof of a conjecture found in an earlier version of this paper, which was generated by GPT-5 (Sol). However, no language model was used in coming up with the statement of the theorem or in the writing of its proof. All the other parts of the paper were done by the authors. The authors take full responsibility for the correctness of the analysis presented in the paper.

L.H.W.\ and D.M.\ acknowledge financial support from VILLUM FONDEN via the QMATH Centre of Excellence (Grant No.10059). D.M.\ acknowledges support from the Novo Nordisk Foundation under grant numbers NNF22OC0071934 and NNF20OC0059939. R.T.\ acknowledges funding from the European Union's Horizon Europe research and innovation program under grant agreement number 101221560 (ToNQS).
This work was supported by NCCR Precision, a National Centre of Competence in Research funded by the
Swiss National Science Foundation grant 51AU-0\_229299. D.M.\ would like to thank
the Swiss National Science Foundation for their financial support.

\let\section\standardsection
\onecolumngrid

\clearpage

\begin{centering}
\section*{Appendices}
\end{centering}

\begin{appendices}
\appendix
\section{Preliminaries}
\subsection{Notation} \label{subsec:notation}
On a finite-dimensional Hilbert space $\mathcal{H}$, we will denote by $\text{B}(\mathcal{H})$ the set of bounded linear operators from $\mathcal{H} \to \mathcal{H}$ and by $\text{D}_1(\mathcal{H})$ the set of trace one positive-semi-definite operators in $\text{B}(\mathcal{H})$. Given a set of operators $\mathcal{A}=\{A_1, A_2 \dots A_m\}$, denote the algebra generated by $\mathcal{A}$ under operator multiplication and linear combinations by $\text{Alg}(\mathcal{A})$.

For a Hermitian operator $X$, and a function $f:\mathbb{R}\to \mathbb{R}$, we will denote by $f(X)$ the Hermitian operator $f(X) = \sum_i f(\lambda_i) \ket{e_i}\!\bra{e_i}$, where $\lambda_i, \ket{e_i}$ is the $i^\text{th}$ eigenvalue-eigenvector pair of $X$. Of particular interest is the \emph{sign function} $\sgn(X)$ of a Hermitian operator $X$, where the ordinary $\mathbb{R} \rightarrow \mathbb{R}$ sign function is defined as $\sgn(x) \coloneqq x / |x|$ if $x \neq 0$ and $\sgn(x)=0$ if $x=0$.
Furthermore, for an operator $A \in \mathcal{B}(\mathcal{H})$, $\norm{A}_p$ will denote the usual Schatten norms of $A$ i.e., $\norm{A}_p = (\sum_i \abs{s_i}^p)^{1/p}$ where $s_i$ is the $i^\text{th}$ singular value of $A$. We will default to using $\norm{A} = \norm{A}_\infty$.

For a general linear operator $A \in \mathcal{B}(\mathcal{H})$ and subspace $\mathcal{S} \subseteq \mathcal{H}$ let $A |_{\mathcal{S}}$ denote the projection of $A$ onto $\mathcal{S}$, i.e. $A |_{\mathcal{S}} \coloneqq P_{\mathcal{S}} A P_{\mathcal{S}}$ with $P_{\mathcal{S}}$ being a projector onto $\mathcal S$.
Furthermore, we will occasionally use the ($ 1 \rightarrow 1$) \emph{super-operator norm} $\superopnorm{ \ \cdot \ }$, which for a linear superoperator $T : \mathcal{B}(\mathcal{H}) \rightarrow \mathcal{B}(\mathcal{H})$ is defined as
\begin{align} \label{def:super 1 1 norm}
\superopnorm{T} \ \coloneqq \ \sup_{x \in \mathcal{B}(\mathcal{H}) , x \neq 0} \frac{\norm{T(x)}_1}{\norm{x}_1}.
\end{align}
Finally, for the purposes of proving and further discussing \cref{thm:complete_characterization}, it will be useful to make the following definition.
\begin{dfn}[Hamiltonian-independent contractivity] \label{def:Hamiltonian independent contractivity}
A dissipator $\mathcal{D}$ is called Hamiltonian-independently contractive if there exist universal constants $K , \gamma > 0$ such that for any time-dependent Hamiltonian $H(t)$, the Lindbladian $\mathcal{L}_t = - i [H(t) , \cdot ] + \mathcal{D}$ is exponentially contractive with constant $K$ and rate $\gamma$ as defined in \cref{def:exp contraction}.
\end{dfn}

\subsection{Analysis results for continuous and right-differentiable functions}\label{app:subsec:analysis result long}
In this appendix, for completeness, we provide proofs of some important general results that are used throughout this paper.

In this section, we shall prove \cref{lem:full analysis result} and \cref{lem:fund thm calculus right diff functions}, both stated below, which are crucial for rigorously proving many of our main results. These results essentially state that under certain conditions we are allowed to work with the right derivative of $\norm{x(t)}_1$ instead of the ordinary derivative for the purposes of proving exponential decay bounds and allowed to use the fundamental theorem of calculus, even when the ordinary derivative is not well-defined. Such results, or something similar, have likely already been established, but we provide rigorous proofs here for completeness. The right derivative of a function $f : [a,b] \rightarrow \mathbb{R}$ at a point $x \in [a,b)$ is denoted $\partial_+ f (x)$ and it is (supposing it is well-defined) defined as $\partial_+ f(x ) \coloneqq \lim_{y \rightarrow x \ : \ y \in (x,b] } \frac{f(y)-f(x)}{y-x}$, where $y \rightarrow x \ : \ y \in (x,b]$ denotes a sequence converging to $x$ that is confined to the interval $(x,b]$. Recall also that we say that $f : [a,b] \rightarrow \mathbb{R}$ is (pointwise) right-differentiable, if $\partial_+ f (x)$ is well-defined for all $x \in [a,b)$, even though $\partial_+ f (x)$ may not be continuous.

\begin{lem} \label{lem:full analysis result}
Let $f : [a,b] \rightarrow \mathbb{R}$ be a continuous, right-differentiable, real-valued, non-negative function on the finite interval $[a,b] \subset \mathbb{R}$. Suppose further that $\partial_+ f$ is upper bounded by $\partial_+ f(x) \leq - K f(x)$ for all $x \in [a,b)$, where $K \in \mathbb{R}$, $K \geq 0$ is some positive constant. Then, for all $ x , y \in [a,b]$ with $y \geq x$, we have
\begin{align*}
f(y) \leq f(x) e^{ - K (y-x) } .
\end{align*}
\end{lem}

Note that the condition from \cref{lem:full analysis result} that $f$ should be continuous is crucial, as the step function defined by $f(x)=0$ if $x < 0$ and $f(x)=1$ if $x \geq 0$ is everywhere right-differentiable and has non-positive right derivative everywhere, but it is strictly increasing in any interval containing $0$.

\begin{lem} \label{lem:fund thm calculus right diff functions}
Let $f : [a,b] \rightarrow \mathbb{R}$ be a continuous, right-differentiable, real-valued function on the finite interval $[a,b]$. Suppose that $\partial_+ f(x)$ is upper bounded by some continuous function $B$ on $[a,b)$, i.e. for all $s \in [a,b)$, $\partial_+ f(s) \leq B(s)$. Then, for any $  x, y \in (a,b)$ with $x \leq y$, we have
\begin{align*}
f(y) \leq f(x) + \int_x^y dt \ B(t) .
\end{align*}
\end{lem}

The proofs of \cref{lem:full analysis result} and \cref{lem:fund thm calculus right diff functions} are given below (page \pageref{proof:of full analysis result} and \pageref{proof:of fundumentla thm calc result} respectively). Both proofs however rely on several intermediary results, which we shall first establish. The following result can be seen as a variant of \emph{Rolle's theorem} (see e.g. \cite[pp. 260-266]{Anton1998} or \cite[p. 184]{Apostol1967}), just for right-differentiable functions instead of ordinary differentiable functions.

\begin{prp} \label{prop:rolle generalization}
Let $g : [a,b] \rightarrow \mathbb{R}$ be a continuous, right-differentiable, real-valued function on the finite interval $[a,b] \subset \mathbb{R}$ satisfying $f(a) = f(b)$. Then, there will exist at least one point $y \in [a,b)$ for which $\partial_+ g(y) \leq 0$.
\end{prp}

\begin{proof}
Since $g : [a,b] \rightarrow \mathbb{R}$ by assumption is continuous, and since $[a,b]$ is a finite interval and hence a compact set, the maximum $\max_{x \in [a,b]} g(x)$ exists and is achieved at some $x_{max} \in [a,b]$, i.e.  $\max_{x \in [a,b]} g(x) = g(x_{max})$. Note now that we can w.l.o.g. assume $x_{max} \in [a,b)$, since if $x_{max} = b$, we have by the assumption $f(a) = f(b)$ that $\max_{x \in [a,b]} g(x) =  g(a) = g(b)$, hence $a$ also achieves the maximum. Again by an assumption of the proposition, $\partial_+ g(x_{max})$ is well-defined (since now $x_{max} \in [a,b)$). Now, we must have $\partial_+ g(x_{max}) \leq 0$, since by assumption of $g(x_{max})$ being a maximum, we have $g(y)-g(x_{max}) \leq 0$ for all $y \geq x_{max}$, which gives us 
\begin{align*}
\partial_+ g(x_{max} ) = \lim_{y \rightarrow x \ : \ y \in (x_{max},b] } \frac{f(y)-f(x_{max})}{y-x_{max}} \leq \lim_{y \rightarrow x \ : \ y \in (x_{max},b] } \frac{0}{y-x_{max}} \leq 0 ,
\end{align*}
where we have used that the limit must converge by assumption. This proves the proposition. 
\end{proof}

\cref{prop:rolle generalization} above allows us to prove the following further result. 

\begin{prp} \label{prop:right derivative negative implies decreasing}
Let $h : [a,b] \rightarrow \mathbb{R}$ be a continuous, right-differentiable, real-valued function on the finite interval $[a,b]$. Suppose that $\partial_+ h(x)$ is non-positive for all $x \in [a,b)$. Then $h$ must be a non-increasing function, i.e. for all $a \leq x \leq y \leq b$, we have $h(x) \geq h(y)$. 
\end{prp}

\begin{proof}
Consider any $c,d \in [a,b]$ such that $c < d$, and suppose for contradiction that we had $h(c) < h(d)$. Define now the new function $g : [a,b] \rightarrow \mathbb{R}$ by
\begin{align} \label{eq:app:def of G function}
g(x) = h(d)-h(x)+\frac{h(d)-h(c)}{d-c} (x-d) .
\end{align}
It follows from $h(x)$ being continuous, real-valued and right-differentiable, that $g(x)$ is also continuous, real-valued and right-differentiable. Note further that we have $g(c)=g(d) = 0$, which means that $g$ satisfies all the conditions from \cref{prop:rolle generalization}. Thus, by \cref{prop:rolle generalization}, there exists a point $y \in [c,d) \subseteq [a,b)$ at which we have $\partial_+ g(y) \leq 0$. But by also computing $\partial_+ g(y)$ from the definition \cref{eq:app:def of G function} of $g$, we get
\begin{equation} \label{eq:app:step prop decreasing}
\begin{aligned}
0 \geq \partial_+ g(y) = \partial_+ \left( h(d)-h(y)+\frac{h(d)-h(c)}{d-c} (y-d)\right) = - \partial_+ h(y) + \frac{h(d)-h(c)}{d-c}  \\ 
\Rightarrow \partial_+ h (y) \geq \frac{h(d)-h(c)}{d-c} > 0 ,
\end{aligned}
\end{equation}
where we have lastly used $h(d) > h(c)$ which we assumed for contradiction, and we have also used that $\partial_+$ must be linear and must by definition act like the normal derivative on differentiable functions. However, \cref{eq:app:step prop decreasing} implies $\partial_+ h(y) > 0$ for $y \in [a,b)$, which contradicts the assumption from the proposition. 
\end{proof}

We are now finally ready to prove \cref{lem:full analysis result}

\begin{proof}[Proof of \cref{lem:full analysis result}] \label{proof:of full analysis result}

Let $f : [a,b] \rightarrow \mathbb{R}$ be a continuous, right-differentiable, real-valued, non-negative function on the finite interval $[a,b] \subset \mathbb{R}$, and suppose that its right derivative is bounded by $\partial_+ f(x) \leq - K f(x)$ for all $x \in [a,b)$, with $K \geq 0$ being some non-negative constant. Consider now the new function $F(x) \coloneqq f(x) e^{Kx}$ for any $x \in [a,b]$. $F(x)$ is clearly real-valued, continuous and non-negative since it is a product of functions with these properties. It is also not difficult to prove that $\partial_+$ obeys the product rule in a similar way to the ordinary derivative, i.e. if $g,h$ are right-differentiable at $x$, then so is $gh$, and we specifically have $\partial_+ (gh)(x) = g(x) \partial_+ h(x) + h(x) \partial_+ g(x)$. This establishes that $F(x) = f(x) e^{Kx}$ is also right-differentiable, and we further have for any $x \in [a,b)$
\begin{align} \label{eq:app:proof final analysis s1}
\partial_+ F(x) = \partial_+ \left( f(x) e^{Kx} \right) = f(x) \partial_+ e^{Kx} + e^{Kx} \partial_+ f(x) = \left( K f(x) + \partial_+ f(x) \right) e^{Kx} .
\end{align} 
We now use the assumption $\partial_+ f(x) \leq - K f(x)$ in \cref{eq:app:proof final analysis s1} above, which entails $\partial_+ F(x) \leq 0$ for any $x \in [a,b)$. This, along with the earlier established properties of $F(x)$ shows that $F(x)$ satisfies all the conditions of \cref{prop:right derivative negative implies decreasing}, which then proves that $F(x)$ is a non-increasing function on $[a,b]$. Hence, for any $x,y \in [a,b]$ with $x \leq y$, we have 
\begin{align*}
F(x) \geq F(y) \ \Rightarrow \ f(x) e^{Kx} \geq f(y) e^{K y} \ \Rightarrow \ f(y) \leq f(x) e^{- K(y-x)} .
\end{align*}
\end{proof}

\cref{prop:right derivative negative implies decreasing} also allows us to prove \cref{lem:fund thm calculus right diff functions}.

\begin{proof}[Proof of \cref{lem:fund thm calculus right diff functions}] \label{proof:of fundumentla thm calc result}
Consider the function $h : [a,b] \rightarrow \mathbb{R}$ given by $h(t) := f(t) - \int_x^t dt \ B(t)$ for all $t \in [a,b]$. Since $B$ is by assumption continuous, we get by the (ordinary) fundamental theorem of calculus that for any $x\in[a,b]$, $\int_x^t dt \ B(t)$ is differentiable, and
\begin{align} \label{eq:int diff fund}
\partial_{t ,+} \int_x^t dt' \ B(t') = \frac{d}{dt} \int_x^t dt' \ B(t') = B(t)  .
\end{align}
Since $\partial_+ f(t)$ was by assumption everywhere defined on $[a,b]$ and upper bounded by $\partial_+ f(t) \leq B(t)$, we get by \cref{eq:int diff fund} that $\partial_+ h(t)$ is everywhere defined on $[a,b]$ and equals $\partial_+ h(t) = \partial_+ f(t) - B(t) \leq 0$. Further, since $f$ and $B$ were by assumption continuous, $h$ will also be continuous. Hence, $h : [a,b] \rightarrow \mathbb{R}$ is a continuous, real, right-differentiable function whose right derivative is non-positive, which by \cref{prop:right derivative negative implies decreasing} implies that $h$ is non-increasing in the relevant interval. Thus, for any $a \leq x \leq y \leq b$ we have
\begin{align*}
h(y) \leq h(x) \ \Rightarrow \ f(y) - \int_x^y dt \ B(t) \leq f(x) - \int_x^x dt B(t) \ \Rightarrow \ f(y) \leq f(x) + \int_x^y dt \ B(t) .
\end{align*}
\end{proof}
\section{Counter-examples}\label{app:counter_examples}
\subsection{\texorpdfstring{\cref{hyp:contractivity_preservation}}{hyp:contractivity preservation}} \label{app:subsec:counterexample hyp5}
We construct a Lindbladian on a qutrit: $\mathcal{H}=\mathbb{C}^3$. We will denote the computational basis for the qutrit via $\ket{0}, \ket{1}, \ket{e}$ (in order), and consider the Lindbladian with a single jump operator:
\begin{align}\label{eq:qutrit_jump}
    \mathcal{D} =\mathcal{D}_L \quad \text{where}\quad L = \begin{bmatrix}
        \ell_0 & 0 & \beta_0 \\
        0 & \ell_1 & \beta_1 \\
        0 & 0 & -(\ell_0 + \ell_1)
    \end{bmatrix},
\end{align}
with $\ell_0, \ell_1, \ell_0-\ell_1, \ell_0+\ell_1\neq 0$ and $\beta_0, \beta_1\neq 0$. We first show that $\mathcal{D}$ is irreducible.
\begin{lem}\label{lemma:irreducibility_qutrit_jump}
    $\mathcal{D}$ from \cref{eq:qutrit_jump} is an irreducible Lindbladian.
\end{lem}
\begin{proof}
To show that $\mathcal{D}$ is irreducible, it suffices to show that $\text{Alg}(\{L, L^\dagger L\})$ is the full matrix algebra. Note that since $\ell_0, \ell_1, -(\ell_0 + \ell_1) \neq 0$, $L$ is invertible: It then follows by the Cayley Hamilton Theorem, $L^{-1}$ can be expressed as a polynomial of $L$ and thus $\text{Alg}(\{L, L^\dagger L\}) = \text{Alg}(\{L, L^\dagger\})$.

Next, we show that $\text{Alg}(\{L, L^\dagger \}) $ is the full matrix algebra. For that, it is sufficient to show that $L, L^\dagger$ have no simultaneous non-zero invariant subspace. Assume that $\mathcal{S}\neq \{0\}$ is a simultaneous invariant subspace of $L$ and $L^\dagger$. Since $L \mathcal{S}\subseteq \mathcal{S}$ and $L^{-1}$ is a polynomial of $L$, $L^{-1}\mathcal{S}\subseteq \mathcal{S} \implies \mathcal{S}\subseteq L \mathcal{S}$ and therefore $L \mathcal{S}=\mathcal{S}$. A similar argument shows that $L^\dagger \mathcal{S} = \mathcal{S}$. Finally, from $L \mathcal{S}=\mathcal{S}$ and $L^\dagger \mathcal{S} = \mathcal{S}$ together with the invertibility of $L$, it then also follows that $L \mathcal{S}^\perp= \mathcal{S}^\perp$ and $L^\dagger \mathcal{S}^\perp = \mathcal{S}^\perp$.

Now, since $\text{dim}(\mathcal{H})=3$, if $\mathcal{S}$ is non-trivial (i.e., $\mathcal{S} \neq \{ 0 \} , \mathcal{H}$), then either $(\dim(\mathcal{S}), \dim(\mathcal{S}^\perp))=(2, 1)$ or $(\dim(\mathcal{S}), \dim(\mathcal{S}^\perp))=(1, 2)$. Therefore, $L, L^\dagger$ will have a one-dimensional invariant subspace: Suppose this subspace is $\text{span}(\{\ket{\phi}\})$, then 
\begin{align}
    \exists \nu \in \mathbb{C}: L\ket{\phi}=\nu \ket{\phi}\quad \text{and}\quad L^\dagger \ket{\phi} =\nu^*\ket{\phi}.
\end{align}
Expressing $\ket{\phi}=a_0\ket{0}+a_1\ket{1}+a_e\ket{e}$, the eigenvalue equations yield
\begin{subequations}
    \begin{align}
        &(\ell_0-\nu)a_0+\beta_0a_e=0,\\
        &(\ell_1-\nu)a_1+\beta_1a_e=0,\\
        &(\ell_0-\nu)^*a_0=0,\\
        &(\ell_1-\nu)^*a_1=0, \\
        &\beta_0^* a_0+\beta_1^* a_1 - (\ell_0^*+\ell_1^*+\nu^*) a_e= 0.
    \end{align}
\end{subequations}
Since $\ell_0\neq\ell_1$ and $\beta_0, \beta_1 \neq 0$, these equations force $\ket{\phi}= 0$, thus contradicting the assumption that $\mathcal{S}\neq \{ 0 \} , \mathcal{H}$. Therefore, we obtain that $\text{Alg}(\{L, L^\dagger L\}) = \text{Alg}(\{L, L^\dagger\}) $ is the full matrix algebra.
\end{proof}
Now, we will define the Hamiltonian $H$
\begin{align}\label{eq:qutrit_hamil}
    H = \frac{1}{2} \begin{bmatrix}
        0 & 0 & -i\ell_0^* \beta_0 \\
        0 & 0 & -i\ell_1^* \beta_1 \\
        i\ell_0 \beta_0^* & i\ell_1\beta_1^* & 0  
    \end{bmatrix},
\end{align}
and consider the Lindbladian
\begin{align}\label{eq:perturbed_lindblad}
    \mathcal{L} = -i \left[ H,   \cdot \, \right] + \mathcal{D}_L.
\end{align}
We next show that this Lindbladian has two degenerate fixed points, which proves \cref{thm:hyp_2_false} from the main text.
\begin{proof}[Proof of \cref{thm:hyp_2_false}]
    \cref{lemma:irreducibility_qutrit_jump} shows that $\mathcal{D} = \mathcal{D}_L$ is irreducible (and thus exponentially contractive). We now consider the Lindbladian $\mathcal{L}$ in \cref{eq:perturbed_lindblad} and rewrite it as
    \begin{align}
        \mathcal{L}(X)=Q X + XQ^\dagger + L X L^\dagger,
    \end{align}
    where
    \begin{align}\label{eq:Q}
        Q = -iH-\frac{1}{2}L^\dagger L = -\frac{1}{2}\begin{bmatrix}
            \abs{\ell_0}^2 & 0 & 2\ell_0^* \beta_0 \\
            0 & \abs{\ell_1}^2 & 2\ell_1^*\beta_1 \\
            0 & 0 & \abs{\ell_0+\ell_1}^2 + \abs{\beta_0}^2 + \abs{\beta_1}^2 
        \end{bmatrix}.
    \end{align}
    Importantly, the specific choice of $H$ ensures that $Q$ is upper triangular even though $L^\dagger L$ is not. Now, observe that for $b \in\{0, 1\}$
    \begin{align}
        L\ket{b}=\ell_b \ket{b} \quad \text{and}\quad Q \ket{b}=-\frac{1}{2}\abs{\ell_b}^2 \ket{b},
    \end{align}
    from which it follows that $\mathcal{L}(\ket{b}\!\bra{b})=0$. Thus, $\mathcal{L}$ has two fixed points, $\ket{0}$ and $\ket{1}$, and is thus not exponentially contractive.
\end{proof}

\subsection{\texorpdfstring{\cref{hyp:uniform_gap_contractivity}}{hyp:uniform gap contractivity}} \label{app:subsec:counterexample hyp2}
We construct a time-dependent Lindbladian on the Hilbert space of a qutrit and a qubit: $\mathcal{H} =\mathbb{C}^3\otimes \mathbb{C}^2$. We will denote the computational basis for the qutrit via $\ket{0}, \ket{1}, \ket{e}$ (in order) and that of the qubit by $\ket{0}, \ket{1}$ (in order). We define the jump-only Lindbladian $\mathcal{D}$ with $4$ jump operators as follows:
\begin{align}\label{eq:time_dep_L0}
    \mathcal{D} = \mathcal{D}_{L_0} + \sum_{a \in \{x, y, z\}}\mathcal{D}_{L_a},
\end{align}
where
\begin{align}
    L_0 = L \otimes I \quad \text{and}\quad L_a=\frac{1}{\sqrt{3}}\ket{e}\!\bra{e}\otimes \sigma^a,
\end{align}
where $L$ is the qutrit jump operator defined in \cref{eq:qutrit_jump}. We will now define the time-dependent Hamiltonian $H(t)$ via
\begin{align}
    H(t)=e^{-i\omega t I\otimes \sigma^x /2}(H \otimes \ket{0}\!\bra{0})e^{i\omega t I\otimes \sigma^x /2} + h I\otimes \sigma^x,
\end{align}
where $H$ is qutrit Hamiltonian defined in \cref{eq:qutrit_hamil}. It is easy to check that
\begin{subequations}\label{eq:Ht_explicit}
   \begin{align}
    H(t)= H_0 + H_1 \cos \omega t + H_2 \sin \omega t,
\end{align}
where
\begin{align}
    H_0 = \frac{1}{2} H \otimes I + h I \otimes \sigma^x, H_1 = \frac{1}{2} H \otimes \sigma^z \quad \text{and}\quad H_2=-\frac{1}{2}H \otimes \sigma^y.
\end{align} 
\end{subequations}
We next show that $\mathcal{L}_t$ satisfies the conditions of \cref{thm:hyp_1_false}. We first provide a useful lemma.
\begin{lem}\label{lemma:simple_algebra_cond}
    Suppose $Q_0$ is an operator on $\mathcal{H}$ such that $\forall b \in\{0, 1\}, \ket{\phi}\in \mathbb{C}^2$, 
    \begin{align}
    (\ket{e}\!\bra{e}\otimes I)Q_0^k\ket{b}\otimes \ket{\phi} = 0\quad \text{for}\quad k \in \{1, 2\} \implies \ket{\phi} =0.
    \end{align}
    Then, $\textnormal{Alg}(\{Q_0, L_0, L_x, L_y, L_z\})$ is the full matrix algebra.
\end{lem}
\begin{proof}
 Suppose that $\mathcal{S}$ is a subspace that is invariant under $Q, L_0, L_x, L_y, L_z$. We now consider two cases: Either $\mathcal{S}$ is entirely contained in the subspace $\text{span}(\{\ket{0}, \ket{1}\})\otimes \mathbb{C}^2$ or it has a component outside it.
\begin{enumerate}
\item[(1)] Suppose $\mathcal{S} \subseteq  \text{span}(\{\ket{0}, \ket{1}\})\otimes \mathbb{C}^2$. Consider $\ket{\psi} = \ket{0}\otimes \ket{\phi_0} + \ket{1}\otimes \ket{\phi_1} \in\mathcal{S}$. Since $L_0 \ket{\psi} = \ell_0 \ket{0}\otimes \ket{\phi_0} + \ell_1 \ket{1}\otimes \ket{\phi_1} \in \mathcal{S}$ and $\ell_0\neq \ell_1$ by assumption, it follows that
    \begin{align}
        \ket{1}\otimes \ket{\phi_1}  =\frac{1}{\ell_1 -\ell_0}(L_0 -\ell_0 I)\ket{\psi} \in \mathcal{S}\quad \text{and}\quad  \ket{0}\otimes \ket{\phi_0} =\frac{1}{\ell_0 -\ell_1} (L_0 -\ell_1 I)\ket{\psi} \in \mathcal{S}.
    \end{align}
    Now, since $\mathcal{S}$ is also invariant under $Q_0$, $Q_0 \ket{b}\otimes \ket{\phi_b}, Q_0^2 \ket{b}\otimes \ket{\phi_b} \subseteq \mathcal{S}$
    \begin{align}
        (\ket{e}\!\bra{e}\otimes I) Q_0^k\ket{b}\otimes \ket{\phi_b}= 0 \quad \text{for}\quad k \in \{1, 2\}.
    \end{align}
    However, by the assumption on $Q$ in the lemma, this would imply that $\ket{\phi_b} = 0$ and hence $\ket{\psi} = 0$, which would imply that $\mathcal{S} = \{0\}$.
\item[(2)] Next, consider the case that $\exists \ket{\psi}  \in \mathcal{S}$ such that $(\ket{e}\!\bra{e}\otimes I) \ket{\psi}\neq 0$. Since $\text{Alg}(\{L_x, L_y, L_z\}) = \ket{e}\!\bra{e} \otimes M(\mathbb{C}^2)$, invariance of $\mathcal{S}$ under $L_x, L_y, L_z$ implies that $\ket{e}\otimes \mathbb{C}^2 \subseteq \mathcal{S}$. Furthermore, since $L_0$ also leaves $\mathcal{S}$ invariant, we then obtain that $\forall \ket{\phi} \in \mathbb{C}^2$: $L_0 (\ket{e}\otimes \ket{\phi}) \in \mathcal{S}$. Using the explicit expression for $L_0$ and the fact that $\ket{e}\otimes \ket{\phi}\in \mathcal{S}$, we then obtain $(\beta_0 \ket{0}+\beta_1\ket{1})\otimes \mathbb{C}^2 \subseteq \mathcal{S}$. Finally, applying $L_0$ to the subspace $(\beta_0 \ket{0}+\beta_1\ket{1})\otimes \mathbb{C}^2$ yields $(\beta_0 \ell_0\ket{0} +\beta_1\ell_1\ket{1})\otimes \mathbb{C}^2 \subseteq \mathcal{S}$. Since $\ell_0 \neq \ell_1$, $\ket{0}, \ket{1} \in \text{span}(\{\beta_0 \ket{0}+\beta_1\ket{1}, \beta_0 \ell_0\ket{0}+\beta_1\ell_1\ket{1}\})$ and thus $\ket{0}\otimes \mathbb{C}^2\subseteq \mathcal{S}$ and $\ket{1}\otimes \mathbb{C}^2\subseteq \mathcal{S} $. Thus, we have established that $\text{span}(\{\ket{0}, \ket{1}, \ket{e}\}) \otimes \mathbb{C}^2 \subseteq \mathcal{S}$, and consequently $\mathcal{S} = \mathbb{C}^3\otimes \mathbb{C}^2$.
\end{enumerate}
Since $\mathcal{S}$ is either $\{0\}$ or $\mathcal{H}$, $\text{Alg}(\{Q_0, L_0, L_x, L_y, L_z\})$ is the full matrix algebra.
\end{proof}

Next, we prove that both $\mathcal{D}$ as well as the time-averaged Lindbladian $\mathcal{D}-i[H_0, \cdot]$ are both irreducible.
\begin{lem}\label{lemma:irred_static_dissip}
    $\mathcal{D}$ as defined in \cref{eq:time_dep_L0} and $\mathcal{D}-i[H_0, \cdot]$  with $H_0$ as defined in \cref{eq:Ht_explicit} are both irreducible.
\end{lem}
\begin{proof}
We will use the fact that a Lindbladian with Hamiltonian $H$ and jump operators $L_1, L_2 \dots L_m$ is irreducible if and only if $\text{Alg}(\{Q_0, L_1, L_2 \dots L_m\})$, where $Q_0 = -iH - \frac{1}{2}\sum_{i=1}^m L_i^\dagger L_i$ is the full algebra. For both $\mathcal{D}$ and $\mathcal{D} - i[H, \cdot]$, the operator $Q_0$ is of the form
\begin{align}
 Q_0 = -\bigg(i \frac{\xi}{2}H+\frac{1}{2}L^\dagger L+ \frac{1}{2}\ket{e}\!\bra{e}\bigg)\otimes I - ih\xi I\otimes \sigma^x,
\end{align}
where $\xi = 0 $ for $\mathcal{L}_0$ and $\xi=1$ for $\mathcal{L}_0-i[\cdot, H]$. We show that this $Q_0$ satisfies the conditions of \cref{lemma:simple_algebra_cond}. Note from the explicit expression of $Q_0$ that for $b \in \{0, 1\}$ and $\ket{\phi} \in \mathbb{C}^2$,
\begin{align}
    ( \ket{e}\!\bra{e} \otimes I) Q_0 \ket{b}\otimes \ket{\phi}= \frac{\ell_b \beta_b^*}{2}\bigg(\frac{\xi}{2} -1\bigg) \ket{e}\otimes \ket{\phi}.
\end{align}
Since $\xi \neq 2$, we obtain that $ ( \ket{e}\!\bra{e} \otimes I) Q_0 \ket{b}\otimes \ket{\phi} = 0 \implies \ket{\phi} = 0$, thus satisfying the condition of \cref{lemma:simple_algebra_cond}. It then follows that $\text{Alg}(\{Q_0, L_0, L_x, L_y, L_z\})$ is the full matrix algebra for both $\mathcal{D}$ and $-i[H_0, \cdot] +\mathcal{D}$, which proves the lemma.
\end{proof}
Next, we prove that the time-dependent Lindbladian has, instantaneously, a unique fixed point and a gap independent of time.
\begin{lem}\label{lemma:inst_prim_and_gap}
    Consider the Lindbladian $\mathcal{L}_t = -i[H(t), \cdot]+\mathcal{D}$ with $h \neq 0$.  $\forall t:$ $\mathcal{L}_t$ is irreducible.
\end{lem}
\begin{proof}
   Recall that, given a unitary channel $\mathcal{U}= U (\cdot) U^\dagger$ and a Lindbladian $\mathcal{L}$ with Hamiltonian $H$ and jump operators $L_1, L_2 \dots L_m$, then
   \begin{align}
       \mathcal{U}^\dagger \mathcal{L} \mathcal{U} = -i[U^\dagger HU, \cdot] + \sum_{i = 1}^m \mathcal{D}_{U^\dagger L_i U}.
   \end{align}
   Now, choosing $U_t = e^{-i\omega t (I \otimes \sigma^x)/2}$, we then obtain that
   \begin{align}
       \mathcal{U}_t^\dagger \mathcal{L}_t \mathcal{U}_t &= -i[H(0), \cdot] + \mathcal{D}_{L_0} + \mathcal{D}_{L_x} + \mathcal{D}_{\cos \omega t L_y-\sin \omega t L_z} + \mathcal{D}_{\cos \omega t L_z+\sin \omega t L_y} \nonumber\\
       &= -i[H(0), \cdot]+ \sum_{a \in \{0, x, y, z\}}\mathcal{D}_{L_a} \nonumber\\
       &= \mathcal{L}_{t=0}.
   \end{align}
   Thus, the spectral gap of $\mathcal{L}_t$ is independent of $t$ and equal to the spectral gap of $\mathcal{L}_{t=0}$. Furthermore, if $\mathcal{L}_{t=0}$ is irreducible, then so is $\mathcal{L}_t$.
   
   We next show that $\mathcal{L}_{t=0}$ is irreducible. We again use \cref{lemma:simple_algebra_cond}, with 
   \begin{align}
       Q_0 = -iH(0)-\frac{1}{2}\sum_{a\in\{0, x, y, z\}}L_a^\dagger L_a= -iH\otimes \ket{0}\!\bra{0} - ih I\otimes \sigma^x - \frac{1}{2}\big(\ket{e}\!\bra{e}\otimes I +L^\dagger L \otimes I\big).
   \end{align}
   Using the explicit expression for $H$ and $L$, we can express $Q_0$ as
   \begin{align}
       Q_0 = Q\otimes \ket{0}\!\bra{0} - \frac{1}{2} L^\dagger L \otimes \ket{1}\!\bra{1} - \frac{1}{2}\ket{e}\!\bra{e}\otimes I - ih I \otimes \sigma^x,
   \end{align}
   with $Q$ being as defined in \cref{eq:Q}.
   
   Let us verify the condition in \cref{lemma:simple_algebra_cond}. Let $b \in\{0, 1\}$ and $\ket{\phi}\in \mathbb{C}^2$. Then, using the explicit expression for $Q_0$, we obtain that
   \begin{align}
       (\ket{e}\!\bra{e}\otimes I ) Q_0 \ket{b}\otimes \ket{\phi}=\frac{\ell_b \beta_b^*}{2}\ket{e}\otimes \big(\bra{0}\phi\rangle \ket{0} - \ket{\phi}\big),
   \end{align}
   and thus $(\ket{e}\!\bra{e}\otimes I ) Q_0 \ket{b}\otimes \ket{\phi} = 0$ implies that $\ket{\phi}=\bra{0}\phi\rangle \ket{0}$ (or, equivalently $\bra{1}\phi\rangle = 0$). Then, it follows that 
   \begin{align}
       Q_0 \ket{b}\otimes \ket{\phi} = \bra{0}\phi\rangle \ket{b}\otimes \bigg(-\frac{1}{2}\abs{\ell_b}^2  \ket{0}-ih  \ket{1}\bigg),
   \end{align}
   and therefore
   \begin{align}
      (\ket{e}\!\bra{e}\otimes I ) Q_0^2\ket{b}\otimes \ket{\phi} = \frac{ih {\ell_b}\beta_b^*}{2} \bra{0}\phi\rangle  \ket{e}\otimes \ket{1},
   \end{align}
   and thus $(\ket{e}\!\bra{e}\otimes I ) Q_0^2 \ket{b}\otimes \ket{\phi} = 0$ implies $\bra{0}\phi\rangle = 0$. Thus, $Q_0$ satisfies the conditions of \cref{lemma:simple_algebra_cond}, proving that $\mathcal{L}_{t = 0}$ is irreducible, and thus $\mathcal{L}_t$ is also irreducible.
\end{proof}
\noindent Finally, we prove that $\mathcal{L}_t$ is non-contractive if $\omega = 2h$.
\begin{proof}[Proof of \cref{thm:hyp_1_false}]
We will pick $\omega =2h$. \cref{lemma:irred_static_dissip} and \cref{lemma:inst_prim_and_gap} already establish that the time-dependent Lindbladian $\mathcal{L}_t$ satisfies all the conditions in \cref{thm:hyp_1_false}. We now show that the time-dependent Lindbladian is not contractive.

Suppose $\rho(t)$ satisfies the master equation with the Lindbladian $\mathcal{L}_t$. Consider now $\tilde{\rho}(t)=e^{i\omega t I\otimes \sigma^x/2}\rho(t) e^{-i\omega t I\otimes \sigma^x/2}$: It can be seen that $\tilde{\rho}(t)$ satisfies a time-independent master equation
\begin{align}
    \frac{d}{dt}\tilde{\rho}(t) = \tilde{\mathcal{L}} \tilde{\rho}(t) \quad \text{where}\quad \tilde{\mathcal{L}} = -i\bigg[H\otimes \ket{0}\!\bra{0} + \bigg(h-\frac{\omega}{2}\bigg) I \otimes \sigma^x, \cdot \bigg] + \mathcal{D}.
\end{align}
When $h = \omega/2$, $\tilde{\mathcal{L}}$ reduces to
\begin{align}
    \tilde{\mathcal{L}}(X) = (Q\otimes \ket{0}\!\bra{0})X + X(Q^\dagger \otimes \ket{0}\!\bra{0}) - \frac{1}{2}\{L^\dagger L \otimes \ket{1}\!\bra{1} + \ket{e}\!\bra{e} \otimes I, X\} + \sum_{a\in \{0, x, y, z\}} L_a X L_a^\dagger,
\end{align}
where $Q$ is defined in \cref{eq:Q}. Now, using the fact that $Q\ket{b} = -\abs{\ell_b}^2\ket{b}/2$ and $L\ket{b} = \ell_b \ket{b}$, we find that
\begin{align}
    \tilde{\mathcal{L}}(\ket{b}\!\bra{b}\otimes \ket{0}\!\bra{0}) = 0.
\end{align}
Thus, $\tilde{\rho}(t) = \ket{0}\!\bra{0}\otimes \ket{0}\!\bra{0}$ and $\ket{1}\!\bra{1}\otimes \ket{0}\!\bra{0}$ are both fixed points of $\tilde{\mathcal{L}}$. Therefore, $\tilde{\rho}(t) = \ket{0}\!\bra{0}\otimes \ket{\phi(t)}\!\bra{\phi(t)}$ and $\ket{1}\!\bra{1}\otimes \ket{\phi(t)}\!\bra{\phi(t)}$ where $\ket{\phi(t)} = e^{-i\omega \sigma^x t/2}\ket{0}$, are solutions of the master equation with $\mathcal{L}_t$ which remain at a trace-distance of $2$ at all times.
    
\end{proof}

\section{Proofs of theorems \ref{thm:perturbation:slow H} and \ref{thm:perturbation:small H}} \label{app:small_slow_results}

In this section, we will provide the proofs of \cref{thm:perturbation:slow H} and \cref{thm:perturbation:small H}. These theorems essentially follow from the following more general result about the stability of contractive behavior in Lindbladian dynamics under small Lindbladian perturbations, which might be of independent interest.

\begin{lem} \label{thm:gen perturbation thm}
Let $\mathcal{L}_t$ be any (time-dependent) exponentially contractive Lindbladian with universal constants $K , \gamma > 0$ from \cref{def:exp contraction}, i.e. $\norm{\mathcal{E}_{t, s} (\rho - \sigma)}_1\leq K e^{-\gamma \abs{t - s}} \norm{\rho - \sigma}_1$ holds for all $t \geq s \geq 0$, where $\mathcal{E}_{t,s} = \mathcal{T}\exp ( \int_s^t \mathcal{L}_\tau d\tau )$ is the time-evolution channel generated by $\mathcal{L}_t$. Let now $\tilde{\mathcal{L}}_t$ be any other time-dependent Lindbladian satisfying 
\begin{align} \label{eq:proof thm small H:cond main sec}
 \Delta \mathcal{L} \coloneqq\sup_{t \geq 0} \superopnorm{\tilde{\mathcal{L}}_t - \mathcal{L}_t} < \frac{\gamma}{1 + \ln(K)}.
\end{align}
Then $\tilde{\mathcal{L}}_t$ will also generate exponentially contractive dynamics with universal constants $\tilde{K} , \tilde{\gamma}$ that can be chosen as 
\begin{align*} 
     \tilde{\gamma} &= -\min_{x \geq 0} \frac{\gamma}{x} \ln \left( \frac{  1+\ln(K) }{\gamma} \Delta \mathcal{L} + \left( 1 - \frac{\Delta \mathcal{L}}{\gamma} \right) e^{-x}  \right), \\
     \tilde{K} &= e^{  x \tilde{\gamma} / \gamma},
\end{align*}
where $x$ in the expression for $\tilde{K} = e^{x \tilde{\gamma} / \gamma }$ is the value of $x \geq 0$ maximizing the expression for $\tilde{\gamma}$ above. In particular, if $K=1$, we get $\tilde{\gamma} = \gamma - \Delta \mathcal{L}$ and $\tilde{K} = 1$. 
\end{lem}
\begin{proof}
Let $\mathcal{L}_t$ and $\tilde{\mathcal{L}}_t$ be two time-dependent Lindbladians generating time-evolution channels $\mathcal{E}_{t,s}$ and $\tilde{\mathcal{E}}_{t,s}$ respectively, that is $\mathcal{E}_{t, s} = \mathcal{T}\exp(\int_s^t \mathcal{L}_\tau d\tau)$ and $\tilde{\mathcal{E}}_{t, s} = \mathcal{T}\exp(\int_s^t \tilde{\mathcal{L}}_\tau d\tau)$ where $\mathcal{T}$ is the time-ordering operator. Consider now two arbitrary quantum states $\rho, \sigma$. We want to compare the time-evolution of the difference $\rho - \sigma $ generated by $\mathcal{L}_t$ with the time-evolution generated by $\tilde{\mathcal{L}}_t$. For ease of notation, define for $t \geq s \geq 0$
\begin{align*}
x(t) \coloneqq \mathcal{E}_{t,s} (\rho - \sigma) = \mathcal{T}\exp \left( \int_s^t \mathcal{L}_\tau d\tau \right) (\rho - \sigma)  , \quad  \tilde{x}(t) \coloneqq \tilde{\mathcal{E}}_{t,s} (\rho - \sigma)  = \mathcal{T}\exp \left( \int_s^t \tilde{\mathcal{L}}_\tau d\tau \right) (\rho - \sigma) .
\end{align*}
Even though $x(t) - \tilde{x}(t)$ is continuous and differentiable on the space of Hermitian matrices, as it is a difference of solutions to Lindbladian master equations, $\norm{x(t) - \tilde{x}(t)}_1$ might not be differentiable for all $t$. Thus, to be rigorous, we technically have to work with the right time-derivative, see \cref{app:subsec:analysis result long}. Due to analysis results established in \cref{app:subsec:analysis result long}, we would here not get different results if we had naively worked with the ordinary derivative. Using various results elaborated on below and the fact that $\norm{x(t) - \tilde{x}(t)}_1$ must be at least right time-differentiable, we get
\begin{equation} \label{eq:proof:gen perturb proof long eq l7}
\begin{aligned} 
\partial_{t,+} \norm{x(t) - \tilde{x}(t)}_1 \overset{(1.)}{=} &  \lim_{\epsilon \rightarrow 0^+} \frac{\norm{x(t+\epsilon) - \tilde{x}(t+\epsilon)}_1 - \norm{x(t) - \tilde{x}(t)}_1}{\epsilon}  \\ 
\overset{(2.)}{=} & \lim_{\epsilon \rightarrow 0^+ } \frac{\norm{x(t) + \epsilon \mathcal{L}_t \left( x(t) \right) - \tilde{x}(t) - \epsilon \tilde{\mathcal{L}}_t \left(\tilde{x}(t) \right)}_1 - \norm{x(t) - \tilde{x}(t)}_1 }{\epsilon}  \\ 
\overset{(3.)}{=} & \lim_{\epsilon \rightarrow 0^+ } \frac{\norm{x(t) - \tilde{x}(t) + \epsilon \tilde{\mathcal{L}}_t \left( \vphantom{1^2} x(t)  - \tilde{x}(t) \right) - \epsilon \left( \tilde{\mathcal{L}}_t -\mathcal{L}_t \right) \left( x (t) \right)}_1 - \norm{x(t) - \tilde{x}(t)}_1 }{\epsilon} \\ 
\overset{(4.)}{\leq} & \lim_{\epsilon \rightarrow 0^+ } \frac{\norm{x(t) - \tilde{x}(t) + \epsilon \tilde{\mathcal{L}}_t \left( \vphantom{1^2} x(t) - \tilde{x}(t) \right) }_1 + \epsilon \norm{\left( \tilde{\mathcal{L}}_t -\mathcal{L}_t \right) \left( x (t) \right)}_1 - \norm{x(t) - \tilde{x}(t)}_1 }{\epsilon} \\
 \overset{(5.)}{=} &\norm{\left( \tilde{\mathcal{L}}_t -\mathcal{L}_t \right) \left( x (t) \right)}_1 + \lim_{\epsilon \rightarrow 0^+} \frac{\norm{x(t) - \tilde{x}(t) + \epsilon \tilde{\mathcal{L}}_t \left( \vphantom{1^2} x(t) - \tilde{x}(t) \right) }_1 - \norm{x(t) - \tilde{x}(t)}_1 }{\epsilon} \\ 
 \overset{(6.)}{\leq}& \norm{\left( \tilde{\mathcal{L}}_t -\mathcal{L}_t \right) \left( x(t) \right)}_1 \overset{(7.)}{\leq} \superopnorm{ \tilde{\mathcal{L}}_t -\mathcal{L}_t } \norm{ x (t) }_1 .
\end{aligned}
\end{equation}
In the derivation \eqref{eq:proof:gen perturb proof long eq l7} above: (1.) follows from definition of the right time-derivative and the fact that $\norm{x(t) - \tilde{x}(t)}_1$ must be right time-differentiable, (2.) follows from definition of $x(t)$, $\tilde{x}(t)$ and the time-evolution channels which generate solutions to the Lindbladian master equations, (3.) follows from linearity of the Lindbladians, (4.) follows from the triangle inequality for the $1$-norm, (5.) follows from basic manipulations, (6.) follows from the fact that the time-evolution channel cannot increase the $1$-norm of any Hermitian operator such as $x(t) - \tilde{x(t)}$, i.e. for any Hermitian $H$ and any (possibly time-dependent) Lindbladian $\mathcal{L}$, we must have $\lim_{\epsilon \rightarrow 0^+} ( \norm{H + \epsilon \mathcal{L}(H)}_1 - \norm{H}_1 ) / \epsilon \leq 0$, (7.) follows from the definition of the ($1 \rightarrow 1$) superoperator-norm $\superopnorm{ \ \cdot \ }$ in \cref{def:super 1 1 norm}.
\\
\\
Now, by definition of $x(t)$ and $\tilde{x}(t)$, we have $x(s) = \tilde{x}(s) = \rho - \sigma$. Combining this with the triangle inequality, and using that \cref{eq:proof:gen perturb proof long eq l7} entails that $\partial_{t,+} \norm{x(t) - \tilde{x}(t)}_1$ is upper bounded by the continuous function $ ||| \tilde{\mathcal{L}}_t -\mathcal{L}_t |||  \norm{ x (t) }_1$, and finally using that $\norm{x(t) - \tilde{x}(t)}_1$ is right time-differentiable and continuous, we get the following result by applying \cref{lem:fund thm calculus right diff functions} from \cref{app:subsec:analysis result long} (which essentially generalizes the fundamental theorem of calculus)
\begin{equation} \label{eq:proof:gen perturb proof int form}
\begin{aligned}
\norm{\tilde{x}(t)}_1 & = \norm{\tilde{x}(t)-x(t)+x(t)}_1 \leq \norm{x(t) - \tilde{x}(t)}_1 + \norm{x(t)}_1 \\ 
& \quad \leq \norm{x(t)}_1 + \norm{x(s) - \tilde{x}(s)}_1 + \int_s^t d t' \superopnorm{\tilde{\mathcal{L}}_{t' } - \mathcal{L}_{t'}} \norm{x(t' )}_1 = \norm{x(t)}_1 + \int_s^t d t' \superopnorm{\tilde{\mathcal{L}}_{t' } - \mathcal{L}_{t'}} \norm{x(t' )}_1 .
\end{aligned}
\end{equation}
We now have an upper bound on the distance between states evolved by $\tilde{\mathcal{L}}_t$ i.e. $ \tilde{x} (t) = \tilde{\mathcal{E}}_{t,s} (\rho - \sigma)$, which depends only on the distance between states evolved by the contractive Lindbladian $\mathcal{L}_t$, and the difference $\tilde{\mathcal{L}}_t - \mathcal{L}_t$. 
\\
\\
We thus see from \cref{eq:proof:gen perturb proof int form} above that if our original Lindbladian $\mathcal{L}_t$ generates asymptotically contractive dynamics, i.e. if for all $s \geq 0$ and quantum states $\rho , \sigma$ have $\lim_{t \rightarrow \infty} \norm{\mathcal{E}_{t,s}(\rho - \sigma)}_1 = 0$, then by requiring $ ||| \tilde{\mathcal{L}}_{t} - \mathcal{L}_t ||| $ to be sufficiently bounded (where the bound might here depend on time), there will for any $s \geq 0$ exist some time $T_s$ such that for all $t \geq T_s$, the distance between the states that are time-evolved by the new Lindbladian $\tilde{\mathcal{L}}_t$ will contract by some factor -- say $\frac{1}{2}$. Since the previous argument holds for arbitrary $s$, we can repeat the argument with $T_s$ being our new $s$, and so on. Thus, under the assumption that $ ||| \tilde{\mathcal{L}}_{t } - \mathcal{L}_t ||| $ remains sufficiently bounded, the new Lindbladian $\tilde{\mathcal{L}}_t$ must also be asymptotically contractive. 
\\
\\
We now focus on establishing exponential contractivity, so assume that $\mathcal{L}_t$ is exponentially contractive with universal constants $K$ and $\gamma$ as in \cref{def:exp contraction}, i.e. for all $t \geq s \geq 0$, we have for arbitrary quantum states $\rho, \sigma$
\begin{align} \label{eq:proof:gen perturb proof exp contract def}
\norm{\mathcal{E}_{t,s}(\rho - \sigma)}_1 \leq K e^{- \gamma |t-s|} \norm{\rho - \sigma}_1  .
\end{align}
Note that since any Lindbladian time-evolution cannot increase the $1$-norm between states, the inequality \eqref{eq:proof:gen perturb proof exp contract def} above can trivially be improved to $\norm{\mathcal{E}_{t,s}(\rho - \sigma)}_1 \leq \min \{ K e^{- \gamma |t-s|} , 1 \} \norm{\rho - \sigma}_1$. Plugging this improvement of \cref{eq:proof:gen perturb proof exp contract def} into \cref{eq:proof:gen perturb proof int form} above as well as defining the number $\Delta \mathcal{L} \coloneqq \sup_{t \geq 0} ||| \tilde{\mathcal{L}}_{t } - \mathcal{L}_t ||| $ and using the definitions of $x(t)$ and $\tilde{x}(t)$, gives us that for any $t \geq s + \ln(K) / \gamma  \geq 0$ the following upper bound on the distance between time-evolved arbitrary quantum states $\rho , \sigma$ according to the new Lindbladian $\tilde{\mathcal{L}}_t$
\begin{equation} \label{eq:proof:gen perturb proof exp decay step}
\begin{aligned}
\norm{\tilde{\mathcal{E}}_{t,s} (\rho - \sigma)}_1 
&\leq  \norm{x(t)}_1 + \int_s^t d t' \superopnorm{\tilde{\mathcal{L}}_{t' } - \mathcal{L}_{t'}} \norm{x(t' )}_1 \\
&\leq K e^{- \gamma |t-s|} \norm{\rho - \sigma}_1 + \Delta \mathcal{L} \int_s^t dt' \min \left\{ K e^{- \gamma |t'-s|} , 1 \right\} \norm{\rho - \sigma}_1 \\
& = \left[ \Delta \mathcal{L} \frac{1 + \ln(K)}{\gamma} + K \left( 1 - \frac{\Delta \mathcal{L}}{\gamma} \right) e^{- \gamma |t-s|} \right] \norm{\rho - \sigma}_1  .
\end{aligned}
\end{equation}
We thus see from \cref{eq:proof:gen perturb proof exp decay step} above that for $|t-s| \gg 1 / \gamma$, $\norm{\tilde{\mathcal{E}}_{t,s} (\rho - \sigma)}_1$ will approximately equal $\Delta \mathcal{L} ( 1 + \ln(K) ) / \gamma \norm{\rho - \sigma}_1$. Thus, if $\Delta \mathcal{L} < \gamma / ( 1+ \ln(K) )$, there will exist a time $T$ such that for all $|t-s| \geq T$, we have $\norm{\tilde{\mathcal{E}}_{t,s} (\rho - \sigma)}_1 \leq C \norm{\rho - \sigma}_1$, where $C$ is some constant strictly less than $1$. Since $s$ was arbitrary in the argument above, we can keep repeating the argument just mentioned and get $\norm{\tilde{\mathcal{E}}_{t,s} (\rho - \sigma)}_1 \leq e^{  \ln (C) \lfloor t / T \rfloor } \norm{\rho - \sigma}_1 \leq e^{  \ln \left(  1 / C \right)} e^{- \ln \left( 1 / C \right) t / T } \norm{\rho - \sigma}_1 $. This thus proves that the new Lindbladian $\tilde{\mathcal{L}}_t$ also generates exponentially contractive dynamics provided $ \Delta \mathcal{L} \coloneqq \sup_{t \geq 0} ||| \tilde{\mathcal{L}}_{t } - \mathcal{L}_t ||| < \gamma / ( 1 + \ln (K) )$. In this case, the time-evolution $\tilde{\mathcal{E}}_{t,s}$ generated by $\tilde{\mathcal{L}}_t$ will satisfy
\begin{align*}
\norm{\tilde{\mathcal{E}}_{t,s} (\rho - \sigma)}_1 \leq \tilde{K} e^{- \tilde{\gamma} |t-s|} \norm{\rho - \sigma}_1  ,
\end{align*}
where $\tilde{K}, \tilde{\gamma}$ can be chosen as
\begin{align*}
\tilde{\gamma} \coloneqq \gamma \ \max_{0 \leq x} \bigg[-\frac{1}{x} \ln \left( 1 - A + B e^{-x} \vphantom{\frac{1}{1}} \right)\bigg] > 0  , \quad \tilde{K} = e^{x \tilde{\gamma} / \gamma},
\end{align*}
where
\begin{equation}
    A  \coloneqq 1 - \frac{1 + \ln (K)}{\gamma} \Delta \mathcal{L}   \qquad
    B  \coloneqq K \left( 1 - \frac{\Delta \mathcal{L}}{\gamma} \right).
\end{equation}
This proves the main statement of the theorem. Note that it follows from the assumption $ \Delta \mathcal{L} \coloneqq \sup_{t \geq 0} ||| \tilde{\mathcal{L}}_{t } - \mathcal{L}_t ||| < \gamma / ( 1 + \ln (K) )$ that we must have $0 < A < 1$ and $B > 0$, which makes $\tilde{\gamma}$ above well-defined. 

Note finally that instead of requiring an upper bound on $\sup_{t \geq 0} |||  \tilde{\mathcal{L}}_{t } - \mathcal{L}_t ||| $, we could from \cref{eq:proof:gen perturb proof int form} have obtained a more general result by simply requiring an upper bound on arbitrary long time-averages of $ ||| \tilde{\mathcal{L}}_{t } - \mathcal{L}_t ||| $. Since we always have $\norm{x(t)}_1 \leq \norm{x(s)}_1$ for $t \geq s$ in \cref{eq:proof:gen perturb proof int form}, we see that assuming there exists some time $T > \ln(K) / \gamma$ such that the average Lindbladian perturbation in the interval is sufficiently upper bounded, the new Lindbladian will still be contractive. Specifically, assuming the following holds
\begin{align*}
\sup_{t \geq 0} \frac{1}{T} \int_t^{t+T} d t' \superopnorm{\tilde{\mathcal{L}}_{t' } - \mathcal{L}_{t'}} < \frac{1- K e^{ - \gamma T}}{T} \qquad \text{for some} \ \ T > \frac{\ln (K)}{\gamma} .
\end{align*}
Then every time a time interval of $T$ passes, the $1$-norm of the difference between two states will contract by at least a factor of $ \overline{\Delta \mathcal{L}} T + K e^{- \gamma T } < 1$ under the time-evolution of $\tilde{\mathcal{L}}_t$, where $\overline{\Delta \mathcal{L}} \coloneqq \sup_{t \geq 0} \int_t^{t+T} d t' ||| \tilde{\mathcal{L}}_{t' } - \mathcal{L}_{t'} ||| / T$, showing that $\tilde{\mathcal{L}}_t$ will generate exponentially contractive dynamics. For example, if the average perturbation in every time-interval of length $T = \frac{\ln( 2 K)}{\gamma}$ is bounded as follows
\begin{align*}
\sup_{t \geq 0} \frac{1}{T} \int_t^{t+T} d t' \superopnorm{\tilde{\mathcal{L}}_{t' } - \mathcal{L}_{t'}} < \frac{\gamma}{ 2 \ln (2 K)}  \quad \text{where} \ \  T=\frac{\ln(2 K)}{\gamma} ,
\end{align*}
then $\tilde{\mathcal{L}}_t$ will generate exponentially contractive dynamics. 
\end{proof}

\begin{proof}[Proof of \cref{thm:perturbation:slow H}] \label{proof:thm:perturbation:slow H}
Suppose $\mathcal{L}_t = -i[H(t), \cdot] + \mathcal{D}$ is a time-dependent Lindbladian that is instantaneously contractive in the sense of \cref{def:exp contraction} with universal constants $K_0 , \gamma_0$, i.e. for all times $t_0 \geq 0$, the time-evolution generated by the constant Lindbladian $\mathcal{L}_{t_0}$ satisfies that for any quantum states $\rho, \sigma$ and times $t \geq s \geq 0$, we have $\norm{\mathcal{E}_{s,t}^{(0)} (\rho - \sigma)}_1 \leq K_0 e^{-\gamma_0 |t-s|} \norm{\rho - \sigma}_1$ where $\mathcal{E}_{t,s}^{(0)} \coloneqq \mathcal{T}\exp ( \int_s^t \mathcal{L}_{t_0} d\tau ) = e^{|t-s| \mathcal{L}_{t_0}}$. Note now that for any fixed time $t_0$, the difference between the original time-dependent Lindbladian $\mathcal{L}_{t}$ and the instantaneous Lindbladian $\mathcal{L}_{t_0}$ is upper bounded by
\begin{align} \label{eq:proof:perturbation derivative bound:s1}
\superopnorm{\mathcal{L}_t - \mathcal{L}_{t_0}} = \superopnorm{\int_{t_0}^t d t' \frac{d}{dt'} \mathcal{\mathcal{L}}_{t'} } \leq  \int_{ \min \{ t , t_0 \} }^{ \max \{ t , t_0 \} } dt' \superopnorm{\frac{d}{dt'} \mathcal{\mathcal{L}}_{t'}}  \leq  |t - t_0| \sup_{t' \geq 0} \superopnorm{\frac{d}{dt'} \mathcal{L}_{t'} } ,
\end{align}
where we have used that the superoperator-norm satisfies the triangle inequality. Let $\mathcal{E}_{t,s} \coloneqq \mathcal{T} \exp ( {\int_s^t d \tau \mathcal{L}_{\tau}} )$ denote the time-evolution governed by the original time-dependent Lindbladian $\mathcal{L}_t$. We now use the assumption that $\mathcal{L}_{t_0}$ for any $t_0$ generates exponentially contractive dynamics with constants $K_0 , \gamma_0$ in the bound \eqref{eq:proof:gen perturb proof int form} from the proof of \cref{thm:gen perturbation thm}, along with \cref{eq:proof:perturbation derivative bound:s1} above to get that for any quantum states $\rho , \sigma$ and times $t_0$ and $t \geq s \geq 0$ that
\begin{equation} \label{eq:proof:perturbation derivative int step}
\begin{aligned}
\norm{\mathcal{E}_{t,s}(\rho - \sigma)}_1 & \leq \norm{\mathcal{E}_{t,s}^{(0)}(\rho - \sigma)}_1 + \int_s^t d t' \norm{\mathcal{E}_{ t' ,s}^{(0)}(\rho - \sigma)}_1 \superopnorm{\mathcal{L}_{t'} - \mathcal{L}_{t_0} } \\
& \leq \left( K_0 e^{- \gamma_0 |t-s|} + l \int_s^t dt' \min \left\{ 1 , K_0 e^{- \gamma_0 |t'-s|} \right\} |t' - t_0 |\right) \norm{\rho - \sigma}_1 ,
\end{aligned}
\end{equation}
where we have defined $l \coloneqq \sup_{t \geq 0} \superopnorm{ d \mathcal{L}_t / dt }$ for ease of notation. Now, if we set $t \geq s + \ln(K_0) / \gamma_0 $ and we choose $t_0$ to lie in the interval $t_0 \in [ s , s + \ln(K_0) / \gamma_0 ]$, while defining $z\coloneqq t_0 - s$, in \cref{eq:proof:perturbation derivative int step} above, we get 
\begin{equation} \label{eq:choice for z}
\begin{aligned}
\norm{\mathcal{E}_{t,s}(\rho - \sigma)}_1 \leq \left[ K_0 e^{- \gamma_0 |t-s|} + l  \left( \frac{1}{2} z^2 + \frac{1}{2} \left( \frac{\ln(K_0)}{\gamma_0} - z \right)^2 + \int_{s+\frac{\ln(K_0)}{\gamma_0}}^t dt' |t'-t_0| K_0 e^{- \gamma_0 |t'-s|} \right)  \right] \norm{\rho - \sigma}_1  \\
= \left[ l \left( z^2 + \frac{\ln(K_0)^2}{2 \gamma_0^2} - \frac{\ln(K_0)}{\gamma_0} z + \frac{1 + \ln(K_0)}{\gamma_0^2}- \frac{z}{\gamma_0} \right)  + K_0 \left( 1 - l \frac{1 + \gamma_0 |t-s| - \gamma_0 z}{\gamma_0^2} \right) e^{- \gamma_0 |t-s|} \right] \norm{\rho - \sigma}_1 .
\end{aligned}
\end{equation}
We now make the following choice for $z$ above, which will be optimal if $ \ln(K_0) \geq 1$, namely $z = (\ln(K_0)+1 ) / ( 2 \gamma_0 )$. If on the other hand we have $\ln(K_0) \leq 1$, and we make the choice $z = (\ln(K_0)+1 ) / (2 \gamma_0 )$, then the integral evaluation in \cref{eq:choice for z} will be wrong, since in the interval $ [ s+ \ln(K_0) / \gamma_0 , s + z ]$, we will have integrated over an integrand $|t-t_0|(2 - e^{-\gamma_0 (t-t_c) })$ instead of the correct integrand of $|t-t_0| e^{-\gamma_0 (t-t_c) }$. Here, $t_c = \frac{\ln(K_0)}{\gamma_0} \leq z$. However, this discrepancy will still result in an upper bound holding, since we overestimate the integral. In conclusion, we make the choice $z = ( \ln(K_0)+1 ) / (2 \gamma_0 )$ in \cref{eq:choice for z} above, which will in all cases result in the following upper bound
\begin{equation} \label{eq:contraction for small dH dt}
\norm{\mathcal{E}_{t,s}(\rho - \sigma)}_1 \leq \left( \frac{l}{\gamma_0^2} \left( \frac{3}{4}+\frac{1}{2}\ln(K_0)+\frac{1}{4}\ln(K_0)^2 \right) + K_0 \left( 1 - l \frac{1-\ln(K_0)+ 2 \gamma_0 |t-s|}{2 \gamma_0^2} \right) e^{-\gamma_0 |t-s|} \right) \norm{\rho - \sigma}_1  .
\end{equation}
We thus see that if $l \coloneqq \sup_{t \geq 0} \superopnorm{ d \mathcal{L}_t / dt } <   4\gamma_0^2 / ( 3 + 2 \ln(K_0) + \ln(K_0)^2 )$, then for sufficiently large $|t-s|$, we will eventually have $\norm{\mathcal{E}_{t,s}(\rho - \sigma)}_1 \leq C \norm{\rho - \sigma}_1$, where $C$ is some constant $C < 1$. Since the times $t \geq s \geq 0$ and quantum states $\rho , \sigma$ were all completely arbitrary, this shows that $\mathcal{L}_t$ must indeed be exponentially contractive in the case $l < 4\gamma_0^2 / ( 3 + 2 \ln(K_0) + \ln(K_0)^2 )$. But we can get a more general contractivity condition by simply requiring $\gamma > 0$ below. Like in the proof of \cref{thm:gen perturbation thm}, we can in this case compute exponential contraction constants $K , \gamma$ for the time-evolution generated by $\mathcal{L}_t$. Specifically, $\mathcal{L}_t$ will be contractive with constants $K , \gamma$, i.e. $\norm{\mathcal{E}_{t,s}(\rho - \sigma) }_1 \leq K e^{- \gamma |t-s|} \norm{\rho - \sigma}_1$, given by
\begin{align*} 
\gamma = \gamma_0 \max_{x \geq 0} \frac{-1}{x} \ln \left( A + B e^{-x} - C x e^{-x} \right)  , \quad K = e^{ x \gamma / \gamma_0 } ,
\end{align*}
for any $x \geq 0$ for which $\gamma > 0$, and the constants $A , B , C$ are given by
\begin{align*}
A \coloneqq \frac{l}{\gamma_0^2} \left( \frac{3}{4}+\frac{1}{2}\ln(K_0)+\frac{1}{4}\ln(K_0)^2 \right)  , \quad B\coloneqq K_0 \left( 1 - l \frac{1- \ln(K_0)}{2 \gamma_0^2} \right)  , \quad C\coloneqq \frac{l K_0}{\gamma_0^2}  .
\end{align*}
Note in particular that in case $K_0=1$ -- i.e. if the instantaneous Lindbladians are always strictly $1$-norm contractive, the simple contraction condition \eqref{eq:bounded dH dt condition} guarantees exponential contractivity of $\mathcal{L}_t$ if $l < 4 \gamma_0^2 / 3$. However, it can be shown from \cref{eq:proof:perturbation derivative int step} that if $K_0 = 1$, then $\mathcal{L}_t$ will be exponentially contractive with rate $\gamma = \gamma_0$, for arbitrary $\gamma_0$, displaying some of the overestimation errors made in deriving the simple contraction condition. 

Consider the time-dependent Lindbladian $\mathcal{L}_t = - i[H(t), \ \cdot \ ] + \mathcal{D}$. Suppose that the constant instantaneous Lindbladian $\mathcal{L}_{t_0}  = - i[H(t_0), \ \cdot \ ] + \mathcal{D}$ obtained from fixing $t=t_0$ in $\mathcal{L}_t$ is exponentially contractive for all choices of $t_0$ with constant $K_0$ and contraction rate $\gamma_0$ as defined in \cref{def:exp contraction}. Suppose also that the operator $dH(t)/dt$ is bounded by $\sup_{t \geq 0} \left\| d H(t)/dt \right\|_\infty < {2\gamma_0^2 } / ({3 + 2 \ln K_0 +  \ln^2 K_0})$. Then, by using that $\superopnorm{[M , \ \cdot \ ]} \leq 2 \norm{M}_\infty$ holds for any matrix $M$, we get
\begin{align*}
\sup_{t \geq 0} \superopnorm{ \frac{d}{dt} \mathcal{L}_t} &= \sup_{t \geq 0} \superopnorm{ - i \left[ \frac{d}{dt} H(t), \ \cdot \ \right]} \nonumber\\
&\leq \sup_{t \geq 0} 2 \left\| \frac{d}{dt} H(t) \right\|_\infty \nonumber\\
&< \frac{4\gamma_0^2 }{3} \frac{1}{1 + \frac{2}{3} \ln(K_0) + \frac{1}{3} \ln(K_0)^2},
\end{align*}
which by using \cref{eq:contraction for small dH dt} and the subsequent discussion proves that the time-dependent Lindbladian $\mathcal{L}_t$ is in this case exponentially contractive.

\end{proof}

\begin{proof}[Proof of \cref{thm:perturbation:small H}] \label{proof:thm:perturbation:small H}

Consider the time-dependent Lindbladian $\mathcal{L}_t = - i[H(t), \ \cdot \ ] + \mathcal{D} = - i[V(t) +H_0, \ \cdot \ ] + \mathcal{D}$. Suppose that the constant Lindbladian $ - i[H_0, \ \cdot \ ] + \mathcal{D}$ is exponentially contractive with constant $K$ and contraction rate $\gamma$ as defined in \cref{def:exp contraction}, and suppose also that $\sup_{t \geq 0} \norm{V (t)}_\infty < {\gamma}/({2 + 2\ln K})$. Then, using that for any matrix $M$, $\superopnorm{[M , \ \cdot \ ]} \leq 2 \norm{M}_\infty$ (as e.g. follows from \cite[eq. 1.175]{Watrous2018}), we get
\begin{align*}
&\sup_{t \geq 0} \superopnorm{\mathcal{L}_t - \left( - i[H_0, \ \cdot \ ] + \mathcal{D} \vphantom{2^2} \right)} 
= \sup_{t \geq 0} \superopnorm{  - i[V(t), \ \cdot \ ] } \nonumber\\
&\qquad \leq \sup_{t \geq 0} 2 \norm{V(t)}_\infty  < \frac{\gamma}{1 + \ln K},
\end{align*}
which proves that $\mathcal{L}_t$ is in this case exponentially contractive by use of \cref{thm:gen perturbation thm}, since the condition in \cref{eq:proof thm small H:cond main sec} is satisfied when we identify $\tilde{\mathcal{L}}_t$ with $- i[V(t) +H_0, \ \cdot \ ] + \mathcal{D}$ and $\mathcal{L}_t$ with the exponentially contractive Lindbladian $- i[H_0, \ \cdot \ ] + \mathcal{D}$. 

\end{proof}

\section{Hamiltonian-independent contractivity: Proof of \texorpdfstring{\cref{thm:complete_characterization}}{thm:complete characterization} and its consequences} \label{app:hamil_ind_contr}
\subsection{Graph-theoretic preliminaries} \label{app:subsec graph theory}

Before proving \cref{thm:complete_characterization}, we first introduce various graph theoretic quantities and lemmas used in this proof. Let $G = (V , E )$ be a \emph{directed graph} with vertex set $V$ and edge set $E$, such that there is an edge in $G$ going from vertex $v_1 \in V$ to vertex $v_2 \in V$ whenever the ordered tuple $(v_1 , v_2) \in E$, in which case we write $v_1 \rightarrow_G v_2$, or simply $v_1 \rightarrow v_2$ if it is clear what graph we are talking about. Otherwise, i.e. if $(v_1 , v_2) \notin E$, we write $v_1 \not\rightarrow_G v_2$. The edge set $E$ of a directed graph can consist of any combination of ordered vertex-pairs, hence $G$ will in general not be symmetric, i.e. $v_1 \rightarrow_G v_2$ will not generally imply $v_2 \rightarrow_G v_1$ for $v_1 , v_2 \in V$. Note however, that with the convention used here, it is still possible to have both $v_1 \rightarrow_G v_2$ and $v_2 \rightarrow_G v_1$, that is, we do not define a directed graph in a way that requires them to be \emph{oriented}. 

For a directed graph $G = (V , E )$, we call a set of vertices $S \subseteq V$ in $G$ \emph{non-trivial} if $S \neq \emptyset , V$, and we say that $S$ has no \emph{inflowing edges} if for all $v \in S$ and all $u \in V \backslash S$, we have $u \not\rightarrow v$. Further, we say that a vertex $v \in V$ can be reached from $u \in V$ if there is a \emph{path} of directed edges from $u$ to $v$ in $G$, i.e. if there exist vertices $v_1 = u , v_2 , \ldots , v_{l-1} , v_l = v$ for some $l \in \mathbb{N}$ such that $v_k \rightarrow_G v_{k+1}$ holds for all $1 \leq k \leq l-1$. Finally, we call a vertex $v \in V$ a \emph{global source} (in $G$) if all vertices in $G$ can be reached from $v$, i.e. if for all $u \in V$, there is a path in $G$ going from $v$ to any $u$. 

The following result provides a useful necessary and sufficient condition for a directed graph to contain a global source. Although the result is likely well known, or at least equivalent to well-known statements, in the graph theory literature (see e.g. \cite{Edmonds1967-dk, Harary1969-lq}, \cite[sec. 1.5]{Bang-Jensen2001-hu}), we provide a proof of it below for completeness. 

\begin{prp} \label{prop:graph theory global source condition}
A finite directed graph $G$ contains a global source if and only if it does not contain two disjoint and non-trivial sets of vertices with no inflowing edges. 
\end{prp}

\begin{proof} \label{proof:grapth theory result directed source}
It is straightforward to see that if a directed graph $G=(V,E)$ contains a global source, then it cannot contain two disjoint and non-trivial sets of vertices with no inflowing edges, since no vertices in a set with no inflowing edges can be reached from any vertex outside of the set, and $G$ could then not contain a global source if two such disjoint and non-trivial sets existed. To complete the proof of \cref{prop:graph theory global source condition}, we must prove the converse, i.e. we must prove that if a directed graph $G$ does not contain two disjoint and non-trivial sets of vertices with no inflowing edges, then $G$ must contain a global source. 

Assume then that $G = (V,E)$ does not contain two disjoint and non-trivial sets of vertices with no inflowing edges, and suppose for contradiction that $G$ does not contain a global source. Let $u \in V$ be a vertex with a maximum number $N_r$ of reachable vertices in $G$, i.e. $| \{ w \in V | \text{$w$ can be reached from $u$ in $G$ } \} | = N_r $ where $ N_r \coloneqq \max_{v \in V} | \{ w \in V | \text{$w$ can be reached from $v$ in $G$} \} |$. Further, let $V_{u , \max}$ be the set of vertices of the largest maximally connected subgraph of $G$ containing $u$ -- i.e. $V_{u, \max }$ is the largest set of vertices in $G$ containing $u$ such that all vertices in $V_{u, \max }$ can be reached from any other vertex in $V_{u, \max }$. It is easy to show that $V_{u, \max }$ must be unique. It follows that any vertex in $V_{u, \max }$ must, like $u$, be a vertex with the maximum number $N_r$ of reachable vertices in $G$.

Let now $R_u$ be the set of vertices in $G$ that can be reached from $u$. By definition, we must have $V_{u ,\max} \subseteq R_u$ and further $| R_u | = N_r$. Since we assumed $G$ to contain no global source, $R_u$ must be a strict subset of $V$, otherwise $u$ would be a global source. We further define the vertex sets $A , B \subset V$ of $G$ as $A \coloneqq R_u \backslash V_{u , \max}$ and $B \coloneqq V \backslash R_u$. Since $R_u$ was a strict subset of $V$, we know that $B \subset V(G)$ must be a non-trivial vertex set (meanwhile $A$ might be empty). Now, for arbitrary vertex-sets $U , W \subseteq V$, we write $U \not\rightarrow W$ if for all $v_U \in U$ and all $v_W \in W$, we have $v_U \not\rightarrow v_W$. We can now conclude that:
\begin{equation} \label{eq:app:graph theory result}
\begin{aligned}
& (1.) \ A \not\rightarrow V_{u , \max} \\
& (2.) \ A \not\rightarrow B \\
& (3.) \ V_{u , \max} \not\rightarrow B \\
& (4.) \ B \not\rightarrow V_{u . \max} .
\end{aligned}
\end{equation}
$(1.)$ follows from how $V_{u, \max}$ was defined as the largest maximally connected subset containing $u$. Also by definition, since any vertex in $A$ can be reached from any vertex in $V_{u, \max}$, if a vertex in $A$ had a directed edge to a vertex in $V_{u, \max}$, then there would exist a strictly larger maximally connected set containing $u$.
$(2.)$ follows from $B$ being the set of vertices that cannot be reached from $u$. Since all vertices in $A$ can be reached from $u$, there can be no directed edge going from $A$ to $B$, since then some vertex in $B$ would be reachable from $u$. 
$(3.)$ similarly follows from how $B$ being the set of vertices that cannot be reached from $u$. 
$(4.)$ follows from the assumption of $u$ being a vertex with the maximum number $N_r$ of reachable vertices in $G$. Since $B$ cannot be reached from $V_{u, \max}$, if there were a directed edge from $B$ to $V_{u, \max}$, then there would exist a vertex in $B$ that can reach at least one more vertex than $u$ in $G$.

We now see that $V_{u, \max}$ and $B$ are two non-trivial, disconnected vertex-sets with no inflowing edges, as follows from the results in Eq. \eqref{eq:app:graph theory result} above and the fact that $V = V_{ u , \max} \cup A \cup B$. However, this contradicts the initial assumption that $G$ did not contain two such vertex sets, proving by contradiction that $G$ must indeed contain a global source. 

\end{proof}

A \emph{directed spanning tree with a global source} $\tau$ is a directed tree graph, i.e. a directed graph wherein each pair of vertices are connected by at most one path, which contains a vertex $v \in V(\tau)$ such that there exists a directed path from $v$ to any other vertex $u \in V(\tau)$ in $\tau$. We let $\mathcal{S}_d$ denote the set of all directed spanning trees with a global source on $d$ (labelled) vertices. We further let $\Xi_d$ denote the number of such distinct trees, $\Xi_d \coloneqq | \mathcal{S}_d |$. Although the exact value of $\Xi_d$ is not crucial to the following proof (the crucial part is that $\Xi_d$ is finite for all $d$), it can be calculated exactly as 
\begin{equation} \label{eq:exact number of spanning trees}
\Xi_d \coloneqq |\mathcal{S}_d| = d^{d-1}  ,
\end{equation}
which follows from \emph{Cayley's formula} \cite{Takacs1990-bc}, stating that the number of undirected spanning trees for the complete graph on $d$ vertices equals $d^{d-2}$, and for each such undirected tree, we have $d$ choices for the global source, each resulting in a unique choice of direction for all edges. 

Further, we let $\mathcal{U}(d)$ denote the set of all $d \times d$ unitary matrices. Next, for any dissipator $\mathcal{D} = \sum_j \mathcal{D}_{L_j}$ acting on $\mathbb{C}^d$ and any unitary $U \in \mathcal{U}(d)$, we define the $d \times d$ matrix $A (U , \mathcal{D})$ to have the following components
\begin{equation} \label{eq:def:A matrix dissipator}
A (U , \mathcal{D})_{ik} \coloneqq \left\{ \begin{matrix}
\sum_j \left| \braket{i | U L_j U^\dagger | k } \right|^2 & \text{if} \ i \neq k \\
0 & \text{if} \ i=k
\end{matrix} \right. ,
\end{equation}
where $\{ \ket{i} \}_{i=1}^d$ denotes the standard orthonormal basis-vectors. We then let $G(\mathcal{D} , U)$ denote the directed graph whose \emph{adjacency matrix} is given by the non-zero components of $A (U , \mathcal{D})$, i.e. if we label the vertices of $G(\mathcal{D} , U)$ by $v_1 , \ldots , v_d$, then
\begin{equation} \label{eq:def of dissipator graph}
v_i \rightarrow v_k \Leftrightarrow i \neq k \ \text{and} \ \sum_j \left| \braket{i | U L_j U^\dagger | k} \right|^2 > 0 .
\end{equation}
In a sense, the graph $G(\mathcal{D} , U)$ thus indicates between which basis-vectors or ``energy-levels" in the basis determined by $U$ the dissipator $\mathcal{D}$ induces a transmission. 

\subsection{Time evolution of eigenvalues for solutions to the Lindblad master equation} \label{app:subsec:eigenvalue lindlabian evolution}

We here derive an evolution-equation for the eigenvalues of an operator that is evolved by a Lindbladian. This result is crucial in the proof of \cref{thm:complete_characterization}, but it is also used in the proof of other results. Recall that $\mathcal{T}$ is here the time-ordering operator, and for any time-dependent operator $O(t)$, the so-called \emph{Peano-Baker series} \cite{Hale2009-sf, Baake2011-xi} $\mathcal{T} \exp ( \int_s^t d \tau O(\tau) )$ is (assuming it exists) defined as 
\begin{equation} \label{eq:app:eigenvalue evolution:Peano Baker series definition}
\mathcal{T} \exp \left( \int_s^t d \tau O(\tau) \right) \coloneqq \sum_{n=0}^\infty \mathcal{I}_{t,s}^n (O) , \qquad \text{with} \ \ \mathcal{I}_{t,s}^n (O) \coloneqq \int_{t \geq \tau_1 \geq \cdots \geq \tau_n \geq s} d \tau_1 \cdots d \tau_n O(\tau_1) \cdots O(\tau_n) ,
\end{equation}
or alternatively, the operator $\mathcal{I}_{t,s}^n (O)$ can be defined recursively by $\mathcal{I}_{t,s}^0 (O) \coloneqq I$ and $\mathcal{I}_{t,s}^{n+1} (O) \coloneqq \int_s^t d \tau O(\tau) \mathcal{I}_{\tau,s}^n (O)$. We are now ready to state the main result of this section, which details how the eigenvalues of a solution to the Lindblad master equation \eqref{eq:time_dep_lindblad} evolve with time. 
\begin{prp} \label{prop:eigenvalue evolution under Lindbladian}
Consider any (finite) $d \times d$ Hermitian matrix $X$ which is evolved by the Lindbladian $\mathcal{L}_t = - i [H(t) , \cdot] + \mathcal{D}$ (where $H(t)$ is locally integrable) with dissipator $\mathcal{D} = \sum_j \mathcal{D}_{L_j}$, i.e. $X(t) \coloneqq \mathcal{E}_{t,s} (X)$ is a solution to \cref{eq:time_dep_lindblad} with initial condition $X(s) = X$ at time $s$. Since the Lindbladian evolution is Hermiticity-preserving, $X(t)$ will at all times be Hermitian, and we can instantaneously diagonalize it as $X(t) = \sum_{k=1}^d \lambda_k (t) \ket{e_k (t)} \bra{e_k (t)} $, where $\{ \lambda_k (t) \}_{k=1}^d$ and $\{ \ket{e_k (t)} \}_{k=1}^d$ are respectively the sets of eigenvalues and associated instantaneous orthonormal eigenvectors of $X(t)$. Let further $ \boldsymbol{\lambda}(t) \coloneqq ( \lambda_1 (t) , \ldots , \lambda_d (t) )^T$ be the vector of eigenvalues. Then, there exists a labeling of the eigenvalues for all times such that $\boldsymbol{\lambda}(t)$ evolves continuously with time, and for almost all times $t$, the time-derivative $d \boldsymbol{\lambda}(t) / dt$ is defined and given by
\begin{equation} \label{eq:prop:eigenvalue evolution Q(t) matrix eq}
\frac{d}{dt} \boldsymbol{\lambda}(t) = Q(t) \boldsymbol{\lambda}(t) , \qquad \text{where} \ \ Q_{kl}(t) \coloneqq \left\{ \begin{matrix}
\sum_j | \braket{e_k (t) | L_j | e_l (t) } |^2 & \text{if} \ k \neq l \\
- \sum_{ m \neq l}^d \sum_j | \braket{e_m (t) | L_j | e_l (t) } |^2 & \text{if} \ k=l
\end{matrix} \right. .
\end{equation}
Further, \cref{eq:prop:eigenvalue evolution Q(t) matrix eq} will have a unique solution (i.e. a unique continuous solution for which the equation is satisfied for almost all times) given by the Peano-Baker series:
\begin{equation} \label{eq:prop:eigenvalue evolution peano baker}
\boldsymbol{\lambda}(t) = \Psi (t ; s) \boldsymbol{\lambda}(s) , \qquad \text{where} \ \ \Psi (t ; s) \coloneqq \mathcal{T} \exp \left( \int_s^t d \tau Q(\tau) \right) .
\end{equation}

\end{prp}

\begin{proof} \label{proof:eigenvalue evolution equation}

To prove \cref{prop:eigenvalue evolution under Lindbladian}, consider any solution $X(t) \coloneqq \mathcal{E}_{t,s} (X)$, with $X$ being a Hermitian matrix, to the Lindblad master equation \eqref{eq:time_dep_lindblad}. Let the orthonormal eigen-decomposition of $X(t)$ be given by $X(t) =  \sum_{k=1}^d \lambda_k (t) \ket{e_k (t)} \bra{e_k (t)}$, and define $\boldsymbol{\lambda}(t) \coloneqq ( \lambda_1 (t) , \ldots , \lambda_d (t) )^T$. Since $X(t)$ is a (continuous) solution to the Lindbladian master equation \eqref{eq:time_dep_lindblad}, we know that its eigenvalues must evolve continuously in time, hence we can label the eigenvalues $\{ \lambda_k (t) \}_{k=1}^d$ at all times $t$ such that they are continuous functions of time. We also know that for almost all $t$, we have 
\begin{equation} \label{eq:proof main thm:perturbation x(t)}
X(t+\epsilon) = X(t) + \epsilon \left( - i \left[ H(t) , X(t) \right] + \mathcal{D} X(t) \right) + o(\epsilon) ,
\end{equation}
for all sufficiently small $|\epsilon| $. Next, it follows from first-order perturbation theory for finite Hermitian operators that for any (finite) Hermitian matrix $A$ with orthonormal eigen-decomposition $A = \sum_j \lambda_j \ket{e_j} \bra{e_j}$, the perturbed eigenvalues of $A+B$ satisfy (under some ordering)
\begin{equation} \label{eq:proof:main result 1st order perturb thm}
\lambda_j (A+B) = \lambda_j (A) + \braket{e_j | B | e_j} + o \left( \| B \| \right) ,
\end{equation}
where in case $A$ has a degenerate eigenspaces with projectors $P_\lambda$, the eigen-vectors $\ket{e_j}$ are chosen such as to also diagonalize $\sum_\lambda P_\lambda B P_\lambda$, see e.g. \cite[sec. 2.6]{Kato1995}. Applying \cref{eq:proof:main result 1st order perturb thm} to \cref{eq:proof main thm:perturbation x(t)} and letting $\epsilon \rightarrow 0$ gives us that for almost all times $t$, the following is satisfied
\begin{equation} \label{eq:proof main thm:eigen perturbation derivation}
\begin{aligned}
\frac{d}{dt} \lambda_k (t)  &= \left\langle e_k (t) \middle| \left( - i \left[ H(t) , X(t) \right] + \mathcal{D} X(t) \right) \vphantom{2^2} \middle| e_k (t) \right\rangle \nonumber\\
&= \sum_j \left\langle e_k (t) \middle|  L_j x(t) L_j^\dagger - \frac{1}{2} \left\{ L_j^\dagger L_j , x(t) \right\}   \middle| e_k (t) \right\rangle \\
&= \sum_{j} \sum_{l=1}^d \left( \lambda_l \left| \braket{e_k (t) | L_j | e_l (t) } \right|^2 - \lambda_k \left| \braket{e_l (t) | L_j | e_k (t) } \right|^2 \right)\nonumber\\
&= \sum_{l=1}^d Q_{kl} (t) \lambda_l (t),
\end{aligned}
\end{equation}
where we have in the second line above used $\braket{e_k (t) | [H(t) , X(t)] | e_k (t)} = H(t) \lambda_k (t) - \lambda_k (t) H(t) = 0$, and we recognize the matrix $Q(t)$ defined in \cref{eq:prop:eigenvalue evolution Q(t) matrix eq}. For \cref{eq:proof main thm:eigen perturbation derivation} to hold, we must also make sure that the instantaneous orthonormal eigenvectors $\{ \ket{e_k(t)} \}_{k=1}^d$ are chosen such that the instantaneous perturbation is diagonalized in any degenerate eigenspaces of $X(t)$, i.e. if we have $X(t) = \sum_{\lambda} \lambda P_\lambda$, we must make sure that $P_\lambda ( - i [ H(t) , X(t) ] + \mathcal{D} X(t) ) P_\lambda = P_\lambda \mathcal{D} X(t) P_\lambda $ is also diagonalized in this basis. This thus establishes \cref{eq:prop:eigenvalue evolution Q(t) matrix eq}.

It is not however immediately obvious that \cref{eq:proof main thm:eigen perturbation derivation} holds for almost all $t$, even though this is the case for \cref{eq:proof main thm:perturbation x(t)}. The potential problem is that the eigenvectors $\ket{e_k (t)}$ can change discontinuously with time, hence the perturbation arising from \cref{eq:proof:main result 1st order perturb thm} may change discontinuously as well. However, as we shall now argue, this can only happen for a set of times of measure $0$, hence \cref{eq:proof main thm:eigen perturbation derivation} will after all hold for almost all times. If some eigenvectors $\ket{e_k (t)}$ change such that $X(t)$ changes discontinuously with time, then the Lindbladian master equation \eqref{eq:time_dep_lindblad} cannot hold at such times, hence this can only happen at a set of times of measure $0$ if $X(t)$ is to be a solution, so this is not a problem. But some eigenvectors $\ket{e_k (t)}$ can change discontinuously at time $t$ while $X(t)$ still changes continuously. This can happen if the discontinuous changes are confined to degenerate eigenspaces of $X(t)$ where we might not always be able to choose continuously change eigenvectors due to the condition that the eigenvectors must at all times diagonalize $P_\lambda \mathcal{D} X(t) P_\lambda$ for any projector $P_\lambda$ onto a degenerate eigenspace of $X(t)$. However, if such a discontinuity is forced, this means that the pathological eigenvalues cannot be equal immediately before and after time $t$, since if the eigenvalues kept being equal on some open interval containing $t$, the pathological eigenvectors would have equal eigenvalues of $X(t')$ and the perturbation $\mathcal{D} X(t')$ for $t'$ in this open interval, which would result in us being able to freely choose the eigenvectors to change continuously. Since the eigenvalues of $X(t)$ change continuously, this thus means that if $t$ is a time where some degenerate eigenvectors $\ket{e_l (t)}$ change discontinuously, there will exist some open time-interval $(t-\delta , t + \delta)$ containing $t$, such that for all $t' \in (t-\delta , t + \delta)$, $t' \neq t$, the pathological eigenvectors $\ket{e_l (t')}$ correspond to distinct eigenvalues, meaning that they cannot change discontinuously for any other time than $t$ in the interval $(t-\delta , t + \delta)$. Hence, the set of times for which \cref{eq:prop:eigenvalue evolution Q(t) matrix eq} does not hold is a set of \emph{isolated points} on $\mathbb{R}$ and must have Lebesgue measure equal to zero, see e.g. \cite{Bartle1995-hg}. This finally shows that \cref{eq:prop:eigenvalue evolution Q(t) matrix eq} must hold for almost all times $t$. 

We shall now prove that \cref{eq:prop:eigenvalue evolution peano baker} offers the unique solution to \cref{eq:prop:eigenvalue evolution Q(t) matrix eq}, in the sense that it is the unique continuous solution for which the equation holds for almost all times. It follows from the definition of $Q(t)$ in \cref{eq:prop:eigenvalue evolution Q(t) matrix eq} that the components of $Q(t)$ are uniformly bounded, specifically by $\| \sum_j L_j^\dagger L_j \|_\infty$. Since \cref{eq:prop:eigenvalue evolution Q(t) matrix eq} is also linear in $\boldsymbol{\lambda} (t)$, it follows that the differential equation satisfies the \emph{Carath{\'e}odory conditions} \cite{Hale2009-sf} for differential equations with discontinuous right-hand sides, and it thus has a unique solution. Meanwhile, since $Q(t)$ is bounded uniformly and also integrable, since it is continuous for almost all times, the \emph{Peano-Baker series} is a continuous solution to \cref{eq:prop:eigenvalue evolution Q(t) matrix eq} which holds for almost all times, see \cite{Baake2011-xi}. Hence, the unique solution to \cref{eq:prop:eigenvalue evolution Q(t) matrix eq} is indeed given by \cref{eq:prop:eigenvalue evolution peano baker}.
\end{proof}

\subsection{Sufficiency of the algebraic conditions in \texorpdfstring{\cref{thm:complete_characterization}}{thm:complete characterization}}

Recall that \cref{thm:complete_characterization} claimed that the dissipator $\mathcal{D}  = \sum_j \mathcal{D}_{L_j}$ is Hamiltonian independently contractive if and only if there are no two orthogonal subspaces of the Hilbert space, $V_1, V_2$ such that $L_j V_1 \subseteq V_1$ and $L_j V_2 \subseteq V_2$ for all $j$. Here, we show the sufficiency of this condition.

Assume that this condition is satisfied, and let $\mathcal{L}_t = - i [H(t) , \cdot ] + \mathcal{D}$ be a Lindbladian with an arbitrary time-dependent Hamiltonian $H(t)$. We shall show that there will then exist constants $\varepsilon , T > 0$ such that for all quantum states $\rho , \sigma$ and times $t \geq s + T$, the time-evolution operator $\mathcal{E}_{t,s}$ generated by $\mathcal{L}_t$ will for any choice of $H(t)$ satisfy 
\begin{equation} \label{eq:proof:main thm:final statement 1}
\left\| \mathcal{E}_{t,s} (\rho - \sigma) \right\|_1 \leq (1 - \varepsilon) \left\| \rho - \sigma  \right\|_1 .
\end{equation}
Applying \cref{eq:proof:main thm:final statement 1} iteratively and using $\lfloor x \rfloor \leq x-1$ gives us that for arbitrary times $t \geq s \geq 0$, we have 
\begin{equation} \label{eq:proof:main thm:final statement 2}
\begin{aligned}
\left\| \mathcal{E}_{t,s} (\rho - \sigma) \right\|_1 & \leq  \left( 1 - \varepsilon \right)^{\lfloor \frac{t-s}{T} \rfloor} \left\| \rho - \sigma  \right\|_1 \\
& \leq \frac{1}{1-\varepsilon} e^{- \frac{\ln ( 1 / ( 1 - \varepsilon) )}{T} (t-s)} \left\| \rho - \sigma  \right\|_1 ,
\end{aligned}
\end{equation}
thus establishing Hamiltonian-independent contractivity with universal constants $K = 1 / (1 - \varepsilon) \geq 1$ and $\gamma = - \ln ( 1 - \varepsilon) / T > 0$ in the sense of \cref{def:Hamiltonian independent contractivity}.

To establish \cref{eq:proof:main thm:final statement 1}, we shall show that we can in any basis always find some directed spanning tree with a global source such that the dissipator induces a transmission along all edges in the spanning tree with a non-zero rate bounded uniformly below away from $0$. To be more specific, we define the following number $\delta (\mathcal{D})$ which we shall prove to be strictly positive 
\begin{equation} \label{eq:proof:main thm:def delta}
\delta (\mathcal{D}) \coloneqq \min_{U \in \mathcal{U}(d)} \max_{\tau \in \mathcal{S}_d} \min_{ (i,j) \in E(\tau) } A (\mathcal{D} , U )_{ij}   ,
\end{equation}
where the matrix $A(\mathcal{D} , U )$ is defined above in \cref{eq:def:A matrix dissipator}, $\tau \in \mathcal{S}_d$ denotes a directed spanning tree with a global source whose vertices are labelled by $1 , \ldots , d$, $E(\tau)$ is the edge-set of $\tau$ where we write $(i,j) \in E(\tau)$ if $i \rightarrow j$ is an edge in $\tau$. The fact that the (non-finite) minimum over $U \in \mathcal{U}(d)$ in \cref{eq:proof:main thm:def delta} exists follows from the set of $d \times d$ unitary matrices $\mathcal{U}(d)$ being compact and from the function $U \mapsto \max_{\tau \in \mathcal{S}_d} \min_{ (i,j) \in E(\tau) } A (\mathcal{D} , U )_{ij}$ being continuous. The fact that this function is continuous follows from it being a finite maximum over a finite minimum of continuous functions -- namely $A (\mathcal{D} , U )_{ij}$, which are continuous in $U$ for all $i,j$ as seen from the definition of $A (\mathcal{D} , U )$ in \cref{eq:def:A matrix dissipator}. 

To show that $\delta (\mathcal{D}) > 0$, it suffices to show that $\delta (\mathcal{D}) \neq 0$, since it follows straightforwardly from the definition of $A (\mathcal{D}, U)$ in \cref{eq:def:A matrix dissipator} that we must have $\delta (\mathcal{D}) \geq 0$. Assume thus for contradiction that we had $\delta (\mathcal{D}) = 0$. By \cref{eq:proof:main thm:def delta} and the definition of the minimum, it would follow that there exists some $U \in \mathcal{U}(d)$ such that $\max_{\tau \in \mathcal{S}_d} \min_{ (i,j) \in E(\tau) } A (\mathcal{D} , U )_{ij} = 0$, which would in turn imply that $\min_{ (i,j) \in E(\tau) } A (\mathcal{D} , U )_{ij} = 0$ holds for all $\tau \in \mathcal{S}_d$, since we again have $A (\mathcal{D} , U )_{ij} \geq 0$. This entails that the graph $G(\mathcal{D} , U)$ defined in \cref{eq:def of dissipator graph} cannot have a global source, since if $G(\mathcal{D} , U)$ had a global source, then $G(\mathcal{D} , U)$ would also contain a directed spanning tree with a global source $\tau \in \mathcal{S}_d$, which by definition of $G(\mathcal{D} , U)$ would mean $A(\mathcal{D} , U )_{ij} > 0$ whenever $(i,j) \in E(\tau)$. 
We now use a specific graph theoretic result stating that a directed graph contains a global source if and only if it does not contain two disjoint non-empty vertex sets with no inflowing edges, see \cref{prop:graph theory global source condition} in \cref{app:subsec graph theory}. This result now entails that since $G(\mathcal{D} , U)$ has no global source, it must contain two disjoint non-empty vertex sets, call them $W_1 , W_2 \subset v_1 , \ldots , v_d$ with no inflowing edges. By definition of $G(\mathcal{D} , U)$, this means that we must have $A(\mathcal{D} , U)_{ik} = 0$ whenever $v_i \not\in W_1$ and $v_k \in W_1$, and have $v_i \not\in W_2$ and $v_k \in W_2$. By definition of $A(\mathcal{D} , U)$ in \cref{eq:def:A matrix dissipator}, this entails that we have $\braket{i | U L_j U^\dagger | k } = 0$ for all $j$ whenever $v_i \not\in W_1$ and $v_k \in W_1$, or $v_i \not\in W_2$ and $v_k \in W_2$. However, this means that the two subspaces $V_1 , V_2 \subseteq \mathbb{C}^d$ defined as
\begin{equation} \label{eq:proof:main thm:def of V1V2}
\begin{aligned}
V_1 \coloneqq \spn \left( \left\{ U^\dagger \ket{k} \middle| v_k \in W_1 \right\} \right) , \\ V_2 \coloneqq \spn \left( \left\{ U^\dagger \ket{k} \middle| v_k \in W_2 \right\} \right) ,
\end{aligned}
\end{equation}
must be invariant under all of $\{ L_j \}_j$, $L_j V_1 \subseteq V_1$ and $L_j V_2 \subseteq V_2$, since $\{ U^\dagger \ket{k} \}_{k=1}^d$ forms an orthonormal basis for $\mathbb{C}^d$. Since $W_1$ and $W_2$ are non-empty and disjoint, the subspaces $V_1$ and $V_2$ defined in \cref{eq:proof:main thm:def of V1V2} above must be non-trivial and orthogonal. This contradicts the initial assumption on $\mathcal{D}$ that there exist no two non-trivial orthogonal subspaces that are invariant under all of $\{ L_j \}_j$, hence we conclude that we must have $\delta (\mathcal{D}) > 0$. 

To proceed with the proof and establish \cref{eq:proof:main thm:final statement 1}, consider two arbitrary quantum states $\rho , \sigma$, and define the operator $x(t)$ by $x(t) \coloneqq \mathcal{E}_{t,s} (\rho - \sigma)$, where $\mathcal{E}_{t,s}$ was the time-evolution operator generated by $\mathcal{L}_t = - i [H(t) , \cdot ] + \mathcal{D}$. Further, since any Lindbladian is Hermiticity-preserving, we can at all times $t$ diagonalize $x(t)$ by an orthonormal basis $\{ \ket{e_k (t)} \}_{k=1}^d$ and let $\{ \lambda_k (t) \}_{k=1}^d$ be the corresponding eigenvalues. 
\begin{equation} \label{eq:proof:main thm:def of x(t)}
x(t) \coloneqq \mathcal{E}_{t,s} (\rho - \sigma) = \sum_{k=1}^d \lambda_k (t) \ket{e_k (t)} \bra{e_k (t)} .
\end{equation}
We now focus on the dynamics of the vector of eigen-values $\boldsymbol{\lambda}(t) \coloneqq ( \lambda_1 (t) , \ldots , \lambda_d (t) )^T$. Applying \cref{prop:eigenvalue evolution under Lindbladian}, we obtain that
\begin{subequations} \label{eq:proof main thm:diff eq for Q(t)}
\begin{align}
\frac{d}{dt} \boldsymbol{\lambda}(t) = Q(t) \boldsymbol{\lambda}(t),
\end{align}
with $Q(t)$ being a matrix given by 
\begin{align}
Q(t)_{kl} \coloneqq \left\{ \begin{matrix}
\sum_j | \braket{e_k (t) | L_j | e_l (t) } |^2 & \text{if} \ k \neq l \\
- \sum_{ m \neq l}^d \sum_j | \braket{e_m (t) | L_j | e_l (t) } |^2 & \text{if} \ k=l
\end{matrix} \right. .
\end{align}
\end{subequations}
As is also proven by \cref{prop:eigenvalue evolution under Lindbladian}, since $Q(t)$ is uniformly bounded, the linear differential equation in \cref{eq:proof main thm:diff eq for Q(t)} satisfies the \emph{Carath{\'e}odory conditions} \cite{Hale2009-sf}, and its unique continuous solution holding for almost all times with initial condition $\boldsymbol{\lambda}(s)$ at $t=s$ is given by the \emph{Peano-Baker series} \cite{Baake2011-xi} :
\begin{equation} \label{eq:proof main thm: peano baker 1}
\begin{aligned}
 \boldsymbol{\lambda} (t) = \Phi (t;s) \boldsymbol{\lambda} (s) \quad \text{with}\quad  \Phi (t;s) \coloneqq \mathcal{T} \exp \left( \int_s^t d \tau Q(\tau) \right) ,
\end{aligned}
\end{equation}
where $\mathcal{T}$ is again the time-ordering operator. The remaining proof will now consist of showing that the result $\delta (\mathcal{D}) > 0$ established earlier entails that for all sufficiently large time-intervals $(s+T,s)$, the time-evolution matrix $\Phi (s+T;s)$ from \cref{eq:proof main thm: peano baker 1} will have all non-negative components and contain at least one row with all its elements positive and bounded away from $0$. This will show that some mixing of positive and negative eigenvalues of $x(t)$ must take place, reducing its $1$-norm.

Define the new matrix $M(t) \coloneqq Q(t) + \Gamma I$, where $\Gamma$ is the constant $\Gamma \coloneqq \| \sum_j L_j^\dagger L_j \|_\infty$. Since the components of $Q(t)$ have their norm bounded by $\Gamma$ and all its non-diagonal components are non-negative as follows from the definition in \cref{eq:proof main thm:diff eq for Q(t)}, the matrix $M(t)$ has all non-negative components. Next, since the time-ordering operator $\mathcal{T}$ commutes with scalars, and since $\Gamma I$ commutes with all matrices, we can rewrite the expression for $\Phi (t;s)$ from \cref{eq:proof main thm: peano baker 1} as
\begin{equation} \label{eq:proof main thm: peano baker 2}
\begin{aligned}
& \Phi (t;s) = e^{- \Gamma (t-s)} \left( I + \int_s^t d \tau M(\tau) + \int_s^t d \tau_1  \int_s^{\tau_1} d \tau_2  M(\tau_1)M(\tau_2)  + \ \cdots \  \right) .
\end{aligned}
\end{equation}
Since the matrix $M(t)$ has all non-negative components for all times $t$ as explained above, we immediately see from \cref{eq:proof main thm: peano baker 2} that the matrix $\Phi (t;s)$ has all non-negative components for all $t \geq s \geq 0$, since all its components are sums of integrals over non-negative functions. Also, it follows from the definition in \cref{eq:proof main thm:diff eq for Q(t)} of $Q(t)$ that $\sum_{k=1}^d Q_{kl} (t) = 0$ holds for all $t$ and $l$. Using this in the formula for $\Phi (t;s)$ in \cref{eq:proof main thm: peano baker 1}, we see that the only term in the Peano-Baker expansion which contributes to the column sums of $\Phi (t;s)$ is $I$, hence $\Phi (t;s)$ has column sum equal to $1$. Thus, for all $t \geq s \geq 0$, we have
\begin{equation} \label{eq:proof main thm:properties of Phi}
\begin{aligned}
(\Phi (t;s))_{kl} \geq 0 , \ \text{and} \ \sum_{i=1}^d (\Phi (t;s))_{kl} = 1 \ \text{for all } k,l,j .
\end{aligned}
\end{equation}
Once we also show that for sufficiently large time-intervals $[t,s]$ with $t = T +s $, there exists some index $1 \leq j \leq d$ for which the associated row of $\Phi (t;s)$ has components lower bounded by $\varepsilon > 0$, i.e. $\Phi (t;s)_{jk} \geq \varepsilon$ for all $k$, the desired result from \cref{eq:proof:main thm:final statement 1} will follow by \cref{eq:proof main thm:properties of Phi}, \cref{eq:proof:main thm:def of x(t)} and the fact that the Lindbladian evolution cannot increase the $1$-norm between two quantum states, since
\begin{equation} \label{eq:proof main thm:how index leads to contract}
\begin{aligned}
\left\| \mathcal{E}_{t,s} (\rho - \sigma) \right\|_1 &= \sum_{k=1}^d \left| \lambda_k (t) \right|\nonumber\\
&= \sum_{k=1}^d \left| \sum_{l=1}^d \Phi (t;s)_{kl} \lambda_l (s) \right| \\
&= \sum_{k=1}^d \left| \sum_{l=1}^d \left( \Phi (t;s)_{kl} - \varepsilon \delta_{kj} \right) \lambda_l (s) \right| \\
&\leq \sum_{l, k=1}^d \left|   \Phi (t;s)_{kl} - \varepsilon \delta_{kj}  \right| \left| \lambda_l (s) \right| \\
&= \sum_{l=1}^d (1-\varepsilon) \left| \lambda_l (s) \right| = (1 - \varepsilon ) \left\| \rho - \sigma \right\|_1 ,
\end{aligned}
\end{equation}
thus completing the proof of \cref{thm:complete_characterization}. In the derivation above, we used that $\Phi (t;s)_{kl} - \varepsilon \delta_{kj}$ has all non-negative components by assumption and constant column sum equal to $1-\varepsilon$ and that $\sum_{l=1}^d \lambda_l (t) = 0$ holds for all $t$, since the Lindbladian evolution is trace-preserving. All that is left to do is to show that there exists some index $j$, such that $\Phi (t;s)_{jk} \geq \varepsilon$ holds for all $k$.

To prove this, note first that it follows from the definitions in \cref{eq:def:A matrix dissipator} and \cref{eq:proof main thm:diff eq for Q(t)} of the matrices $A(U , \mathcal{D})$ and $Q( t )$, that there for all times $\tau$ exists some unitary $U( t )$ such that $Q( t )$ and $A(U( t ) , \mathcal{D})$ have the same off-diagonal components. Specifically, we choose $U( t )$ such that $U(t)^\dagger \ket{k} = \ket{e_k ( t )} \Leftrightarrow U(t)_{kl} = \braket{e_k ( t ) | l }$ holds for all $k,l$, i.e. $U(t)$ transforms the eigen-basis of $x(t)$ into the standard basis. Thus, since $M( t ) \coloneqq Q( t ) + \Gamma I$ has all non-negative components, we must for all $\tau$ have
\begin{equation} \label{eq:proof main thm:matrix M A relation}
M( t ) = A(U( t ) , \mathcal{D}) + D( t ) ,
\end{equation}
where $D( t )$ is some time-dependent diagonal matrix with non-negative components. It now follows from \cref{eq:proof main thm:matrix M A relation} and the definition \eqref{eq:proof:main thm:def delta} of $\delta (\mathcal{D})$ and the result $\delta (\mathcal{D}) > 0$ proven above, that for all times $t \in [s , s + T]$, there exists some directed spanning tree with a global source $\tau_t \in \mathcal{S}_d$ such that $M(t)_{kl} \geq \delta (\mathcal{D}) > 0$ holds whenever $(k,l) \in E(\tau_t ) $, i.e. whenever $k \rightarrow l$ is a directed edge in $\tau_t $ (where we associate vertices with the labels $1, \ldots , d$). 

By the pigeonhole principle, or the measure-theoretic analogue of it, there must then exist at least one directed spanning tree with a global source, call it $\tau^{*} \in \mathcal{S}_d$, such that the set of times $\mathcal{I}^{*} \subseteq [s,s+T]$ for which the above property holds for $\tau^{*} \in \mathcal{S}_d$ has integral measure $\mu (\mathcal{I}^{*}) \coloneqq \int_{t \in \mathcal{I}^{*}} dt $ equalling at least $\mu (\mathcal{I}^{*}) \geq T / \Xi_d$, where we recall that $\Xi_d$ is the number of distinct directed spanning trees with a global source on $d$ labelled vertices. Thus, for all $t \in \mathcal{I}^{*}$, we have $M(t)_{kl} \geq \delta (\mathcal{D})$ whenever $k \rightarrow l$ is an edge in $\tau^{*}$. We now let $j_s \in \{ 1 , \ldots , d \}$ denote the label of the vertex corresponding to the global source in $\tau^{*} \in \mathcal{S}_d$, and we shall now show that $\Phi (s+T;s)_{j_s k} \geq \varepsilon$ holds for all $k$ for some $\varepsilon > 0$, thus completing the proof by the earlier argument. First, we note that since all terms in the final expression from \cref{eq:proof main thm: peano baker 2} consist of matrices with all non-negative components, we get the following lower bound by only focusing on the contribution from the identity term $I$ in the expansion
\begin{equation} \label{eq:proof main thm:jj term}
\Phi (s+T;s)_{j_s j_s }  \geq e^{- \Gamma T} I_{j_s j_s } = e^{- \Gamma T} . 
\end{equation}
Next, consider any $k \neq j_s$. By assumption on $\tau^{*} \in \mathcal{S}_d$, there must exist a directed path $j_s = v_1  \rightarrow v_2 \rightarrow \cdots \rightarrow v_{L-1} \rightarrow v_L = k$ in $\tau^{*}$ going from $j_s$ to $k$, which must then have length $L \in \{ 1, \ldots , d-1 \}$. Since $M(t)$ has all non-negative components, and since $M(t)_{lm} \geq \delta (\mathcal{D})$ whenever $l \rightarrow m$ is an edge in $\tau^{*}$ for all $t \in \mathcal{I}^{*}$, we must have $\forall {t_1 , t_2 \ldots  t_L \in \mathcal{I}^{*}}$
\begin{equation} \label{eq:proof main thm:matrix products lower bound}
\begin{aligned}
\left( M(t_1) M (t_2) \cdots M(t_L) \right)_{j_s k} 
\geq M (t_1)_{j_s v_2} \cdots M (t_L)_{v_{L-1} k} \geq \delta (\mathcal{D})^L \quad,
\end{aligned}
\end{equation}
which follows from only including the contribution from the components in the path $j_s = v_1  \rightarrow v_2 \rightarrow \cdots \rightarrow v_{L-1} \rightarrow v_L = k$. Using \cref{eq:proof main thm:matrix products lower bound}, we can lower bound $\Phi (s+T;s)_{j_s k}$ by only including the contribution from the $L$'th term in the Peano-Baker expansion and only including the integral contributions from times in the interval $\mathcal{I}^{*}$, which again follows from $M(t)$, and hence all mixed-time products thereof, having non-negative components at all times
\begin{equation} \label{eq:proof main thm:final peano baker}
\begin{aligned}
\Phi (s+T;s)_{j_s k} &\geq e^{- \Gamma T} \left( \mathcal{T} \exp \left( \int_s^{s+T} d \tau  M(\tau) \right) \right)_{j_s k} \\
&\geq e^{- \Gamma T} \left( \int_s^{s+T} dt_1 M(t_1)  \cdots \int_s^{t_{L-1}} d t_L M(t_L) \right)_{j_s k} \\
&\geq e^{ - \Gamma T} \int_{t_1 \geq \cdots \geq t_L \in { \mathcal{I}^{*} }^{\times L} } dt_1 \cdots dt_L
\left( M(t_1) \cdots M(t_L) \right)_{j_s k} \\
&\geq e^{ - \Gamma T} \delta (\mathcal{D})^L \int_{t_1 \geq \cdots \geq t_L \in { \mathcal{I}^{*} }^L } dt_1 \cdots dt_L \\
&= e^{ - \Gamma T} \delta (\mathcal{D})^L \frac{1}{L !} \int_{t_1 \in \mathcal{I}^{*}} \cdots \int_{t_L \in \mathcal{I}^{*}} dt_1 \cdots dt_L \\
&= e^{ - \Gamma T} \delta (\mathcal{D})^L \frac{1}{L !} \left( \mu \left( \mathcal{I}^{*} \right) \right)^L \geq e^{ - \Gamma T} \frac{1}{L !} \left( \frac{ \delta (\mathcal{D}) T}{\Xi_d} \right)^L .
\end{aligned}
\end{equation}
Combining the results of \cref{eq:proof main thm:final peano baker} and \cref{eq:proof main thm:jj term} and using that we must have $L \in \{ 1 , \ldots , d-1 \}$ proves that there exists some index $j_s$ such that $\Phi (s+T;s)_{j_s k} \geq \varepsilon$ holds for all $k$ for some $\varepsilon > 0$, which finally proves \cref{eq:proof:main thm:final statement 1} via \cref{eq:proof main thm:how index leads to contract} as explained above. This thus proves Hamiltonian-independent contractivity in the form of \cref{eq:proof:main thm:final statement 2}, where the universal constants $T , \varepsilon$ can be chosen as $T > 0$ being arbitrary while $\varepsilon$ equals
\begin{equation} \label{eq:proof main thm:explicit bound epsilon}
\varepsilon = e^{ - \Gamma T} \min_{ 0 \leq L \leq d-1} \frac{1}{L !} \left( \frac{\delta (\mathcal{D} )  T}{\Xi_d} \right)^L .
\end{equation}
To get an explicit Hamiltonian-independent contraction rate, we choose $T = (d-1) / \Gamma$. Then, using the result in \cref{eq:exact number of spanning trees} and the fact that we always have $\Gamma \geq \delta (\mathcal{D})$, we get that for all quantum states $\rho , \sigma$ and times $t \geq s \geq 0$
\begin{subequations} \label{eq:proof main thm:optimal rate expression}
\begin{align}
\left\| \mathcal{E}_{t,s} (\rho - \sigma) \right\|_1 \leq K e^{ - \gamma (t-s)} \left\| \rho - \sigma \right\|_1,
\end{align}
where
\begin{align}
K &= \left( 1 - \frac{ (d-1)^{d-1} \left( \delta (\mathcal{D}) / \Gamma \right)^{d-1}}{e^{d-1} (d-1)! d^{(d-1)^2}} \right)^{-1} \geq 1 , \\
\gamma &= - \ln \left( 1 - \frac{ (d-1)^{d-1} \left( \delta (\mathcal{D}) / \Gamma \right)^{d-1}}{e^{d-1} (d-1)! d^{(d-1)^2}} \right)  \frac{\Gamma}{d-1} > 0 ,
\end{align}
\end{subequations}
where we recall that $\Gamma \coloneqq \|  \sum_j L_j^\dagger L_j \|_\infty$, and $\delta (\mathcal{D}) > 0$ is defined in \cref{eq:proof:main thm:def delta} and must be strictly positive if the condition from \cref{thm:complete_characterization} was satisfied. 

\subsection{Necessary condition for Hamiltonian-independent contractivity} \label{app:subsec:necessary condition for HI contract}

We next prove that the algebraic conditions in \cref{thm:complete_characterization} are in fact necessary for Hamiltonian-independent contractivity. In particular, we show that if there are two non-trivial orthogonal subspaces that are separately invariant under the jump operators, then there is a \emph{time-independent} Hamiltonian that can be applied to the system to create at least two fixed points.

Denote the dissipator by $\mathcal{D} = \sum_j \mathcal{D}_{L_j}$, and suppose that $V_1 , V_2 \neq \{ 0 \} , \mathbb{C}^d$ are two non-trivial orthogonal subspaces of $\mathbb{C}^d$ that are invariant under all the jump operators, i.e. we have $V_1 \perp V_2$ and for all $j$, $L_j V_1 \subseteq V_1$ and $L_j V_2 \subseteq V_2 \ \forall j$. Let now $V_c$ be the orthogonal complement to $V_1 \oplus V_2$ (which might be trivial), such that  $\mathbb{C}^d = V_1 \oplus V_2 \oplus V_c$. Then, we can choose some orthonormal basis in which every vector in $\mathbb{C}^d$ can be written as $ ( v_1^T , v_2^T , v_c^T 
)^T$, where $v_1 , v_2 , v_c $ are components of vectors in the subspaces $V_1 , V_2 , V_c$ respectively. In this basis, the projectors $P_1 , P_2 , P_c$ onto the subspaces $V_1 , V_2 , V_c$ respectively, and all the jump operators must take the following form if both $V_1$ and $V_2$ are to be invariant subspaces:
\begin{equation} \label{eq:restr jump forms}
L_j = \begin{pmatrix}
N_j & \mathbf{0} & S_j \\
\mathbf{0} & M_j & R_j \\
\mathbf{0} & \mathbf{0} & Q_j \\
\end{pmatrix} , \quad L_j^\dagger = \begin{pmatrix}
N_j^\dagger & \mathbf{0} & \mathbf{0} \\
\mathbf{0} & M_j^\dagger & \mathbf{0} \\
S_j^\dagger & R_j^\dagger & Q_j^\dagger \\
\end{pmatrix} , 
\quad P_1 = \begin{pmatrix}
I & \mathbf{0} & \mathbf{0} \\
\mathbf{0} & \mathbf{0} & \mathbf{0} \\
\mathbf{0} & \mathbf{0} & \mathbf{0} \\
\end{pmatrix} , P_2 = \begin{pmatrix}
\mathbf{0} & \mathbf{0} & \mathbf{0} \\
\mathbf{0} & I & \mathbf{0} \\
\mathbf{0} & \mathbf{0} & \mathbf{0} \\
\end{pmatrix} , 
P_c = \begin{pmatrix}
\mathbf{0} & \mathbf{0} & \mathbf{0} \\
\mathbf{0} & \mathbf{0} & \mathbf{0} \\
\mathbf{0} & \mathbf{0} & I \\
\end{pmatrix}
\end{equation}
where each $N_j , M_j , Q_j , R_j , S_j$ can be arbitrary complex matrices with the right dimensions (the dimensions are determined by the dimensions of $V_1 , V_2$ and $V_c$), and $\mathbf{0}$ are all zero-matrices with the right dimensions. In case $V_c = \{ 0 \}$, we neglect the bottom and right-most blocks associated with $V_c$ in what follows. Consider now arbitrary quantum states $\rho_1 , \rho_2$ of the following form:
\begin{equation} \label{eq:app:proof necessary 2 state}
\rho_1 = \begin{pmatrix}
\sigma_1 & \mathbf{0} & \mathbf{0} \\
\mathbf{0} & \mathbf{0} & \mathbf{0} \\
\mathbf{0} &\mathbf{0} & \mathbf{0} \\
\end{pmatrix} , \quad \rho_2 = \begin{pmatrix}
\mathbf{0} & \mathbf{0} & \mathbf{0} \\
\mathbf{0} & \sigma_2 & \mathbf{0} \\
\mathbf{0} &\mathbf{0} & \mathbf{0} \\
\end{pmatrix}  ,
\end{equation}
where $\sigma_1 , \sigma_2$ are arbitrary quantum states of the right dimensions (which exist since neither $V_1$ nor $V_2$ are trivial). We now calculate the action that the dissipator $\mathcal{D} = \sum_j ( L_j \cdot L_j^\dagger - \frac{1}{2} \{  L_j^\dagger L_j , \cdot \} ) $ has on both $\rho_1$ and $\rho_2$. Using the expressions from Eq. \eqref{eq:restr jump forms} and \eqref{eq:app:proof necessary 2 state}, it can be directly verified that we get:
\begin{equation} \label{eq:app:necessary contraction cond action}
\begin{aligned}
\mathcal{D}(\rho_1) & = \begin{pmatrix}
\sum_j ( N_j \sigma_1 N_j^\dagger - \{ N^\dagger_j N_j , \sigma_1 \} / 2 ) & \mathbf{0} & -\sigma_1 \sum_j N_j^\dagger S_j / 2 \\
\mathbf{0} & \mathbf{0} & \mathbf{0} \\
-( \sum_j S_j^\dagger N_j / 2 ) \sigma_1 & \mathbf{0} & \mathbf{0} \\
\end{pmatrix} \coloneqq \begin{pmatrix}
\mathcal{D}_1 (\sigma_1) & \mathbf{0} & -\sigma_1 h_1 \\
\mathbf{0} & \mathbf{0} & \mathbf{0} \\
-h_1^\dagger \sigma_1 & \mathbf{0} & \mathbf{0} \\
\end{pmatrix} \ , \\
\mathcal{D}(\rho_2) & = \begin{pmatrix}
\mathbf{0} & \mathbf{0} & \mathbf{0} \\
\mathbf{0} & \sum_j ( M_j \sigma_2 M_j^\dagger - \{ M^\dagger_j M_j , \sigma_2 \} / 2 ) & -\sigma_2 \sum_j M_j^\dagger R_j / 2 \\
\mathbf{0} & -( \sum_j R_j^\dagger M_j / 2 ) \sigma_2 & \mathbf{0} \\
\end{pmatrix} \coloneqq \begin{pmatrix}
\mathbf{0}  & \mathbf{0} & \mathbf{0} \\
\mathbf{0} & \mathcal{D}_2 (\sigma_2) & -\sigma_2 h_2 \\
\mathbf{0}  & -h_2^\dagger \sigma_2 & \mathbf{0} \\
\end{pmatrix} ,
\end{aligned}
\end{equation}
where we have defined:
\begin{equation} \label{eq:app:necessary contraction defs}
\begin{aligned}
\mathcal{D}_1 (\sigma_1) & \coloneqq \sum_j ( N_j \sigma_1 N_j^\dagger - \{ N_j^\dagger N_j , \sigma_1 \} / 2 ) , \quad h_1 \coloneqq \sum_j N_j^\dagger S_j / 2 , \quad \text{and} \\
\mathcal{D}_2 (\sigma_2) & \coloneqq \sum_j ( M_j \sigma_2 M_j^\dagger - \{ M_j^\dagger M_j , \sigma_2 \} / 2 ) , \quad h_2 \coloneqq \sum_j M_j^\dagger R_j / 2 .
\end{aligned}
\end{equation}
Now, since $\mathcal{D}_1$ and $\mathcal{D}_2$ are themselves both finite dissipators on the Hilbert spaces $V_1$ and $V_2$, they must both have at least one fixed-point \cite[prop. 5]{Baumgartner2008-nb}. We denote these fixed-points by $\sigma_{\infty , 1}$ and $\sigma_{\infty , 2}$ respectively, and thus: $\mathcal{D}_1 (\sigma_{\infty , 1}) = \mathbf{0}$ and $\mathcal{D}_2 (\sigma_{\infty , 2}) = \mathbf{0}$. We lastly define the following quantum states $\rho_{\infty , 1}$ and $\rho_{\infty , 2}$ on $\mathbb{C}^d$ along with the Hamiltonian $H$ as the following Hermitian matrix
\begin{equation} \label{eq:app:necessary contraction defs H}
H \coloneqq \begin{pmatrix}
\textbf{0} & \textbf{0} & - i h_1 \\
\textbf{0} & \textbf{0} & - i h_2 \\
i h_1^\dagger & i h_2^\dagger & \textbf{0} \\
\end{pmatrix} , \quad \rho_{\infty , 1} \coloneqq \begin{pmatrix}
\sigma_{\infty , 1 } & \textbf{0} & \textbf{0} \\
\textbf{0} & \textbf{0} & \textbf{0} \\
\textbf{0} & \textbf{0} & \textbf{0} \\
\end{pmatrix} , \quad \rho_{\infty , 2} \coloneqq \begin{pmatrix}
\textbf{0} & \textbf{0} & \textbf{0} \\
\textbf{0} & \sigma_{\infty , 2 } & \textbf{0} \\
\textbf{0} & \textbf{0} & \textbf{0} \\
\end{pmatrix} ,
\end{equation}
where $h_1$ and $h_2$ were defined in Eq. \eqref{eq:app:necessary contraction defs}. Using Eq. \eqref{eq:app:necessary contraction cond action}, we then note that the Lindbladian $\mathcal{L} \coloneqq -  i[H , \cdot ] + \mathcal{D}$ has both $\rho_{\infty , 1}$ and $\rho_{\infty , 2}$ as fixed-points: 
\begin{equation*}
\begin{aligned}
\mathcal{L} \rho_{\infty , 1} = -i [H , \rho_{\infty , 1}] + \mathcal{D} \rho_{\infty , 1} = \begin{pmatrix}
\mathcal{D}_1 (\sigma_{\infty , 1}) & \textbf{0} & -i ( - \sigma_{\infty , 1} ) ( - i h_1 ) - \sigma_{\infty , 1} h_1 \\
\textbf{0} & \textbf{0} & \textbf{0} \\
-i (i h_1^\dagger) \sigma_{\infty , 1} - h_1^\dagger \sigma_{\infty , 1} & \textbf{0} & \textbf{0} \\
\end{pmatrix} = \mathbf{0} , \\
\mathcal{L} \rho_{\infty , 2}  = -i [H , \rho_{\infty , 2}] + \mathcal{D} \rho_{\infty , 2} = \begin{pmatrix}
\mathbf{0}  & \mathbf{0} & \mathbf{0} \\
\mathbf{0} & \mathcal{D}_2 (\sigma_{\infty , 2}) & -i (-\sigma_{\infty , 2}) (- i h_2) -\sigma_{\infty , 2} h_2 \\
\mathbf{0}  & -i (i h_2^\dagger ) \sigma_{\infty , 2} -h_2^\dagger \sigma_{\infty , 2} & \mathbf{0} \\
\end{pmatrix} = \mathbf{0} .
\end{aligned}
\end{equation*}
We have thus shown that if there exist orthogonal non-trivial mutually invariant subspaces of all $\{ L_j \}_j$, then we can always find a constant Hamiltonian $H$ such that the full Lindbladian $\mathcal{L} = - i [H , \cdot ] + \mathcal{D}$ has multiple fixed points and thus cannot be contractive. Further, it follows from the expression in Eq. \eqref{eq:app:necessary contraction defs H} that we can obtain the following simpler basis-independent expression for the Hamiltonian $H$ that results in multiple fixed points:
\begin{equation*}
\begin{aligned}
H = \frac{i}{2} \sum_j \left[  L_j^\dagger L_j , P_1 + P_2  \right] = \frac{i}{2} \sum_j \left[ P_c ,  L_j^\dagger L_j   \right] .
\end{aligned}
\end{equation*}
\subsection{Some corollaries of \texorpdfstring{\cref{thm:complete_characterization}}{thm:complete characterization}}
\begin{cor} \label{cor:immediate corollary of main thm}
For any (finite dimensional) dissipators, the following statements regarding Hamiltonian-independent contractivity hold:
\begin{enumerate}
    \item Allowing for time-dependent Hamiltonians does not qualitatively change anything. Specifically, a dissipator $\mathcal{D}$ is Hamiltonian-independent contractive if and only if all Lindbladians $- i [H , \cdot ] + \mathcal{D}$ are contractive for all constant Hamiltonians $H$. 
    \item Hamiltonian-independent contractivity cannot be halted by the application of further dissipation. Specifically, if $\mathcal{D}_1$ and $\mathcal{D}_2$ are dissipators, and $\mathcal{D}_1$ is Hamiltonian-independently contractive, then $\mathcal{D}_1 + \mathcal{D}_2$ will also be Hamiltonian-independently contractive. 
    \item No generality was lost in requiring exponential contractivity. Specifically, a dissipator $\mathcal{D}$ is Hamiltonian-independently contractive if and only if all (time-dependent) Lindbladians of the form $\mathcal{L}_t = - i [H(t) , \cdot] + \mathcal{D}$ are asymptotically contractive, i.e. if $\lim_{t \rightarrow \infty} \| \mathcal{E}_{t,s}(\rho - \sigma) \|_1 = 0$ holds for all quantum states $\rho , \sigma$ (where we can again assume universal convergence).
    \item No generality was lost by requiring universal convergence, i.e. by requiring the existence of universal contraction coefficients. Specifically, a dissipator $\mathcal{D}$ is Hamiltonian-independently contractive if and only if all (time-dependent) Lindbladians of the form $\mathcal{L}_t = - i [H(t) , \cdot] + \mathcal{D}$ are exponentially contractive, with the contraction coefficients from \cref{def:exp contraction} allowed to depend on $H(t)$.  
\end{enumerate}
\end{cor}
Although we see from \cref{cor:immediate corollary of main thm} that allowing for time-dependent driving Hamiltonians $H(t)$ does not qualitatively change whether a dissipator is Hamiltonian-independently contractive or not, it may still be the case that a time-dependent drive can reduce the contraction-rate of a dissipator far more.

The condition from \cref{thm:complete_characterization} can be related to the much simpler conditions if we e.g. constrain the Hilbert space to have dimension $2$ or constrain the dissipator to contain a single jump operator, as seen in the following corollaries. 

\begin{cor} A $2$-dimensional dissipator $\mathcal{D} = \sum_j \mathcal{D}_{L_j}$ is Hamiltonian-independently contractive if and only if two elements from the set $\{ L_j , L_j^\dagger \}_j$ do not commute, or alternatively, if and only if either one jump operator is non-normal or two jump operators do not commute. 
\end{cor}
\begin{proof}
This follows from Theorem \ref{thm:complete_characterization} and the fact that if $\{ L_j \}_j$ admits two non-trivial orthogonal mutually invariant subspaces, then these subspaces must be the span of two orthogonal vectors. These two vectors form a basis for $\mathbb{C}^2$, and all $\{ L_j \}_j$ must then be diagonal in this basis, if each subspace is to be mutually invariant. Since all $L_j$ are diagonal in the same basis, all elements of the set $\{ L_j , L_j^\dagger \}_j$ must commute. The converse statement follows from the fact that a set of commuting normal matrices can always be mutually diagonalized in some basis, and must hence admit two non-trivial orthogonal mutually invariant subspaces in $\mathbb{C}^2$. 
\end{proof}

\begin{cor} \label{cor:single jump eigenvector condition}
Consider a (finite) dissipator $\mathcal{D} = \mathcal{D}_L$ with only a single jump operator $L$. Then, $\mathcal{D}$ is Hamiltonian-independently contractive if and only if $L$ does not have two distinct orthogonal eigenvectors. 
\end{cor}

\begin{proof}
To prove this corollary, we must show that for a dissipator $\mathcal{D}_L$ with a single jump operator $L$, the condition from \cref{thm:complete_characterization} is equivalent to $L$ having no two orthogonal eigenvectors. This follows, because if $L$ had two orthogonal eigenvectors, then their individual spans would form two non-trivial orthogonal invariant subspaces. Meanwhile, if $L$ had two non-trivial orthogonal invariant subspaces, then $L$ would have one eigenvector in each subspace, since any linear map from a (finite) complex vector-space to itself has at least one eigenvector, and these two eigenvectors must be orthogonal since they would belong to different orthogonal subspaces. 
\end{proof}

We next turn our attention to specific families of dissipators that can be proven Hamiltonian-independently contractive using \cref{thm:complete_characterization}. We first look at \emph{ladder dissipators}. 

\begin{cor}\label{cor:ladder dissipators are HI contractive} Let $\mathcal{D} = \mathcal{D}_L$ be a non-degenerate ladder dissipator, meaning that it has a single jump operator $L$ of the form $L = \sum_{k=1}^{d-1} \gamma_k \ket{k-1} \bra{k}$ with $\gamma_k \neq 0$ for all $k$, where $\{ \ket{k} \}_{k=0}^{d-1}$ is an orthonormal basis for the Hilbert space. Then $\mathcal{D}$ is Hamiltonian-independently contractive.
\end{cor}

\begin{proof}
This follows straight from \cref{cor:single jump eigenvector condition} and the fact that $L = \sum_{k=1}^{d-1} \gamma_k \ket{k-1} \bra{k}$ only has a single eigenvector - namely $\ket{0}$. This can be seen from $L$ being idempotent, $L^d = 0$, hence all its eigenvalues must equal $0$, meanwhile we have $L \ket{v} \neq 0$ unless $\ket{v} = \ket{0}$ since all $\gamma_k \neq 0$. 
\end{proof}

We can also establish results for many-body dissipators, i.e. dissipators on large tensor-product spaces which are built out of local dissipators. In these cases, we can obtain a sufficient condition for many-body dissipators to be Hamiltonian-independently contractive that only concerns their local terms. To do this, we first establish the following intermediary result, where $\alg (\mathcal{S})$ denotes the complex matrix algebra generated by some set of matrices $\mathcal{S}$ under matrix addition, multiplication and scalar multiplication. 

\begin{cor} \label{cor:generating condition}
If $\alg ( \{ L_j \}_j ) = M_d (\mathbb{C})$, then the dissipator $\mathcal{D} = \sum_j \mathcal{D}_{L_j}$ is Hamiltonian-independently contractive. 
\end{cor}

\begin{proof}
This follows immediately from \cref{thm:complete_characterization} and \emph{Burnside's theorem} for finite-dimensional complex matrix algebras  \cite{Burnside1905} (cf. e.g. \cite{Halperin1980, Lomonosov2004}), which entails that a set of matrices generates the full complex matrix algebra if and only if the set has no non-trivial mutually invariant subspace. If $\{ L_j \}_j$ has no non-trivial invariant subspaces, then the condition from \cref{thm:complete_characterization} is trivially satisfied.  
\end{proof}

\noindent Using \cref{cor:generating condition}, we obtain the following result. 

\begin{cor} \label{cor:many body general}
Consider an $N$-body quantum system $ \bigotimes_{j=1}^N \mathcal{H}_j$ with $N \in \mathbb{N}$. For each $\mathcal{J} \subseteq [N]$, let $\mathcal{D}_\mathcal{J}$ be a dissipator acting non-trivially on $\mathcal{J}$ and like the identity on $[N] \backslash \mathcal{J}$. Then the total dissipator $\mathcal{D} = \sum_{\mathcal{J} \subseteq [N]} \mathcal{D}_{\mathcal{J}}$ is Hamiltonian-independently contractive if $[N]$ can be covered by subsets $\mathcal{J} \subseteq [N]$ for which the jump operators of $\mathcal{D}_{\mathcal{J}}$ generate the full matrix algebra. 
\end{cor}

\begin{proof}
This follows from \cref{cor:generating condition}, and the fact that under the assumptions of the corollary, all pure tensor-product matrices $\bigotimes_{k=1}^N M_k$ can be generated by the jump operators from $\mathcal{D}$, since $\prod_{k=1}^N \left( \id_{[N] \backslash k} \otimes M_k \right) = \bigotimes_{k=1}^N M_k$. The pure tensor-product matrices span the space of matrices on the total system, hence the jump operators of the full many-body dissipator $\mathcal{D}$ generate the full matrix algebra. 
\end{proof}

\noindent We also get the following result as a special case of \cref{cor:many body general}. 

\begin{cor}
Let $\mathcal{D} = \sum_j \mathcal{D}_{L_j}$ be an arbitrary  dissipator satisfying $\alg ( \{ L_j \}_j ) = M_d (\mathbb{C})$. Then, the following many-body dissipator $\mathcal{D}_{N}$ is Hamiltonian-independently contractive
\begin{equation}
\mathcal{D}_{N} = \mathcal{D} \otimes I \otimes \ldots \otimes I + I \otimes \mathcal{D} \otimes I \otimes \ldots \otimes I + \ldots + I \otimes \ldots \otimes I \otimes \mathcal{D}
\end{equation}
\end{cor}

\section{Deriving strong Hamiltonian-independent contraction rates} \label{app:Hamiltonian indpendent rates}

We saw from \cref{thm:complete_characterization} that we can completely characterize which dissipators that are Hamiltonian-independently contractive in terms of an algebraic condition on the associated jump operators. Given a Hamiltonian-independently contractive dissipator $\mathcal{D} = \sum_j \mathcal{D}_{L_j}$, we might also be interested in learning its largest Hamiltonian-independent contraction rate, that is, the largest exponential contraction rate in terms of \cref{def:exp contraction} for which \cref{eq:contractivity} holds for all Lindbladians with arbitrary time-dependent Hamiltonians and dissipator $\mathcal{D}$. From the proof of \cref{thm:complete_characterization} we do technically obtain a Hamiltonian-independent contraction rate (see \cref{eq:proof main thm:explicit bound epsilon}). However, the rate derived in this proof will in general be astronomically small: Even if we in that expression upper bound the quantity $( \delta (\mathcal{D}) / \Gamma )^{d-1}$ by $1$ in order to obtain an upper bound on the contraction rate, which will in general be an extreme overestimation, we will for a dissipator acting on a Hilbert space with dimension $d=10$ still only obtain a Hamiltonian-independent contraction rate of less than $\gamma \leq 1.464 \times 10^{-83} \ \Gamma$, where $\Gamma \coloneqq \| \sum_j L_j^\dagger L_j \|_\infty$. Such rates are of little practical relevance.

While it is likely possible to somewhat improve upon the contraction rate obtained from the proof of \cref{thm:complete_characterization} using increasingly sophisticated arguments, in this section we shall take a new approach and derive new expressions for Hamiltonian-independent contraction rates for dissipators that satisfy stronger conditions than those from \cref{thm:complete_characterization}. These expressions will thus not always give strictly positive rates for all Hamiltonian-independently contractive dissipators. 

\subsection{Some more general theorems}

Below, we present two conditions that guarantee that any driven Lindbladian constructed from a given dissipator $\D$ is strictly contractive in trace distance (\cref{thm:thm contr frm orthog vects}) or strictly contractive in Hilbert-Schmidt distance (\cref{thm:thm eigen condition}).
Both imply exponential contractivity, but are inequivalent, and neither is necessary, as illustrated below in \cref{prop:independence of contractivity conditions}. However, if these conditions are satisfied, then one obtains Hamiltonian-independent contraction rates that are likely much stronger than what one obtains from the proof of \cref{thm:complete_characterization}. 

\begin{thm} \label{thm:thm contr frm orthog vects}
Consider a dissipator $\mathcal{D} = \sum_j \mathcal{D}_{L_j}$. Define $R(\mathcal{D})$ and $r (\mathcal{D})$ via
\begin{align} \label{eq:r1 def}
R (\mathcal{D}) & \coloneqq  \min_{\begin{matrix}
\| u \| , \| v \| = 1 \\
\braket{u | v} = 0 \\
\end{matrix}} \sum_j \left( |\bra{v} L_j \ket{u}|^2 + |\bra{v} L_j^\dagger \ket{u}|^2 \right) \\ \label{eq:r def thm}
r (\mathcal{D}) & \coloneqq \min_{\begin{matrix}
\| u \| , \| v \| = 1 \\
\braket{u | v} = 0 \\
\end{matrix}} \sum_j  |\bra{v} L_j \ket{u}|^2.
\end{align}
If $R(\mathcal{D}) > 0$, then $\mathcal{D}$ is Hamiltonian-independently contractive with a contraction rate $\gamma \geq \max ( R(\mathcal{D} , r(\mathcal{D}) d)  $, where $d=\dim(\mathcal{H})$ is the Hilbert space dimension. The constant $K$ from \cref{def:exp contraction} can be chosen as $K=1$. 
\end{thm}
\begin{proof}
First, we note that the minima from the definitions of $R (\mathcal{D})$ in \cref{eq:r1 def} and $r (\mathcal{D})$ in \cref{eq:r def thm} always exist, because the relevant functions are continuous in $\ket{v}$ and $\ket{u}$ and the set $\left\{ \ket{v}, \ket{u} \in \mathbb{C}^d \ \middle| \ \| v \| , \| u \| = 1 , \ \braket{v | u} = 0 \right\}$ is compact, as it is the intersection of a product of compact sets and a closed set. Consider now any Lindbladian $\mathcal{L}_t = - i [ H(t) , \cdot ] + \mathcal{D}$, on a $d$-dimensional Hilbert space, with dissipator $\mathcal{D} = \sum_j \mathcal{D}_{L_j}$, and consider any quantum states $\rho , \sigma$. Let now $\mathcal{E}_{t,s} (\rho - \sigma)$ be the time-evolution of $\rho - \sigma$ under the Lindbladian master equation (\cref{eq:time_dep_lindblad}). Since $\mathcal{L}_t$ is Hermiticity preserving, we can at all time orthogonally diagonalize $\mathcal{E}_{t,s} (\rho - \sigma)$, and define thus 
\begin{equation} \label{eq:rates results def x(t)}
x(t) \coloneqq \mathcal{E}_{t,s} (\rho - \sigma ) = \sum_{k=1}^d \lambda_k (t) \ket{e_k (t)} \bra{e_k (t)} ,
\end{equation}
with $\{ \lambda_k (t) \}_{k=1}^d$ being the eigenvalues of $\mathcal{E}_{t,s} (\rho - \sigma)$, and $\{ \ket{e_k (t)} \}_{k=1}^d$ the corresponding orthonormal eigenvectors. It now follows from \cref{prop:eigenvalue evolution under Lindbladian} that we can label the eigenvalues such that the following evolution-equation is satisfied for almost all times $t$
\begin{subequations} \label{eq:proof rate results:Q matrix introduction}
\begin{align}
\frac{d}{dt} \lambda_k (t) = \sum_{l=1}^d Q_{kl}(t) \lambda_l (t),
\end{align}
where
\begin{align}
Q_{kl}(t) \coloneqq \left\{ \begin{matrix}
\sum_j | \braket{e_k (t) | L_j | e_l (t) } |^2 & \text{if} \ k \neq l \\
- \sum_{ m \neq l}^d \sum_j | \braket{e_m (t) | L_j | e_l (t) } |^2 & \text{if} \ k=l
\end{matrix} \right .
\end{align}
\end{subequations}
Now, defining $x(t) \coloneqq \mathcal{E}_{t,s} (\rho - \sigma)$ and using $\| x(t) \|_1 = \sum_{k=1}^d | \lambda_k (t) |$, we get that for almost all times, the following is satisfied
\begin{equation} \label{eq:proof rate results:1 norm derivative}
\begin{aligned}
\frac{d}{dt} \left\| x(t) \right\|_1 &= \sum_{k,l = 1}^d \sgn_k Q_{kl}(t) \lambda_l (t) \\
&= \sum_{l,k=1}^d \sum_j |\lambda_l (t)| \sgn_k \sgn_l \left( \left| \braket{e_k (t) | L_j | e_l (t) }  \right|^2 - \delta_{kl} \sum_{m=1}^d \left| \braket{e_m (t) | L_j | e_l (t) }  \right|^2 \right) \\
&= -\sum_j \sum_{k,l=1}^d |\lambda_l (t)| f(k,l) \left| \braket{e_k (t) | L_j | e_l (t) }  \right|^2 ,
\end{aligned}
\end{equation}
where $\sgn_k \coloneqq \sgn (\lambda_k (t))$ in case $\lambda_k (t) \neq 0$ and $\sgn_k \coloneqq \sgn ( d \lambda_k (t) / dt ) $ in case $\lambda_k (t) = 0$, and finally in case $0 = \lambda_k (t) = d \lambda_k (t) / dt$, we can choose $\sgn_k = \pm 1$ arbitrarily. We have also defined $f(k,l) \coloneqq 1 - \sgn_k \sgn_l \in \{ 0  , 2 \}$. It is not immediately obvious that \cref{eq:proof rate results:1 norm derivative} will hold for almost all times even though this is the case for \cref{eq:proof rate results:Q matrix introduction}. The potential problem is that $|\lambda_k (t) |$ may fail to be differentiable at $t$ if $\lambda_k (t) = 0$, since $\sgn_k$ can here change discontinuously. However, if this is the case, this means that since the eigenvalues $\lambda_k (t)$ must change continuously for all times (see \cref{prop:eigenvalue evolution under Lindbladian}), there will exist some open interval $(t-\delta , t + \delta)$ around $t$ for which $\lambda_k (t')$, $t' \in (t-\delta , t + \delta)$, is either positive for all $t' > t$ and negative for $t' < t$, or the other way around, showing that $|\lambda_k (t') |$ will be differentiable everywhere on this open interval, except at $t'=t$. Thus, the set of times for which these additional failures of \cref{eq:proof rate results:1 norm derivative} can occur is a set of \emph{isolated points} on $\mathbb{R}$ and must have Lebesgue measure equal to zero, see e.g. \cite{Bartle1995-hg}, which shows that the equation still holds for almost all times. Note that since the Lindbladian evolution is trace preserving, we have at all times $\tr (\rho - \sigma) = 0 = \tr x(t) = \sum_{k=1}^d \lambda_k (t)$, and therefore $\norm{x(t)} = 2 \sum_{\sgn_k = 1} |\lambda_k (t)| = \sum_{\sgn_k = -1} |\lambda_k (t)| $.

We are now ready to prove \cref{thm:thm contr frm orthog vects} via \cref{eq:proof rate results:1 norm derivative}. Note first that from the definition of $R(\mathcal{D})$ in \cref{eq:r1 def}, we have $\sum ( |\braket{e_k (t) | L_j | e_l (t)}|^2 + |\braket{e_l (t) | L_j | e_k (t)}|^2 ) \geq R(\mathcal{D})$ for any $k \neq l$ and $t$. Then, using this in \cref{eq:proof rate results:1 norm derivative}, if we consider terms with $k\neq l$ that correspond to non-zero eigenvalues $\lambda_k (t) , \lambda_l (t) \neq 0$ with opposite signs, we can upper bound their mutual contribution to the sum by $-2\min(|\lambda_k|,|\lambda_l|)R(\D)$.
Dropping the other terms in the sum, we thus obtain that for almost all times $t$:
\begin{equation} \label{eq:proof rate results: Big R}
\begin{aligned}
\frac{d}{dt} & \norm{x(t)}_1 
= - \sum_{k,l=1}^d \sum_j |\lambda_k (t)|  f(k,l) |\!\bra{e_l (t)} L_j \ket{e_k (t)}\!|^2\\
& \qquad \quad \leq - 2 \sum_{\lambda_k (t) > 0, \lambda_l (t) < 0}\!\! \min(|\lambda_k (t)| , |\lambda_l (t) |) R (\mathcal{D})   \\ 
& \qquad \quad \leq - R (\mathcal{D}) \left\| x(t) \right\|_1 .
\end{aligned}
\end{equation}
The first inequality above follows from restricting the sum to non-zero eigenvalues.
The second inequality follows by first noting that $\norm{x(t)}_1 = \sum_k |\lambda_k (t)| = 2 \sum_{\lambda_k > 0} |\lambda_k (t)| = 2 \sum_{\lambda_k < 0} |\lambda_k(t)|$ since $x(t)$ is traceless.
Now, we let $\lambda_M$ be the eigenvalue with largest magnitude, i.e.\ $|\lambda_M| = \max_k |\lambda_k (t)|$, and we can w.l.o.g.\ assume $\lambda_M > 0$ (if $\lambda_M < 0$ we simply change the signs in the following argument). Then, we have $\sum_{\lambda_k (t) > 0, \lambda_l (t) < 0} \min(|\lambda_k (t)| , |\lambda_l (t)|) \geq \sum_{ \lambda_l (t) < 0} \min(|\lambda_M| , |\lambda_l (t)|) = \sum_{\lambda_l (t) < 0} |\lambda_l (t)| = {\norm{x(t)}_1}/{2}$.

Similarly, it follows from the definition of $r ( \mathcal{D})$ in \cref{eq:r def thm} that we have $\sum_j |\braket{e_k (t) | L_j | e_l (t)}|^2 \geq r (\mathcal{D})$ for every $k \neq l$ and $t$. If we use this in \cref{eq:proof rate results:1 norm derivative}, we get that for almost all times:
\begin{equation} \label{eq:proof rates results:small r}
\begin{aligned}
\frac{d}{dt}  \norm{x(t)}_1 
&= - \sum_{k,l=1}^d \sum_j |\lambda_k (t)|  f (k,l) |\!\bra{e_l (t)} L_j \ket{e_k (t)}\!|^2 \\
&\leq - \sum_{k,l=1}^d  |\lambda_k (t)|  f (k,l) r (\mathcal{D}) \\
&= - 2 r (\mathcal{D}) \sum_{\sgn_k \neq \sgn_l} | \lambda_k (t)| \\
&= - 2 r (\mathcal{D}) \left( \sum_{\sgn_l = 1} \frac{\norm{x(t)}}{2} + \sum_{\sgn_l = -1} \frac{\norm{x(t)}}{2}  \right) \\
&= -  r (\mathcal{D}) d \norm{x(t)} .
\end{aligned}
\end{equation}
The contraction rates from \cref{thm:thm contr frm orthog vects} now follow from \cref{eq:proof rate results: Big R} and \cref{eq:proof rates results:small r} and the definition of $x(t)$ from \cref{eq:rates results def x(t)}, since if $y(t)$ is any continuous function that satisfies the differential inequality $ d y (t) / dt \leq K y (t)$ with $K \in \mathbb{R}$ for almost all times $t$, then we must necessarily have $y(t) \leq y (0) \exp ( K t )$ for any $t \geq 0$. This latter result can e.g. be proven by noting that the fundamental theorem of calculus still holds for continuous functions that are differentiable almost everywhere. If $ d y (t) / dt \leq K y (t)$ holds for almost all times, then $ d ( y (t) \exp (- K t )  ) / dt \leq 0$ must also hold for almost all times, and integrating this inequality while using that $y (t) \exp (- K t )$ is continuous and differentiable at almost all times gives us
\begin{equation*}
\begin{aligned}
\left. y(t') e^{- K t'} \right|_0^t = \int_0^t dt' \frac{d}{dt'} \left( y(t') e^{- K t'} \right) \leq \int_0^t dt' 0 = 0 \\
\Rightarrow y(t) e^{- K t} \leq y (0) \Rightarrow y (t) \leq y (0) e^{K t} .
\end{aligned}
\end{equation*}

\end{proof}

\begin{thm} \label{thm:thm eigen condition}
Let $\Delta$ be the linear map $\Delta (x) \coloneqq x - \tau \tr(x)$, $x \in \mathcal{B}(\mathcal{H})$, where $\tau = I / d$ is the maximally mixed state, and let $\mathcal{D}^{\dagger}$ denote the adjoint of the dissipator $\mathcal{D}$, with respect to the Hilbert-Schmidt inner product. Then, the super-operator $\tilde{\mathcal{D}} \coloneqq \Delta \circ ( \mathcal{D}+\mathcal{D}^{\dagger}) / 2 \circ \Delta$, acting on the space of Hermitian operators, is self-adjoint and has real eigenvalues. If the second largest eigenvalue $\mu_2$ of $\tilde{\mathcal{D}}$ is strictly negative, the dissipator $\mathcal{D}$ will be Hamiltonian-independently contractive, with a contraction rate of at least $\gamma \geq |\mu_2| $, independently of $H(t)$ (and the constant $K$ from \cref{def:exp contraction} can be chosen as $K=\sqrt{d}$).
\end{thm}
\begin{proof}[Proof of \cref{thm:thm eigen condition}] \label{proof:thm:thm eigen condition}

Suppose that the second largest eigenvalue $\mu_2$ of $\tilde{\mathcal{D}}$ is strictly negative, where $\tilde{\mathcal{D}}$ is defined in \cref{thm:thm eigen condition}. Consider now any Lindbladian $\mathcal{L}_t = - i [H(t) , \ \cdot \ ]  + \mathcal{D}$ where $H(t)$ is time-dependent and define $x(t) \coloneqq \mathcal{E}_{t,s}(\rho - \sigma)$ for arbitrary quantum states $\rho , \sigma$ and times $t \geq s \geq 0$. By definition and linearity of the time-evolution channel, we have $ d x(t) / dt = \mathcal{L}_t (x(t))$ holding at almost all times. Note now that
\begin{equation} \label{eq:proof HS eigen:s1}
\begin{aligned}
\frac{d}{dt} \left\| x(t) \right\|_2^2 &= 2 \tr \left[ - i x(t) \left[ H(t) , x(t) \right] + x(t) \mathcal{D} x(t) \vphantom{2^2} \right] \\ 
&=  2 \left\langle x(t) , \frac{\mathcal{D} + \mathcal{D}^{\dagger}}{2} x(t) \right\rangle_{ \text{HS} }  ,
\end{aligned}
\end{equation}
where $\langle A , B \rangle_{ \text{HS} } \coloneqq \tr \left(A^\dagger B \right)$ is the \emph{Hilbert-Schmidt inner product}. Since $x(t)$ by construction is traceless, we have $\Delta (x(t)) = x(t)$, where $\Delta$ is the superoperator defined as $\Delta (y) \coloneqq y - \tau \tr (y)$ with $\tau = I / d$ being the maximally mixed state. Using that $\Delta = \Delta^\dagger$ is self-adjoint w.r.t. the Hilbert-Schmidt inner product, as can be checked relatively easily, the expression \cref{eq:proof HS eigen:s1} can be further rewritten into
\begin{equation} \label{eq:proof HS eigen:s2}
\begin{aligned}
 \frac{d}{dt} \left\| x(t) \right\|_2^2 
&= \left\langle \Delta (x(t)) , (\mathcal{D} + \mathcal{D}^{\dagger}) \circ \Delta (x(t)) \right\rangle_{ \text{HS} } \\
&=  2\langle x(t) , \tilde{\mathcal{D}} x(t) \rangle_{ \text{HS} }  ,
\end{aligned}
\end{equation}
where we have used $\tilde{\mathcal{D}} \coloneqq \Delta \circ ( \mathcal{D} + \mathcal{D}^{\dagger} ) / 2 \circ \Delta$ of $\tilde{\mathcal{D}}$, and the fact that $\Delta$ is self-adjoint, which also entails that $\tilde{\mathcal{D}}$ is self-adjoint and has real eigenvalues. By assumption, the second-largest eigenvalue $\mu_2 < 0$ of $\tilde{\mathcal{D}}$ is strictly negative. Since $0$ is a trivial eigenvalue of $\tilde{\mathcal{D}}$ with eigenvector $\tau$, this means that all eigenvectors of $\tilde{\mathcal{D}}$ perpendicular to $\tau$ have eigenvalues at most $\mu_2 < 0$. $x(t)$ must be perpendicular to $\tau$ since it is traceless $\langle x(t) , \tau \rangle_{ \text{HS} }  = \tr (x(t)) / d = 0$, which means that $x(t)$ belongs to the union of eigenspaces for $\tilde{\mathcal{D}}$ with eigenvalues of at most $\mu_2$. By the maximization principle, we then have $\langle x(t) , \tilde{\mathcal{D}} x(t) \rangle_{ \text{HS} } \leq \mu_2 \norm{x(t)}_2^2$, which when plugged into the expression \cref{eq:proof HS eigen:s2} above gives
\begin{align*}
 \frac{d}{dt} \left\| x(t) \right\|_2^2 
\leq 2 \mu_2 \left\| x(t) \right\|_2^2.
\end{align*}
Integrating the inequality above between times $s$ and $t$ gives
\begin{align*}
\norm{x(t)}_2^2 \leq \norm{x(s)}_2^2 e^{2 \mu_2(t - s)} = \norm{x(s)}_2^2 e^{-2 \abs{\mu_2}({t - s})} ,
\end{align*}
which proves the Hamiltonian-independent contraction rate of at least $|\mu_2|$ for $\mathcal{D}$ by definition of $x(t) \coloneqq \mathcal{E}_{t,s}(\rho - \sigma)$, once we take the square root of both sides of the inequality above. The fact that we can choose the constant from \cref{def:exp contraction} to equal $K = \sqrt{d}$ follows from the inequalities $\left\| A \right\|_2 \leq \left\| A \right\|_1 \leq \sqrt{d} \left\| A \right\|_2$ holding for all $A \in \mathcal{B}(\mathcal{H})$ \cite[Sec. 1.1.3.]{Watrous2018}.
\end{proof}

Note lastly in relation to \cref{thm:thm eigen condition} that a large framework has been constructed linking spectral gaps of various operators to the contraction rates~\cite{Temme2010, Kastoryano2013-ot, capel2020modified}, also called \emph{decay rates} or \emph{mixing times}, of irreducible Lindbladians. The operator $\tilde{\mathcal{D}}$ from \cref{thm:thm eigen condition} can be seen as a simple instance of such an operator, which raises the question whether some of these powerful operator eigenvalue-based contraction results could yield more general conditions for better Hamiltonian-independent contraction rates.
The issue is that these results do not seem to apply in our case.
Specifically, the more complicated operators in those works depend on the fixed point, which in our time-dependent setting will in general fail to commute with the time-dependent Hamiltonian, which makes it difficult to say anything about convergence. 
In the special case when the maximally mixed state is the unique fixed point of the dissipator, \cref{thm:thm eigen condition} already provides the necessary and sufficient condition for Hamiltonian-independent contractivity, see \cref{subsec:case study:unital} below. And finally, a state evolving under a driven Lindbladian may never get arbitrarily close to a fixed point of the dissipator, since a sufficiently strong time-dependent Hamiltonian can continuously "push away" the state from the fixed point. This shows that the mentioned results, aiming to prove convergence to a fixed point, are simply too strong to be generalized to arbitrary driven Lindbladians. 

Finally, we show that the conditions in \cref{thm:thm contr frm orthog vects} and \cref{thm:thm eigen condition} are distinct. In particular, we here show that there exist dissipators for which $\mu_2 < 0$ but $R \not> 0$, and there exist other dissipators for which $R > 0$ but $\mu_2 \not< 0$. In fact, we shall show that even the less general condition of having $r > 0$ from \cref{thm:thm contr frm orthog vects} is not implied by $\mu_2 < 0$. 
\begin{prp} \label{prop:independence of contractivity conditions}
Define the two dissipators $\mathcal{D}_1$ and $\mathcal{D}_2$, acting on $3$-dimensional systems, as
\begin{align*}
\mathcal{D}_1 & \coloneqq \mathcal{D}_{L_{\alpha = 1}} \qquad \text{where} \ \ L_{\alpha = 1} \coloneqq \ket{0} \bra{1} + \ket{1} \bra{2} , \\
\mathcal{D}_2 & \coloneqq \mathcal{D}_{L_{\alpha = 4}} + \epsilon \sum_{i,j = 1}^3 \mathcal{D}_{L_{ij}^e} \qquad \text{where} \ \ L_{\alpha = 4} \coloneqq \ket{0} \bra{1} + 4 \ket{1} \bra{2}  \ \ \text{and} \ \ L_{ij}^e \coloneqq \ket{i} \bra{j} ,
\end{align*}
and $\eps$ is an arbitrary constant in the range $0 < \epsilon < 0.3$. Then, we have $\mu_2 ( \mathcal{D}_1 ) < 0$ while $R(\mathcal{D}_1) = r (\mathcal{D}_1)  = 0$. And for $\mathcal{D}_2$, we have $\mu_2 (\mathcal{D}_2) > 0$ while $R (\mathcal{D}_2) , r (\mathcal{D}_2) > 0$. Hence, $\mathcal{D}_1$ satisfies the contractivity condition from \cref{thm:thm eigen condition}, but not the general one from \cref{thm:thm contr frm orthog vects}, while $\mathcal{D}_2$ satisfies both contractivity conditions from \cref{thm:thm contr frm orthog vects} but not the condition from \cref{thm:thm eigen condition}.
\end{prp}

\begin{proof}
We first show $\mu_2 (\mathcal{D}_1) < 0$ and $R (\mathcal{D}_1) = r (\mathcal{D}_1) = 0$. As shown by \cref{sec:continuous family of 3 dim ladders}, the dissipator $\mathcal{D}_1$ satisfies the contractivity condition of \cref{thm:thm eigen condition}, since the second largest eigenvalue $\mu_2 (\mathcal{D}_1)$ of $\Delta \circ (\mathcal{D}_1 + \mathcal{D}_1^\dagger)/2 \circ \Delta$ equals $\mu_2 (\mathcal{D}_1) = - 1 / (3+\sqrt{5} ) < 0$. However, the conditions from \cref{thm:thm contr frm orthog vects} are not satisfied, which can e.g. be seen by $|\bra{0} L_{\alpha=1} \ket{2}|^2 + |\bra{0} L_{\alpha=1}^\dagger \ket{2}|^2 = 0$, hence $R (\mathcal{D}_1) = 0$ and $r (\mathcal{D}_1) \leq R(\mathcal{D}_1)= 0$. 
\\
\\
We now show that we have $\mu_2 (\mathcal{D}_2) > 0$ while $R (\mathcal{D}_2) , r (\mathcal{D}_2) > 0$. Let $\{ L_k \}_k$ be the jump operators appearing in $\mathcal{D}_2$. Since all terms in the following sum are positive, we have for any orthogonal unit-vectors $\ket{v} , \ket{u}$
\begin{align*}
\sum_k  |\bra{v} L_k \ket{u}|^2 \geq \sum_{i,j=1}^3 |\bra{v} \sqrt{\epsilon} L_{ij}^e \ket{u}|^2 = \sum_{i,j=1}^3 |\bra{v} \sqrt{\epsilon} \ket{i} \bra{j} \ket{u}|^2 = \epsilon \sum_{i=1}^3 \braket{v | i} \braket{i | v} \sum_{j=1}^3 \braket{u | j} \braket{j | u} = \epsilon \braket{v|v} \braket{u|u} = \epsilon ,
\end{align*}
Hence, $R (\mathcal{D}_2) \geq r (\mathcal{D}_2) \geq \epsilon > 0$. However, as is proven in \cref{sec:continuous family of 3 dim ladders}, the dissipator $\mathcal{D}_{L_{\alpha = 4}}$ has largest non-trivial eigenvalue equal to $\mu_2 (\mathcal{D}_{L_{\alpha = 4}}) = \sqrt{964 / 3} -17  \approx 0.926$, and since $\mu_2 (\mathcal{D})$ can be calculated as $\mu_2 (\mathcal{D}) = \max_{\tr(x)=0 , \tr(x^2) = 1} \tr \left( x \mathcal{D}x \right)$ where $x$ ranges over Hermitian matrices, we can lower bound $\mu_2 (\mathcal{D}_2)$ by
\begin{align*}
\mu_2 (\mathcal{D}_2) & = \max_{\tr(x)=0 , \tr(x^2) = 1} \tr \left( x \mathcal{D}_2 x \right)\nonumber\\
&\geq\max_{\tr(x)=0 , \tr(x^2) = 1} \tr \left( x \mathcal{D}_{L_{\alpha = 4 }} x \right) + \min_{\tr(x)=0 , \tr(x^2) = 1} \epsilon \sum_{i,j=1}^3 \tr \left( x \mathcal{D}_{L_{ij}^e} x \right) \\
& > 0.9 + \epsilon \min_{\tr(y)=0 , \tr(y^2) = 1} \left( \sum_{i,j=1}^3 \braket{i|y|i} \braket{j|y|j} - 3 \sum_{j=1}^3 \braket{j|y^2|j} \right) \nonumber \\
&= 0.9 - 3 \epsilon.
\end{align*}
Since by assumption $\epsilon < 0.3$, we obtain that $\mu_2 (\mathcal{D}_2) > 0$. 
\end{proof}

\subsection{Applying the new results to specific dissipators} \label{subsec:case studies}

In this section we apply the new results, especially \cref{thm:thm contr frm orthog vects} and \cref{thm:thm eigen condition}, to some families of dissipators and obtain Hamiltonian-independent contraction rates or simply prove Hamiltonian-independent contractivity. 

\subsubsection{Unital dissipators} \label{subsec:case study:unital}

When we apply \cref{thm:thm eigen condition} to unital dissipators, we get the following result that like \cref{thm:complete_characterization} classifies which dissipators are Hamiltonian-independently contractive, but also provides strong Hamiltonian-independent contraction rates. We remark that in Ref. \cite[thm. II.2]{Plenio2025}, it was already established that if the jump operator set of a unital dissipator is self-adjoint and generates the whole matrix algebra (a stronger condition than the one in \cref{cor:unital dissipator}), the dynamics is Hamiltonian-independently contractive. 

\begin{prp}\label{cor:unital dissipator}
Assume that a dissipator $\mathcal{D} = \sum_j \mathcal{D}_{L_j}$ is unital, i.e. $\mathcal{D}(I) = 0$. Then, $\mathcal{D}$ is Hamiltonian-independently contractive if and only if $\{ L_j , L_j^\dagger \}_j$ has a trivial commutant $\{ L_j , L_j^\dagger \}_j' = \mathbb{C} I $, in which case the Hamiltonian-independent contraction rate $\gamma$ for $\mathcal{D}$ can be chosen as
\begin{equation}
\gamma = \frac{1}{2} \min_{ \tr x = 0} \sum_j \frac{\left\| \left[ x , L_j \right] \right\|_2^2}{ \| x \|_2^2 }
\end{equation}
\end{prp}
\begin{proof}[Proof] \label{proof:cor:unital dissipator}

Consider a unital dissipator $\mathcal{D} = \sum_j \mathcal{D}_{L_j}$. Note first that any operator $x \in \mathcal{B}(\mathcal{H})$ which commutes with all of $\{ L_j , L_j^\dagger \}_j$ must be annihilated by $\mathcal{D}$ 
\begin{equation}
\begin{aligned}
\mathcal{D} (x) = \sum_j \left( L_j x L_j^\dagger - \frac{1}{2} \left\{ L_j^\dagger L_j , x \right\} \right) = \sum_j \left( L_j L_j^\dagger -  L_j^\dagger L_j \right) x = \mathcal{D}(I) x = 0 ,
\end{aligned}
\end{equation}
where we have used above that $\mathcal{D}$ is unital, hence $\mathcal{D}(I) = 0$. Thus, if $\{ L_j , L_j^\dagger \}$ has non-trivial commutant, this means that there exists some operator $x$ which is linearly independent of $I$ and for which $\mathcal{D} (x) = 0$. Since $\mathcal{D}$ is Hermiticity-preserving, we can w.l.o.g. assume $x$ to be Hermitian, otherwise we can just consider its Hermitian or anti-Hermitian part. By linearity, we thus have $\mathcal{D} ( I / d ) = 0$ and $\mathcal{D}  ( ( x + \alpha I ) / \tr (x + \alpha I) ) = 0$, where we choose $\alpha \in \mathbb{R}$ such that $( x + \alpha I ) / \tr (x + \alpha I)$ is a quantum state. Thus, $\mathcal{D}$ has multiple fixed points, which means that it cannot be Hamiltonian-independently contractive.

To prove the opposite implication, assume $\{ L_j , L_j^\dagger \}_j'=\mathbb{C} I$ and by noting that $\mathcal{D}^\dagger = \sum_j L_j^\dagger \cdot L_j - \{L_j^\dagger L_j, \cdot\}/2$, we get that for any Hermitian $y \in \mathcal{B}(\mathcal{H})$
\begin{equation} \label{eq:proof unital: l1}
\begin{aligned}
\tr \left( y \left( \mathcal{D}+\mathcal{D}^{\dagger} \right) y \right)&= \tr \bigg( y \sum_j \left( L_j y L_j^\dagger + L_j^\dagger y L_j   -L_j^\dagger L_j y - y L_j^\dagger L_j \right) \bigg) \\
&= \sum_j  \tr \left( (y L_j - L_j y ) ( y L_j^\dagger - L_j^\dagger y ) - y^2 [L_j , L_j^\dagger ] \right) \\
&= - \sum_j  \left\Vert \left[ y , L_j \right] \right\Vert_2^2  -  \tr \bigg( y^2 \sum_j [L_j , L_j^\dagger ] \bigg)  .
\end{aligned}
\end{equation}
By the assumption of $\mathcal{D}$ being unital, i.e. $\mathcal{D}(I)=0$, we have $0 = \mathcal{D}(I) = \sum_j [L_j , L_j^\dagger ]$, which when plugged into \cref{eq:proof unital: l1} above while using the definition $\tilde{\mathcal{D}} = \Delta \circ ( \mathcal{D} + \mathcal{D}^\dagger ) /  2  \circ \Delta$ of the operator $\tilde{\mathcal{D}}$ from \cref{thm:thm eigen condition}, gives us for any Hermitian $x \in \mathcal{B}(\mathcal{H})$
\begin{equation} \label{eq:proof unital: eigen eq}
\begin{aligned}
\langle x , \tilde{\mathcal{D}} x \rangle_{ \text{HS} } &= \tr \left( \Delta(x) \frac{\mathcal{D} + \mathcal{D}^\dagger}{2} \Delta (x) \right) \\
&= - \frac{1}{2} \sum_j  \left\Vert \left[ \Delta (x) , L_j \right] \right\Vert_2^2 ,
\end{aligned}
\end{equation}
where we recall that $\Delta (x) \coloneqq x - \tr (x) \tau$ is self-adjoint (w.r.t. the Hilbert-Schmidt inner product). It follows from \cref{eq:proof unital: eigen eq} that we have $\langle x , \tilde{\mathcal{D}} x \rangle_{ \text{HS} } \leq 0$ for any $x$, and $\langle x , \tilde{\mathcal{D}} x \rangle_{ \text{HS} } = 0$ entails $\left[ \Delta (x) , L_j \right] = 0$ for all $j$. Now, since $\tau =  I / d  $ is an eigenvector of $\tilde{\mathcal{D}}$ with eigenvalue $0$ and since all other eigenvectors of $\tilde{\mathcal{D}}$ are orthogonal to $\tau$, the eigenvector $x_2$ of $\tilde{\mathcal{D}}$ with eigenvalue $\mu_2$ must satisfy $\langle \tau , x_2 \rangle_{ \text{HS} } = \tr (x_2 ) / d = 0$ and hence $\Delta (x_2) = x_2$. Thus, plugging the eigenvector $\mathcal{D}x_2 = \mu_2 x_2$ into \cref{eq:proof unital: eigen eq} above gives us
\begin{align} \label{eq:proof unital:mu2}
\mu_2 \left\| x_2 \right\|_2^2 & = \langle x_2 , \tilde{\mathcal{D}} x_2 \rangle_{ \text{HS} } = - \frac{1}{2} \sum_j  \left\Vert \left[ x_2 , L_j \right] \right\Vert_2^2 .
\end{align}
Thus, we see from \cref{eq:proof unital:mu2} that we must have $\mu_2 < 0$, since we cannot have $ \left[ x_2 , L_j \right] = 0$ for all $j$, given that $\tr x_2 = 0$ and $\{ L_j , L_j^\dagger \}_j$ has trivial commutant. This proves the result via \cref{thm:thm eigen condition}, and the Hamiltonian-independent contraction rate $\gamma = - \mu_2$ can further be computed by the minimization principle as
\begin{equation*}
\gamma = \frac{1}{2} \min_{ \tr x = 0} \sum_j \frac{\left\| \left[ x , L_j \right] \right\|_2^2}{ \| x \|_2^2 } .
\end{equation*}
\end{proof}

\subsubsection{Ladder dissipators} \label{subsec:ladder lindblad calculations}

\emph{Ladder dissipators} arise frequently in quantum optics, for example to model harmonic oscillators~\cite{carmichael1993open}, multi-level atoms~\cite{drake2007springer} or collective spin-models~\cite{andreev1980collective, dicke1954coherence}.
Specifically, in a $d$-dimensional Hilbert space with orthonormal basis $\{ \ket{j} \}_{j=0}^{d-1}$, a ladder dissipator takes the form
\begin{align} \label{eq:ladder def}
\mathcal{D} = \mathcal{D}_L \quad \text{with} \ L = \sum_{j=1}^{d-1} \alpha_j \ket{j-1} \bra{j} ,
\end{align}
for some set of real coefficients $\alpha_j>0$. If $\{\alpha_j\}_j$ are not real, we can always find a diagonal unitary $U$ such that $U L U^\dagger = \sum_{j = 1}^{d - 1}\abs{\alpha_j} \ket{j - 1}\!\bra{j}$, which gives rise to an equivalent dissipator w.r.t. Hamiltonian-independent properties. As established by \cref{cor:ladder dissipators are HI contractive}, all (non-degenerate) ladder-dissipators are Hamiltonian-independently contractive, and we shall now focus on proving strong Hamiltonian-independent contraction rates via \cref{thm:thm eigen condition}.  

\begin{figure}
    \centering
    \includegraphics[width=0.5\linewidth]{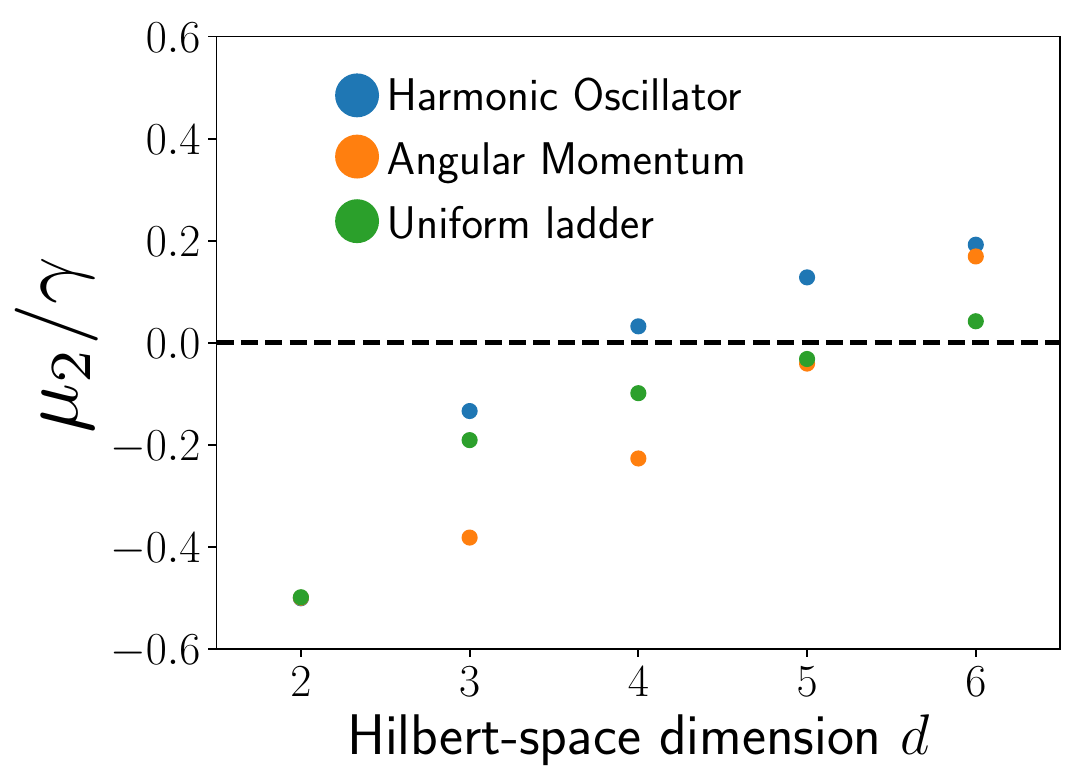}
    \caption{$\mu_2$ as a function of the Hilbert space dimension $d$ for the harmonic oscillator, angular momentum and uniform ladder dissipator. When $\mu_2 < 0$, \cref{thm:thm eigen condition} establishes a Hamiltonian-independent contraction rate.}
    \label{fig:mu2}
\end{figure}

To check whether a general ladder dissipator \cref{eq:ladder def} satisfies the condition from \cref{thm:thm eigen condition}, we compute the eigenvalues of $\tilde{\mathcal{D}}$, which is on a relatively simple block-matrix form in the standard Hermitian basis $\{ \ket{j}\bra{k}+\ket{k} \bra{j} , i (\ket{j}\bra{k}-\ket{k} \bra{j}) \}_{j,k}$. To this end, define $B_l$ for any $0 \leq l \leq d-1$ to be the following $(d-l) \times (d-l)$ tri-diagonal matrix
\begin{align} \label{eq:B matrix}
 B_l \coloneqq -\frac{1}{2}
\begin{pmatrix}
\alpha_l^2 & - \alpha_1 \alpha_{l+1}  & 0 & \ddots  \\
- \alpha_1  \alpha_{l+1} &  \alpha_1^2 + \alpha_{l+1}^2 & - \alpha_2 \alpha_{l+2} & \ddots  \\
0 & - \alpha_2 \alpha_{l+2} & \alpha_2^2 + \alpha_{l+2}^2 & \ddots  \\
\ddots & \ddots & \ddots & \ddots \\
\end{pmatrix} , 
\end{align}
where we use the convention $\alpha_0 = 0$. 
The eigenvalues of $\tilde{\mathcal{D}}$ can then be obtained as
\begin{itemize}
   \item $\tilde{\mathcal{D}}$ contains $1$ eigenvalue for every eigenvalue of $P B_0 P$, where $P$ is the $d \times d$ projection matrix
   \begin{align} 
P = \frac{1}{d} 
\begin{pmatrix}
d - 1 & - 1 & -1 & -1 &  \dots  \\
-1 &  d - 1 & - 1 & - 1 & \dots  \\
-1 & -1 & d - 1 & - 1& \dots  \\
-1 & -1 & -1 & d - 1& \dots \\
\vdots & \vdots & \vdots & \vdots & \ddots \\
\end{pmatrix} .  \nonumber
\end{align} 
\item $\tilde{\mathcal{D}}$ further contains $2$ equal eigenvalues for every eigenvalue of $B_l$, for all $ 1 \leq l \leq d-1$. 
\end{itemize}

We specifically analyze the following ladder dissipators:
\begin{itemize}
    \item \emph{Harmonic oscillator} dissipator, with $\alpha_j = \sqrt{ \gamma j}$ for some $\gamma > 0$.
    \item \emph{Angular momentum} dissipator, with $\alpha_j = \sqrt{\gamma j (d-j)}$ for some $\gamma > 0$.
    \item \emph{Uniform} ladder dissipator, with $\alpha_j = \sqrt{\gamma}$ for some $\gamma > 0$. 
\end{itemize}
\Cref{fig:mu2} shows the eigenvalue $\mu_2$ as a function of $d$ for these three different cases: we obtain strong Hamiltonian-independent contraction rates for the angular momentum ladder dissipators and uniform ladder dissipators when $ d \leq 5$, and for the $d$-dimensional harmonic oscillator dissipator when $ d \leq 3$.

Let us now focus on dimension $d=3$, and consider here the most general non-degenerate ladder dissipator $\mathcal{D}$, which can be parameterized as
\begin{align} \label{def:general d=3 diss}
\mathcal{D}_{\eta , \alpha} \coloneqq \eta \ \mathcal{D}_{L_\alpha} \qquad \text{with} \quad L_\alpha \coloneqq \ket{0} \bra{1} + \alpha \ket{1} \bra{2} ,
\end{align}
with $\eta , \alpha > 0$ positive.
When we apply \cref{thm:thm eigen condition} to $\mathcal{D}_{\eta , \alpha}$, we obtain the following result by utilizing the reduction of this procedure given in the preceding paragraph.
\begin{prp} \label{sec:continuous family of 3 dim ladders}
For $ \alpha \in ( \sqrt{3} - \sqrt{2}, \sqrt{3}+\sqrt{2} ) = ( 0.318 ,\ 3.146)$, the general $3$-dimensional dissipator $\mathcal{D}_{\eta , \alpha}$ from \cref{def:general d=3 diss} has a Hamiltonian-independent contraction rate $c_{\alpha} \eta$, with $c_\alpha$ being given by
\begin{equation} \label{eq:d=3 ladder alpha def}
\begin{aligned} 
c_\alpha =
    \min& \left\{\frac{2+\alpha^2}{4}-\frac{ \alpha }{4}\sqrt{4+\alpha^2},\right.  \left.\frac{1+\alpha^2}2-\sqrt{\frac{1-\alpha^2+\alpha^4}3} \right\},
\end{aligned}
\end{equation}
and the constant $K$ from \cref{def:exp contraction} can be chosen as $K=\sqrt{3}$. 
\end{prp}
\begin{proof}
We use the following orthonormal basis for the space of Hermitian operators, which is orthonormal w.r.t. the Hilbert-Schmidt inner product $\langle \ \cdot \ , \ \cdot \ \rangle_{ \text{HS} }$: 
\begin{align*}
\boldsymbol{e}_1 \coloneqq 
\begin{pmatrix}
1 & 0 & 0 \\
0 & 0 & 0 \\
0 & 0 & 0 \\
\end{pmatrix}  , \ \ 
\boldsymbol{e}_2 \coloneqq 
\begin{pmatrix}
0 & 0 & 0 \\
0 & 1 & 0 \\
0 & 0 & 0 \\
\end{pmatrix}  , \ \ 
\boldsymbol{e}_3 \coloneqq 
\begin{pmatrix}
0 & 0 & 0 \\
0 & 0 & 0 \\
0 & 0 & 1 \\
\end{pmatrix} , \ \ 
\boldsymbol{e}_4 \coloneqq \frac{1}{\sqrt{2}}
\begin{pmatrix}
0 & 1 & 0 \\
1 & 0 & 0 \\
0 & 0 & 0 \\
\end{pmatrix}  , \ \
\boldsymbol{e}_5 \coloneqq \frac{1}{\sqrt{2}} 
\begin{pmatrix}
0 & 0 & 0 \\
0 & 0 & 1 \\
0 & 1 & 0 \\
\end{pmatrix}  , \\ 
\boldsymbol{e}_6 \coloneqq \frac{1}{\sqrt{2}}
\begin{pmatrix}
0 & i & 0 \\
-i & 0 & 0 \\
0 & 0 & 0 \\
\end{pmatrix}  , \ \ 
\boldsymbol{e}_7 \coloneqq \frac{1}{\sqrt{2}}
\begin{pmatrix}
0 & 0 & 0 \\
0 & 0 & i \\
0 & -i & 0\\
\end{pmatrix}  , \ \ 
\boldsymbol{e}_8 \coloneqq \frac{1}{\sqrt{2}}
\begin{pmatrix}
0 & 0 & 1 \\
0 & 0 & 0 \\
1 & 0 & 0 \\
\end{pmatrix}   , \ \ 
\boldsymbol{e}_9 \coloneqq \frac{1}{\sqrt{2}}
\begin{pmatrix}
0 & 0 & i \\
0 & 0 & 0 \\
-i & 0 & 0 \\
\end{pmatrix} .
\end{align*}
In this basis, it can be checked by direct calculation that the operator $\tilde{\mathcal{D}}$ has the following matrix representation
\begin{align} \label{eq:d=3 ladder matrix rep}
\left( \langle \boldsymbol{e}_i , \tilde{\mathcal{D}} \boldsymbol{e}_j \rangle_{ \text{HS} } \right)_{1 \leq i,j \leq 9}  = - \frac{1}{2} \eta \begin{pmatrix}
\frac{2}{3} & \frac{\alpha^2 - 3}{3} & \frac{1 - \alpha^2}{3} & 0 & 0 & 0 & 0 & 0 & 0 \\
\frac{ \alpha^2 - 3}{3} & \frac{2 \alpha^2 + 4}{3} & - \frac{3 \alpha^2 + 1}{3} & 0 & 0 & 0 & 0 & 0 & 0 \\
\frac{1 - \alpha^2 }{3} & - \frac{3 \alpha^2 + 1 }{3} &  \frac{4 \alpha^2 }{3} & 0 & 0 & 0 & 0 & 0 & 0 \\
0 & 0 & 0 & 1 &  -\alpha & 0 & 0 & 0 & 0 \\
0 & 0 & 0 &  -\alpha &  1+\alpha^2  & 0 & 0 & 0 & 0 \\
0 & 0 & 0 & 0 & 0 & 1 & -\alpha & 0 & 0 \\
0 & 0 & 0 & 0 & 0 & -\alpha & 1+\alpha^2 & 0 & 0 \\
0 & 0 & 0 & 0 & 0 & 0 & 0 &  \alpha^2 & 0 \\
0 & 0 & 0 & 0 & 0 & 0 & 0 & 0 &  \alpha^2 \\
\end{pmatrix} .
\end{align}
We can now determine the eigenvalues of the matrix \cref{eq:d=3 ladder matrix rep} above by direct calculation. We are helped slightly along the way by the fact that $\boldsymbol{e}_1 + \boldsymbol{e}_2 + \boldsymbol{e}_3 $, i.e. the vector-representation of the identity $I$, must be an eigenvector with eigenvalue $0$, by construction of $\tilde{\mathcal{D}}$. The eigenvalues of the matrix \cref{eq:d=3 ladder matrix rep} can be computed explicitly, and they are (counting multiplicities):
\begin{itemize}
\item $1$ eigenvalue equaling $0$,

\item $1$ eigenvalue equaling $- \frac{\eta}{2} \left( (1+\alpha^2) +  \sqrt{(1+\alpha^2)^2 - \frac{10 \alpha^2 - 1 - \alpha^4 }{3} } \right)$,

\item $1$ eigenvalue equaling $- \frac{\eta}{2} \left( (1+\alpha^2) -  \sqrt{(1+\alpha^2)^2 - \frac{10 \alpha^2 - 1 - \alpha^4 }{3} } \right)$,

\item $2$ eigenvalues equaling $ - \frac{1}{2} \eta \alpha ^2$,

\item $2$ eigenvalues equaling $ - \frac{\eta}{2} \left( 1 + \frac{ \alpha ^2}{2} + \sqrt{\left( 1 + \frac{ \alpha ^2}{2} \right)^2 - 1} \right)$,

\item $2$ eigenvalues equaling $ - \frac{\eta}{2} \left( 1 + \frac{ \alpha ^2}{2} - \sqrt{\left( 1 + \frac{ \alpha ^2}{2} \right)^2 - 1} \right)$.
\end{itemize}

Since $\mu_2$ is the second largest among these eigenvalues, it can be shown by taking the maximum over the eigenvalues presented above and slightly rewriting them using basic arithmetic, that we indeed have $\mu_2 = - c_\alpha \eta$, with $c_\alpha$ given in \cref{eq:d=3 ladder alpha def}. 
\end{proof}

\end{appendices}

\bibliography{refs}

\begin{thebibliography}{64}%
\makeatletter
\providecommand \@ifxundefined [1]{%
 \@ifx{#1\undefined}
}%
\providecommand \@ifnum [1]{%
 \ifnum #1\expandafter \@firstoftwo
 \else \expandafter \@secondoftwo
 \fi
}%
\providecommand \@ifx [1]{%
 \ifx #1\expandafter \@firstoftwo
 \else \expandafter \@secondoftwo
 \fi
}%
\providecommand \natexlab [1]{#1}%
\providecommand \enquote  [1]{``#1''}%
\providecommand \bibnamefont  [1]{#1}%
\providecommand \bibfnamefont [1]{#1}%
\providecommand \citenamefont [1]{#1}%
\providecommand \href@noop [0]{\@secondoftwo}%
\providecommand \href [0]{\begingroup \@sanitize@url \@href}%
\providecommand \@href[1]{\@@startlink{#1}\@@href}%
\providecommand \@@href[1]{\endgroup#1\@@endlink}%
\providecommand \@sanitize@url [0]{\catcode `\\12\catcode `\$12\catcode
  `\&12\catcode `\#12\catcode `\^12\catcode `\_12\catcode `\%12\relax}%
\providecommand \@@startlink[1]{}%
\providecommand \@@endlink[0]{}%
\providecommand \url  [0]{\begingroup\@sanitize@url \@url }%
\providecommand \@url [1]{\endgroup\@href {#1}{\urlprefix }}%
\providecommand \urlprefix  [0]{URL }%
\providecommand \Eprint [0]{\href }%
\providecommand \doibase [0]{https://doi.org/}%
\providecommand \selectlanguage [0]{\@gobble}%
\providecommand \bibinfo  [0]{\@secondoftwo}%
\providecommand \bibfield  [0]{\@secondoftwo}%
\providecommand \translation [1]{[#1]}%
\providecommand \BibitemOpen [0]{}%
\providecommand \bibitemStop [0]{}%
\providecommand \bibitemNoStop [0]{.\EOS\space}%
\providecommand \EOS [0]{\spacefactor3000\relax}%
\providecommand \BibitemShut  [1]{\csname bibitem#1\endcsname}%
\let\auto@bib@innerbib\@empty
\bibitem [{\citenamefont {Breuer}\ and\ \citenamefont
  {Petruccione}(2007)}]{Breuer2007theory}%
  \BibitemOpen
  \bibfield  {author} {\bibinfo {author} {\bibfnamefont {H.-P.}\ \bibnamefont
  {Breuer}}\ and\ \bibinfo {author} {\bibfnamefont {F.}~\bibnamefont
  {Petruccione}},\ }\href
  {https://doi.org/10.1093/acprof:oso/9780199213900.001.0001} {\emph {\bibinfo
  {title} {The Theory of Open Quantum Systems}}}\ (\bibinfo  {publisher}
  {Oxford University Press},\ \bibinfo {year} {2007})\BibitemShut {NoStop}%
\bibitem [{\citenamefont {Rivas}\ and\ \citenamefont
  {Huelga}(2012)}]{rivas2012open}%
  \BibitemOpen
  \bibfield  {author} {\bibinfo {author} {\bibfnamefont {A.}~\bibnamefont
  {Rivas}}\ and\ \bibinfo {author} {\bibfnamefont {S.~F.}\ \bibnamefont
  {Huelga}},\ }\href {https://doi.org/10.1007/978-3-642-23354-8} {\emph
  {\bibinfo {title} {Open quantum systems}}},\ Vol.~\bibinfo {volume} {10}\
  (\bibinfo  {publisher} {Springer},\ \bibinfo {year} {2012})\BibitemShut
  {NoStop}%
\bibitem [{\citenamefont {Carmichael}(1993)}]{carmichael1993open}%
  \BibitemOpen
  \bibfield  {author} {\bibinfo {author} {\bibfnamefont {H.}~\bibnamefont
  {Carmichael}},\ }\href {https://doi.org/10.1007/978-3-540-47620-7} {\emph
  {\bibinfo {title} {An open systems approach to quantum optics: lectures
  presented at the Universit{\'e} Libre de Bruxelles October 28 to November 4,
  1991}}}\ (\bibinfo  {publisher} {Springer},\ \bibinfo {year}
  {1993})\BibitemShut {NoStop}%
\bibitem [{\citenamefont {Lindblad}(1976)}]{lindblad1976generators}%
  \BibitemOpen
  \bibfield  {author} {\bibinfo {author} {\bibfnamefont {G.}~\bibnamefont
  {Lindblad}},\ }\bibfield  {title} {\bibinfo {title} {On the generators of
  quantum dynamical semigroups},\ }\href {https://doi.org/10.1007/BF01608499}
  {\bibfield  {journal} {\bibinfo  {journal} {Communications in mathematical
  physics}\ }\textbf {\bibinfo {volume} {48}},\ \bibinfo {pages} {119}
  (\bibinfo {year} {1976})}\BibitemShut {NoStop}%
\bibitem [{\citenamefont {Gorini}\ \emph {et~al.}(1976)\citenamefont {Gorini},
  \citenamefont {Kossakowski},\ and\ \citenamefont
  {Sudarshan}}]{gorini1976completely}%
  \BibitemOpen
  \bibfield  {author} {\bibinfo {author} {\bibfnamefont {V.}~\bibnamefont
  {Gorini}}, \bibinfo {author} {\bibfnamefont {A.}~\bibnamefont
  {Kossakowski}},\ and\ \bibinfo {author} {\bibfnamefont {E.~C.~G.}\
  \bibnamefont {Sudarshan}},\ }\bibfield  {title} {\bibinfo {title} {Completely
  positive dynamical semigroups of n-level systems},\ }\href
  {https://doi.org/10.1063/1.522979} {\bibfield  {journal} {\bibinfo  {journal}
  {Journal of Mathematical Physics}\ }\textbf {\bibinfo {volume} {17}},\
  \bibinfo {pages} {821} (\bibinfo {year} {1976})}\BibitemShut {NoStop}%
\bibitem [{\citenamefont {Wolf}\ and\ \citenamefont
  {Cirac}(2008)}]{wolf2008dividing}%
  \BibitemOpen
  \bibfield  {author} {\bibinfo {author} {\bibfnamefont {M.~M.}\ \bibnamefont
  {Wolf}}\ and\ \bibinfo {author} {\bibfnamefont {J.~I.}\ \bibnamefont
  {Cirac}},\ }\bibfield  {title} {\bibinfo {title} {Dividing quantum
  channels},\ }\href {https://doi.org/10.1007/s00220-008-0411-y} {\bibfield
  {journal} {\bibinfo  {journal} {Communications in Mathematical Physics}\
  }\textbf {\bibinfo {volume} {279}},\ \bibinfo {pages} {147} (\bibinfo {year}
  {2008})}\BibitemShut {NoStop}%
\bibitem [{\citenamefont {Kessler}\ \emph {et~al.}(2012)\citenamefont
  {Kessler}, \citenamefont {Giedke}, \citenamefont {Imamoglu}, \citenamefont
  {Yelin}, \citenamefont {Lukin},\ and\ \citenamefont
  {Cirac}}]{kessler2012dissipative}%
  \BibitemOpen
  \bibfield  {author} {\bibinfo {author} {\bibfnamefont {E.~M.}\ \bibnamefont
  {Kessler}}, \bibinfo {author} {\bibfnamefont {G.}~\bibnamefont {Giedke}},
  \bibinfo {author} {\bibfnamefont {A.}~\bibnamefont {Imamoglu}}, \bibinfo
  {author} {\bibfnamefont {S.~F.}\ \bibnamefont {Yelin}}, \bibinfo {author}
  {\bibfnamefont {M.~D.}\ \bibnamefont {Lukin}},\ and\ \bibinfo {author}
  {\bibfnamefont {J.~I.}\ \bibnamefont {Cirac}},\ }\bibfield  {title} {\bibinfo
  {title} {Dissipative phase transition in a central spin system},\ }\href
  {https://doi.org/10.1103/PhysRevA.86.012116} {\bibfield  {journal} {\bibinfo
  {journal} {Physical Review A}\ }\textbf {\bibinfo {volume} {86}},\ \bibinfo
  {pages} {012116} (\bibinfo {year} {2012})}\BibitemShut {NoStop}%
\bibitem [{\citenamefont {Horstmann}\ \emph {et~al.}(2013)\citenamefont
  {Horstmann}, \citenamefont {Cirac},\ and\ \citenamefont
  {Giedke}}]{horstmann2013noise}%
  \BibitemOpen
  \bibfield  {author} {\bibinfo {author} {\bibfnamefont {B.}~\bibnamefont
  {Horstmann}}, \bibinfo {author} {\bibfnamefont {J.~I.}\ \bibnamefont
  {Cirac}},\ and\ \bibinfo {author} {\bibfnamefont {G.}~\bibnamefont
  {Giedke}},\ }\bibfield  {title} {\bibinfo {title} {Noise-driven dynamics and
  phase transitions in fermionic systems},\ }\href
  {https://doi.org/10.1103/PhysRevA.87.012108} {\bibfield  {journal} {\bibinfo
  {journal} {Physical Review A}\ }\textbf {\bibinfo {volume} {87}},\ \bibinfo
  {pages} {012108} (\bibinfo {year} {2013})}\BibitemShut {NoStop}%
\bibitem [{\citenamefont {Benedict}(2018)}]{benedict2018super}%
  \BibitemOpen
  \bibfield  {author} {\bibinfo {author} {\bibfnamefont {M.~G.}\ \bibnamefont
  {Benedict}},\ }\href {https://doi.org/10.1201/9780203737880} {\emph {\bibinfo
  {title} {Super-radiance: Multiatomic coherent emission}}}\ (\bibinfo
  {publisher} {CRC Press},\ \bibinfo {year} {2018})\BibitemShut {NoStop}%
\bibitem [{\citenamefont {Masson}\ and\ \citenamefont
  {Asenjo-Garcia}(2022)}]{masson2022universality}%
  \BibitemOpen
  \bibfield  {author} {\bibinfo {author} {\bibfnamefont {S.~J.}\ \bibnamefont
  {Masson}}\ and\ \bibinfo {author} {\bibfnamefont {A.}~\bibnamefont
  {Asenjo-Garcia}},\ }\bibfield  {title} {\bibinfo {title} {Universality of
  dicke superradiance in arrays of quantum emitters},\ }\href
  {https://doi.org/10.1038/s41467-022-29805-4} {\bibfield  {journal} {\bibinfo
  {journal} {Nature Communications}\ }\textbf {\bibinfo {volume} {13}},\
  \bibinfo {pages} {2285} (\bibinfo {year} {2022})}\BibitemShut {NoStop}%
\bibitem [{\citenamefont {Asenjo-Garcia}\ \emph {et~al.}(2017)\citenamefont
  {Asenjo-Garcia}, \citenamefont {Moreno-Cardoner}, \citenamefont {Albrecht},
  \citenamefont {Kimble},\ and\ \citenamefont {Chang}}]{asenjo2017exponential}%
  \BibitemOpen
  \bibfield  {author} {\bibinfo {author} {\bibfnamefont {A.}~\bibnamefont
  {Asenjo-Garcia}}, \bibinfo {author} {\bibfnamefont {M.}~\bibnamefont
  {Moreno-Cardoner}}, \bibinfo {author} {\bibfnamefont {A.}~\bibnamefont
  {Albrecht}}, \bibinfo {author} {\bibfnamefont {H.}~\bibnamefont {Kimble}},\
  and\ \bibinfo {author} {\bibfnamefont {D.~E.}\ \bibnamefont {Chang}},\
  }\bibfield  {title} {\bibinfo {title} {Exponential improvement in photon
  storage fidelities using subradiance and “selective radiance” in atomic
  arrays},\ }\href {https://doi.org/10.1103/PhysRevX.7.031024} {\bibfield
  {journal} {\bibinfo  {journal} {Physical Review X}\ }\textbf {\bibinfo
  {volume} {7}},\ \bibinfo {pages} {031024} (\bibinfo {year}
  {2017})}\BibitemShut {NoStop}%
\bibitem [{\citenamefont {Ruskai}(1994)}]{ruskai1994beyond}%
  \BibitemOpen
  \bibfield  {author} {\bibinfo {author} {\bibfnamefont {M.~B.}\ \bibnamefont
  {Ruskai}},\ }\bibfield  {title} {\bibinfo {title} {Beyond strong
  subadditivity? improved bounds on the contraction of generalized relative
  entropy},\ }\href {https://doi.org/10.1142/S0129055X94000407} {\bibfield
  {journal} {\bibinfo  {journal} {Reviews in Mathematical Physics}\ }\textbf
  {\bibinfo {volume} {6}},\ \bibinfo {pages} {1147} (\bibinfo {year}
  {1994})}\BibitemShut {NoStop}%
\bibitem [{\citenamefont {Wolf}(2012)}]{wolf2012quantum}%
  \BibitemOpen
  \bibfield  {author} {\bibinfo {author} {\bibfnamefont {M.~M.}\ \bibnamefont
  {Wolf}},\ }\href {https://mediatum.ub.tum.de/node?id=1701036} {\bibinfo
  {title} {Quantum channels and operations-guided tour}} (\bibinfo {year}
  {2012}),\ \bibinfo {note} {graue Literatur}\BibitemShut {NoStop}%
\bibitem [{\citenamefont {Aharonov}\ \emph {et~al.}(1996)\citenamefont
  {Aharonov}, \citenamefont {Ben-Or}, \citenamefont {Impagliazzo},\ and\
  \citenamefont {Nisan}}]{aharonov1996limitations}%
  \BibitemOpen
  \bibfield  {author} {\bibinfo {author} {\bibfnamefont {D.}~\bibnamefont
  {Aharonov}}, \bibinfo {author} {\bibfnamefont {M.}~\bibnamefont {Ben-Or}},
  \bibinfo {author} {\bibfnamefont {R.}~\bibnamefont {Impagliazzo}},\ and\
  \bibinfo {author} {\bibfnamefont {N.}~\bibnamefont {Nisan}},\ }\href@noop {}
  {\bibinfo {title} {Limitations of noisy reversible computation}} (\bibinfo
  {year} {1996}),\ \Eprint {https://arxiv.org/abs/quant-ph/9611028}
  {arXiv:quant-ph/9611028 [quant-ph]} \BibitemShut {NoStop}%
\bibitem [{\citenamefont {Alicki}\ \emph {et~al.}(2009)\citenamefont {Alicki},
  \citenamefont {Fannes},\ and\ \citenamefont
  {Horodecki}}]{alicki2009thermalization}%
  \BibitemOpen
  \bibfield  {author} {\bibinfo {author} {\bibfnamefont {R.}~\bibnamefont
  {Alicki}}, \bibinfo {author} {\bibfnamefont {M.}~\bibnamefont {Fannes}},\
  and\ \bibinfo {author} {\bibfnamefont {M.}~\bibnamefont {Horodecki}},\
  }\bibfield  {title} {\bibinfo {title} {On thermalization in kitaev's 2d
  model},\ }\href {https://doi.org/10.1088/1751-8113/42/6/065303} {\bibfield
  {journal} {\bibinfo  {journal} {Journal of Physics A: Mathematical and
  Theoretical}\ }\textbf {\bibinfo {volume} {42}},\ \bibinfo {pages} {065303}
  (\bibinfo {year} {2009})}\BibitemShut {NoStop}%
\bibitem [{\citenamefont {Temme}\ and\ \citenamefont
  {Kastoryano}(2015)}]{temme2015fast}%
  \BibitemOpen
  \bibfield  {author} {\bibinfo {author} {\bibfnamefont {K.}~\bibnamefont
  {Temme}}\ and\ \bibinfo {author} {\bibfnamefont {M.~J.}\ \bibnamefont
  {Kastoryano}},\ }\href@noop {} {\bibinfo {title} {How fast do stabilizer
  hamiltonians thermalize?}} (\bibinfo {year} {2015}),\ \Eprint
  {https://arxiv.org/abs/1505.07811} {arXiv:1505.07811 [quant-ph]} \BibitemShut
  {NoStop}%
\bibitem [{\citenamefont {Temme}(2017)}]{temme2017thermalization}%
  \BibitemOpen
  \bibfield  {author} {\bibinfo {author} {\bibfnamefont {K.}~\bibnamefont
  {Temme}},\ }\bibfield  {title} {\bibinfo {title} {Thermalization time bounds
  for pauli stabilizer hamiltonians},\ }\href
  {https://doi.org/10.1007/s00220-016-2746-0} {\bibfield  {journal} {\bibinfo
  {journal} {Communications in Mathematical Physics}\ }\textbf {\bibinfo
  {volume} {350}},\ \bibinfo {pages} {603} (\bibinfo {year}
  {2017})}\BibitemShut {NoStop}%
\bibitem [{\citenamefont {Mishra}\ \emph {et~al.}(2024)\citenamefont {Mishra},
  \citenamefont {Fr{\'\i}as-P{\'e}rez},\ and\ \citenamefont
  {Trivedi}}]{mishra2024classically}%
  \BibitemOpen
  \bibfield  {author} {\bibinfo {author} {\bibfnamefont {S.~D.}\ \bibnamefont
  {Mishra}}, \bibinfo {author} {\bibfnamefont {M.}~\bibnamefont
  {Fr{\'\i}as-P{\'e}rez}},\ and\ \bibinfo {author} {\bibfnamefont
  {R.}~\bibnamefont {Trivedi}},\ }\bibfield  {title} {\bibinfo {title}
  {Classically computing performance bounds on depolarized quantum circuits},\
  }\href {https://doi.org/10.1103/PRXQuantum.5.020317} {\bibfield  {journal}
  {\bibinfo  {journal} {PRX Quantum}\ }\textbf {\bibinfo {volume} {5}},\
  \bibinfo {pages} {020317} (\bibinfo {year} {2024})}\BibitemShut {NoStop}%
\bibitem [{\citenamefont {De~Palma}\ \emph {et~al.}(2023)\citenamefont
  {De~Palma}, \citenamefont {Marvian}, \citenamefont {Rouz{\'e}},\ and\
  \citenamefont {Fran{\c{c}}a}}]{de2023limitations}%
  \BibitemOpen
  \bibfield  {author} {\bibinfo {author} {\bibfnamefont {G.}~\bibnamefont
  {De~Palma}}, \bibinfo {author} {\bibfnamefont {M.}~\bibnamefont {Marvian}},
  \bibinfo {author} {\bibfnamefont {C.}~\bibnamefont {Rouz{\'e}}},\ and\
  \bibinfo {author} {\bibfnamefont {D.~S.}\ \bibnamefont {Fran{\c{c}}a}},\
  }\bibfield  {title} {\bibinfo {title} {Limitations of variational quantum
  algorithms: a quantum optimal transport approach},\ }\href
  {https://doi.org/10.1103/PRXQuantum.4.010309} {\bibfield  {journal} {\bibinfo
   {journal} {PRX Quantum}\ }\textbf {\bibinfo {volume} {4}},\ \bibinfo {pages}
  {010309} (\bibinfo {year} {2023})}\BibitemShut {NoStop}%
\bibitem [{\citenamefont {Stilck~Fran{\c{c}}a}\ and\ \citenamefont
  {Garcia-Patron}(2021)}]{stilck2021limitations}%
  \BibitemOpen
  \bibfield  {author} {\bibinfo {author} {\bibfnamefont {D.}~\bibnamefont
  {Stilck~Fran{\c{c}}a}}\ and\ \bibinfo {author} {\bibfnamefont
  {R.}~\bibnamefont {Garcia-Patron}},\ }\bibfield  {title} {\bibinfo {title}
  {Limitations of optimization algorithms on noisy quantum devices},\ }\href
  {https://doi.org/10.1038/s41567-021-01356-3} {\bibfield  {journal} {\bibinfo
  {journal} {Nature Physics}\ }\textbf {\bibinfo {volume} {17}},\ \bibinfo
  {pages} {1221} (\bibinfo {year} {2021})}\BibitemShut {NoStop}%
\bibitem [{\citenamefont {Abbasgholinejad}\ \emph {et~al.}(2025)\citenamefont
  {Abbasgholinejad}, \citenamefont {Malz}, \citenamefont {Asenjo-Garcia},\ and\
  \citenamefont {Trivedi}}]{abbasgholinejad2025theory}%
  \BibitemOpen
  \bibfield  {author} {\bibinfo {author} {\bibfnamefont {E.}~\bibnamefont
  {Abbasgholinejad}}, \bibinfo {author} {\bibfnamefont {D.}~\bibnamefont
  {Malz}}, \bibinfo {author} {\bibfnamefont {A.}~\bibnamefont
  {Asenjo-Garcia}},\ and\ \bibinfo {author} {\bibfnamefont {R.}~\bibnamefont
  {Trivedi}},\ }\href@noop {} {\bibinfo {title} {Theory of quantum-enhanced
  interferometry with general markovian light sources}} (\bibinfo {year}
  {2025}),\ \Eprint {https://arxiv.org/abs/2504.05111v2} {arXiv:2504.05111v2
  [quant-ph]} \BibitemShut {NoStop}%
\bibitem [{\citenamefont {Yang}\ \emph {et~al.}(2023)\citenamefont {Yang},
  \citenamefont {Huelga},\ and\ \citenamefont {Plenio}}]{yang2023efficient}%
  \BibitemOpen
  \bibfield  {author} {\bibinfo {author} {\bibfnamefont {D.}~\bibnamefont
  {Yang}}, \bibinfo {author} {\bibfnamefont {S.~F.}\ \bibnamefont {Huelga}},\
  and\ \bibinfo {author} {\bibfnamefont {M.~B.}\ \bibnamefont {Plenio}},\
  }\bibfield  {title} {\bibinfo {title} {Efficient information retrieval for
  sensing via continuous measurement},\ }\href
  {https://doi.org/10.1103/PhysRevX.13.031012} {\bibfield  {journal} {\bibinfo
  {journal} {Physical Review X}\ }\textbf {\bibinfo {volume} {13}},\ \bibinfo
  {pages} {031012} (\bibinfo {year} {2023})}\BibitemShut {NoStop}%
\bibitem [{\citenamefont {Gammelmark}\ and\ \citenamefont
  {M{\o}lmer}(2014)}]{gammelmark2014fisher}%
  \BibitemOpen
  \bibfield  {author} {\bibinfo {author} {\bibfnamefont {S.}~\bibnamefont
  {Gammelmark}}\ and\ \bibinfo {author} {\bibfnamefont {K.}~\bibnamefont
  {M{\o}lmer}},\ }\bibfield  {title} {\bibinfo {title} {Fisher information and
  the quantum cram{\'e}r-rao sensitivity limit of continuous measurements},\
  }\href {https://doi.org/10.1103/PhysRevLett.112.170401} {\bibfield  {journal}
  {\bibinfo  {journal} {Physical review letters}\ }\textbf {\bibinfo {volume}
  {112}},\ \bibinfo {pages} {170401} (\bibinfo {year} {2014})}\BibitemShut
  {NoStop}%
\bibitem [{\citenamefont {Wu}\ \emph {et~al.}(2007)\citenamefont {Wu},
  \citenamefont {Pechen}, \citenamefont {Brif},\ and\ \citenamefont
  {Rabitz}}]{wu2007controllability}%
  \BibitemOpen
  \bibfield  {author} {\bibinfo {author} {\bibfnamefont {R.}~\bibnamefont
  {Wu}}, \bibinfo {author} {\bibfnamefont {A.}~\bibnamefont {Pechen}}, \bibinfo
  {author} {\bibfnamefont {C.}~\bibnamefont {Brif}},\ and\ \bibinfo {author}
  {\bibfnamefont {H.}~\bibnamefont {Rabitz}},\ }\bibfield  {title} {\bibinfo
  {title} {Controllability of open quantum systems with kraus-map dynamics},\
  }\href {https://doi.org/10.1088/1751-8113/40/21/015} {\bibfield  {journal}
  {\bibinfo  {journal} {Journal of Physics A: Mathematical and Theoretical}\
  }\textbf {\bibinfo {volume} {40}},\ \bibinfo {pages} {5681} (\bibinfo {year}
  {2007})}\BibitemShut {NoStop}%
\bibitem [{\citenamefont {Kurniawan}\ \emph {et~al.}(2012)\citenamefont
  {Kurniawan}, \citenamefont {Dirr},\ and\ \citenamefont
  {Helmke}}]{kurniawan2012controllability}%
  \BibitemOpen
  \bibfield  {author} {\bibinfo {author} {\bibfnamefont {I.}~\bibnamefont
  {Kurniawan}}, \bibinfo {author} {\bibfnamefont {G.}~\bibnamefont {Dirr}},\
  and\ \bibinfo {author} {\bibfnamefont {U.}~\bibnamefont {Helmke}},\
  }\bibfield  {title} {\bibinfo {title} {Controllability aspects of quantum
  dynamics: a unified approach for closed and open systems},\ }\href
  {https://doi.org/10.1109/TAC.2012.2195870} {\bibfield  {journal} {\bibinfo
  {journal} {IEEE transactions on automatic control}\ }\textbf {\bibinfo
  {volume} {57}},\ \bibinfo {pages} {1984} (\bibinfo {year}
  {2012})}\BibitemShut {NoStop}%
\bibitem [{\citenamefont {Koch}(2016)}]{koch2016controlling}%
  \BibitemOpen
  \bibfield  {author} {\bibinfo {author} {\bibfnamefont {C.~P.}\ \bibnamefont
  {Koch}},\ }\bibfield  {title} {\bibinfo {title} {Controlling open quantum
  systems: tools, achievements, and limitations},\ }\href
  {https://doi.org/10.1088/0953-8984/28/21/213001} {\bibfield  {journal}
  {\bibinfo  {journal} {Journal of Physics: Condensed Matter}\ }\textbf
  {\bibinfo {volume} {28}},\ \bibinfo {pages} {213001} (\bibinfo {year}
  {2016})}\BibitemShut {NoStop}%
\bibitem [{\citenamefont {Kallush}\ \emph {et~al.}(2022)\citenamefont
  {Kallush}, \citenamefont {Dann},\ and\ \citenamefont
  {Kosloff}}]{kallush2022controlling}%
  \BibitemOpen
  \bibfield  {author} {\bibinfo {author} {\bibfnamefont {S.}~\bibnamefont
  {Kallush}}, \bibinfo {author} {\bibfnamefont {R.}~\bibnamefont {Dann}},\ and\
  \bibinfo {author} {\bibfnamefont {R.}~\bibnamefont {Kosloff}},\ }\bibfield
  {title} {\bibinfo {title} {Controlling the uncontrollable: Quantum control of
  open-system dynamics},\ }\href {https://doi.org/10.1126/sciadv.add0828}
  {\bibfield  {journal} {\bibinfo  {journal} {Science Advances}\ }\textbf
  {\bibinfo {volume} {8}},\ \bibinfo {pages} {eadd0828} (\bibinfo {year}
  {2022})}\BibitemShut {NoStop}%
\bibitem [{\citenamefont {Yoshida}(2024)}]{yoshida2024uniqueness}%
  \BibitemOpen
  \bibfield  {author} {\bibinfo {author} {\bibfnamefont {H.}~\bibnamefont
  {Yoshida}},\ }\bibfield  {title} {\bibinfo {title} {Uniqueness of steady
  states of gorini-kossakowski-sudarshan-lindblad equations: A simple proof},\
  }\href {https://doi.org/10.1103/PhysRevA.109.022218} {\bibfield  {journal}
  {\bibinfo  {journal} {Physical Review A}\ }\textbf {\bibinfo {volume}
  {109}},\ \bibinfo {pages} {022218} (\bibinfo {year} {2024})}\BibitemShut
  {NoStop}%
\bibitem [{\citenamefont {Frigerio}(1977)}]{frigerio1977quantum}%
  \BibitemOpen
  \bibfield  {author} {\bibinfo {author} {\bibfnamefont {A.}~\bibnamefont
  {Frigerio}},\ }\bibfield  {title} {\bibinfo {title} {Quantum dynamical
  semigroups and approach to equilibrium},\ }\href
  {https://doi.org/10.1007/BF00398571} {\bibfield  {journal} {\bibinfo
  {journal} {Letters in Mathematical Physics}\ }\textbf {\bibinfo {volume}
  {2}},\ \bibinfo {pages} {79} (\bibinfo {year} {1977})}\BibitemShut {NoStop}%
\bibitem [{\citenamefont {Frigerio}(1978)}]{frigerio1978stationary}%
  \BibitemOpen
  \bibfield  {author} {\bibinfo {author} {\bibfnamefont {A.}~\bibnamefont
  {Frigerio}},\ }\bibfield  {title} {\bibinfo {title} {Stationary states of
  quantum dynamical semigroups},\ }\href {https://doi.org/10.1007/BF01196936}
  {\bibfield  {journal} {\bibinfo  {journal} {Communications in Mathematical
  Physics}\ }\textbf {\bibinfo {volume} {63}},\ \bibinfo {pages} {269}
  (\bibinfo {year} {1978})}\BibitemShut {NoStop}%
\bibitem [{\citenamefont {Spohn}(1976)}]{spohn1976approach}%
  \BibitemOpen
  \bibfield  {author} {\bibinfo {author} {\bibfnamefont {H.}~\bibnamefont
  {Spohn}},\ }\bibfield  {title} {\bibinfo {title} {Approach to equilibrium for
  completely positive dynamical semigroups of n-level systems},\ }\href
  {https://doi.org/10.1016/0034-4877(76)90040-9} {\bibfield  {journal}
  {\bibinfo  {journal} {Reports on Mathematical Physics}\ }\textbf {\bibinfo
  {volume} {10}},\ \bibinfo {pages} {189} (\bibinfo {year} {1976})}\BibitemShut
  {NoStop}%
\bibitem [{\citenamefont {Spohn}(1977)}]{Spohn1977}%
  \BibitemOpen
  \bibfield  {author} {\bibinfo {author} {\bibfnamefont {H.}~\bibnamefont
  {Spohn}},\ }\bibfield  {title} {{\selectlanguage {english}\bibinfo {title}
  {An algebraic condition for the approach to equilibrium of an open n-level
  system}},\ }\href {https://doi.org/10.1007/BF00420668} {\bibfield  {journal}
  {\bibinfo  {journal} {Lett. Math. Phys.}\ }\textbf {\bibinfo {volume} {2}},\
  \bibinfo {pages} {33} (\bibinfo {year} {1977})}\BibitemShut {NoStop}%
\bibitem [{\citenamefont {Evans}(1977)}]{evans1977irreducible}%
  \BibitemOpen
  \bibfield  {author} {\bibinfo {author} {\bibfnamefont {D.~E.}\ \bibnamefont
  {Evans}},\ }\bibfield  {title} {\bibinfo {title} {Irreducible quantum
  dynamical semigroups},\ }\href {https://doi.org/10.1007/BF01614091}
  {\bibfield  {journal} {\bibinfo  {journal} {Communications in Mathematical
  Physics}\ }\textbf {\bibinfo {volume} {54}},\ \bibinfo {pages} {293}
  (\bibinfo {year} {1977})}\BibitemShut {NoStop}%
\bibitem [{\citenamefont {Temme}\ \emph {et~al.}(2010)\citenamefont {Temme},
  \citenamefont {Kastoryano}, \citenamefont {Ruskai}, \citenamefont {Wolf},\
  and\ \citenamefont {Verstraete}}]{Temme2010}%
  \BibitemOpen
  \bibfield  {author} {\bibinfo {author} {\bibfnamefont {K.}~\bibnamefont
  {Temme}}, \bibinfo {author} {\bibfnamefont {M.~J.}\ \bibnamefont
  {Kastoryano}}, \bibinfo {author} {\bibfnamefont {M.~B.}\ \bibnamefont
  {Ruskai}}, \bibinfo {author} {\bibfnamefont {M.~M.}\ \bibnamefont {Wolf}},\
  and\ \bibinfo {author} {\bibfnamefont {F.}~\bibnamefont {Verstraete}},\
  }\bibfield  {title} {\bibinfo {title} {The $\chi$2-divergence and mixing
  times of quantum markov processes},\ }\href
  {https://doi.org/10.1063/1.3511335} {\bibfield  {journal} {\bibinfo
  {journal} {J. Math. Phys.}\ }\textbf {\bibinfo {volume} {51}},\ \bibinfo
  {pages} {122201} (\bibinfo {year} {2010})}\BibitemShut {NoStop}%
\bibitem [{\citenamefont {Kastoryano}\ and\ \citenamefont
  {Temme}(2013)}]{Kastoryano2013-ot}%
  \BibitemOpen
  \bibfield  {author} {\bibinfo {author} {\bibfnamefont {M.~J.}\ \bibnamefont
  {Kastoryano}}\ and\ \bibinfo {author} {\bibfnamefont {K.}~\bibnamefont
  {Temme}},\ }\bibfield  {title} {{\selectlanguage {english}\bibinfo {title}
  {Quantum logarithmic sobolev inequalities and rapid mixing}},\ }\href
  {https://doi.org/10.1063/1.4804995} {\bibfield  {journal} {\bibinfo
  {journal} {J. Math. Phys.}\ }\textbf {\bibinfo {volume} {54}},\ \bibinfo
  {pages} {052202} (\bibinfo {year} {2013})}\BibitemShut {NoStop}%
\bibitem [{\citenamefont {Capel}\ \emph {et~al.}(2018)\citenamefont {Capel},
  \citenamefont {Lucia},\ and\ \citenamefont
  {P{\'e}rez-Garc{\'\i}a}}]{capel2018quantum}%
  \BibitemOpen
  \bibfield  {author} {\bibinfo {author} {\bibfnamefont {{\'A}.}~\bibnamefont
  {Capel}}, \bibinfo {author} {\bibfnamefont {A.}~\bibnamefont {Lucia}},\ and\
  \bibinfo {author} {\bibfnamefont {D.}~\bibnamefont {P{\'e}rez-Garc{\'\i}a}},\
  }\bibfield  {title} {\bibinfo {title} {Quantum conditional relative entropy
  and quasi-factorization of the relative entropy},\ }\href
  {https://doi.org/10.1088/1751-8121/aae4cf} {\bibfield  {journal} {\bibinfo
  {journal} {Journal of Physics A: Mathematical and Theoretical}\ }\textbf
  {\bibinfo {volume} {51}},\ \bibinfo {pages} {484001} (\bibinfo {year}
  {2018})}\BibitemShut {NoStop}%
\bibitem [{\citenamefont {Ángela Capel}\ \emph {et~al.}(2021)\citenamefont
  {Ángela Capel}, \citenamefont {Rouzé},\ and\ \citenamefont
  {França}}]{capel2020modified}%
  \BibitemOpen
  \bibfield  {author} {\bibinfo {author} {\bibnamefont {Ángela Capel}},
  \bibinfo {author} {\bibfnamefont {C.}~\bibnamefont {Rouzé}},\ and\ \bibinfo
  {author} {\bibfnamefont {D.~S.}\ \bibnamefont {França}},\ }\href@noop {}
  {\bibinfo {title} {The modified logarithmic sobolev inequality for quantum
  spin systems: classical and commuting nearest neighbour interactions}}
  (\bibinfo {year} {2021}),\ \Eprint {https://arxiv.org/abs/2009.11817v2}
  {arXiv:2009.11817v2 [quant-ph]} \BibitemShut {NoStop}%
\bibitem [{\citenamefont {M{\"u}ller-Hermes}\ and\ \citenamefont
  {Franca}(2018)}]{muller2018sandwiched}%
  \BibitemOpen
  \bibfield  {author} {\bibinfo {author} {\bibfnamefont {A.}~\bibnamefont
  {M{\"u}ller-Hermes}}\ and\ \bibinfo {author} {\bibfnamefont {D.~S.}\
  \bibnamefont {Franca}},\ }\bibfield  {title} {\bibinfo {title} {Sandwiched
  r{\'e}nyi convergence for quantum evolutions},\ }\href
  {https://doi.org/10.22331/q-2018-02-27-55} {\bibfield  {journal} {\bibinfo
  {journal} {Quantum}\ }\textbf {\bibinfo {volume} {2}},\ \bibinfo {pages} {55}
  (\bibinfo {year} {2018})}\BibitemShut {NoStop}%
\bibitem [{\citenamefont {Temme}\ \emph {et~al.}(2014)\citenamefont {Temme},
  \citenamefont {Pastawski},\ and\ \citenamefont
  {Kastoryano}}]{temme2014hypercontractivity}%
  \BibitemOpen
  \bibfield  {author} {\bibinfo {author} {\bibfnamefont {K.}~\bibnamefont
  {Temme}}, \bibinfo {author} {\bibfnamefont {F.}~\bibnamefont {Pastawski}},\
  and\ \bibinfo {author} {\bibfnamefont {M.~J.}\ \bibnamefont {Kastoryano}},\
  }\bibfield  {title} {\bibinfo {title} {Hypercontractivity of quasi-free
  quantum semigroups},\ }\href {https://doi.org/10.1088/1751-8113/47/40/405303}
  {\bibfield  {journal} {\bibinfo  {journal} {Journal of Physics A:
  Mathematical and Theoretical}\ }\textbf {\bibinfo {volume} {47}},\ \bibinfo
  {pages} {405303} (\bibinfo {year} {2014})}\BibitemShut {NoStop}%
\bibitem [{\citenamefont {Montanaro}(2012)}]{montanaro2012some}%
  \BibitemOpen
  \bibfield  {author} {\bibinfo {author} {\bibfnamefont {A.}~\bibnamefont
  {Montanaro}},\ }\bibfield  {title} {\bibinfo {title} {Some applications of
  hypercontractive inequalities in quantum information theory},\ }\bibfield
  {journal} {\bibinfo  {journal} {Journal of Mathematical Physics}\ }\textbf
  {\bibinfo {volume} {53}},\ \href {https://doi.org/10.1063/1.4769269}
  {10.1063/1.4769269} (\bibinfo {year} {2012})\BibitemShut {NoStop}%
\bibitem [{\citenamefont {Beigi}\ \emph {et~al.}(2020)\citenamefont {Beigi},
  \citenamefont {Datta},\ and\ \citenamefont {Rouz{\'e}}}]{beigi2020quantum}%
  \BibitemOpen
  \bibfield  {author} {\bibinfo {author} {\bibfnamefont {S.}~\bibnamefont
  {Beigi}}, \bibinfo {author} {\bibfnamefont {N.}~\bibnamefont {Datta}},\ and\
  \bibinfo {author} {\bibfnamefont {C.}~\bibnamefont {Rouz{\'e}}},\ }\bibfield
  {title} {\bibinfo {title} {Quantum reverse hypercontractivity: its
  tensorization and application to strong converses},\ }\href
  {https://doi.org/10.1007/s00220-020-03750-z} {\bibfield  {journal} {\bibinfo
  {journal} {Communications in Mathematical Physics}\ }\textbf {\bibinfo
  {volume} {376}},\ \bibinfo {pages} {753} (\bibinfo {year}
  {2020})}\BibitemShut {NoStop}%
\bibitem [{\citenamefont {Cubitt}\ \emph {et~al.}(2015)\citenamefont {Cubitt},
  \citenamefont {Kastoryano}, \citenamefont {Montanaro},\ and\ \citenamefont
  {Temme}}]{cubitt2015quantum}%
  \BibitemOpen
  \bibfield  {author} {\bibinfo {author} {\bibfnamefont {T.}~\bibnamefont
  {Cubitt}}, \bibinfo {author} {\bibfnamefont {M.}~\bibnamefont {Kastoryano}},
  \bibinfo {author} {\bibfnamefont {A.}~\bibnamefont {Montanaro}},\ and\
  \bibinfo {author} {\bibfnamefont {K.}~\bibnamefont {Temme}},\ }\bibfield
  {title} {\bibinfo {title} {Quantum reverse hypercontractivity},\ }\bibfield
  {journal} {\bibinfo  {journal} {Journal of Mathematical Physics}\ }\textbf
  {\bibinfo {volume} {56}},\ \href {https://doi.org/10.1063/1.4933219}
  {10.1063/1.4933219} (\bibinfo {year} {2015})\BibitemShut {NoStop}%
\bibitem [{\citenamefont {Fang}\ \emph {et~al.}(2025)\citenamefont {Fang},
  \citenamefont {Lu},\ and\ \citenamefont {Tong}}]{fang2025mixing}%
  \BibitemOpen
  \bibfield  {author} {\bibinfo {author} {\bibfnamefont {D.}~\bibnamefont
  {Fang}}, \bibinfo {author} {\bibfnamefont {J.}~\bibnamefont {Lu}},\ and\
  \bibinfo {author} {\bibfnamefont {Y.}~\bibnamefont {Tong}},\ }\bibfield
  {title} {\bibinfo {title} {Mixing time of open quantum systems via
  hypocoercivity},\ }\href {https://doi.org/10.1103/PhysRevLett.134.140405}
  {\bibfield  {journal} {\bibinfo  {journal} {Physical Review Letters}\
  }\textbf {\bibinfo {volume} {134}},\ \bibinfo {pages} {140405} (\bibinfo
  {year} {2025})}\BibitemShut {NoStop}%
\bibitem [{\citenamefont {Di~Meglio}\ \emph {et~al.}(2025)\citenamefont
  {Di~Meglio}, \citenamefont {Chru{\'s}ci{\'n}ski}, \citenamefont {Audenaert},
  \citenamefont {Plenio},\ and\ \citenamefont {Huelga}}]{Plenio2025}%
  \BibitemOpen
  \bibfield  {author} {\bibinfo {author} {\bibfnamefont {G.}~\bibnamefont
  {Di~Meglio}}, \bibinfo {author} {\bibfnamefont {D.}~\bibnamefont
  {Chru{\'s}ci{\'n}ski}}, \bibinfo {author} {\bibfnamefont {K.}~\bibnamefont
  {Audenaert}}, \bibinfo {author} {\bibfnamefont {M.~B.}\ \bibnamefont
  {Plenio}},\ and\ \bibinfo {author} {\bibfnamefont {S.~F.}\ \bibnamefont
  {Huelga}},\ }\href@noop {} {\bibinfo {title} {Asymptotic relaxation in
  quantum markovian dynamics}} (\bibinfo {year} {2025}),\ \Eprint
  {https://arxiv.org/abs/2410.14313v2} {arXiv:2410.14313v2 [quant-ph]}
  \BibitemShut {NoStop}%
\bibitem [{Note1()}]{Note1}%
  \BibitemOpen
  \bibinfo {note} {For the solution to the time-dependent Lindblad equation to
  exist, a sufficient condition is for $\protect \mathcal {L}_t$ to be locally
  integrable, i.e. $\DOTSI \intop \ilimits@ _s^t \protect \mathcal {L}_{\tau }
  d \tau $ is well-defined for all finite intervals $[s,t] \subset \protect
  \mathbb {R}$. We shall always consider this to be the case, hence we only
  consider driven Lindbladians whose time-dependent Hamiltonians $H(t)$ have
  components that are locally integrable. Importantly, when we make statements
  about some result concerning $\protect \mathcal {L}_t$ holding \protect \text
  {``independently"} of the Hamiltonian or holding for \protect \text {``all"}
  Hamiltonians $H(t)$ that can appear in \protect \cref {eq:time_dep_lindblad},
  what we really mean is that the statement holds for all Hamiltonians with
  locally integrable components.}\BibitemShut {Stop}%
\bibitem [{\citenamefont {Baake}\ and\ \citenamefont
  {Schl{\"a}gel}(2011)}]{Baake2011-xi}%
  \BibitemOpen
  \bibfield  {author} {\bibinfo {author} {\bibfnamefont {M.}~\bibnamefont
  {Baake}}\ and\ \bibinfo {author} {\bibfnamefont {U.}~\bibnamefont
  {Schl{\"a}gel}},\ }\bibfield  {title} {{\selectlanguage {english}\bibinfo
  {title} {The peano-baker series}},\ }\href
  {https://doi.org/10.1134/S0081543811080098} {\bibfield  {journal} {\bibinfo
  {journal} {Proc. Steklov Inst. Math.}\ }\textbf {\bibinfo {volume} {275}},\
  \bibinfo {pages} {155} (\bibinfo {year} {2011})}\BibitemShut {NoStop}%
\bibitem [{\citenamefont {Hale}(2009)}]{Hale2009-sf}%
  \BibitemOpen
  \bibfield  {author} {\bibinfo {author} {\bibfnamefont {J.~K.}\ \bibnamefont
  {Hale}},\ }\href@noop {} {\emph {\bibinfo {title} {Ordinary Differential
  Equations}}},\ Dover Books on Mathematics\ (\bibinfo  {publisher} {Dover
  Publications},\ \bibinfo {address} {Mineola, NY},\ \bibinfo {year}
  {2009})\BibitemShut {NoStop}%
\bibitem [{Note2()}]{Note2}%
  \BibitemOpen
  \bibinfo {note} {In particular, if $\protect \tr (L_i) \protect \neq 0$, then
  do the transformation \begin {align*} L_j &\to L_j - \protect \text {Tr}(L_j)
  \tau \ \protect \text { and} \\ H(t) &\to H(t) + \protect \frac {i}{2d}
  \DOTSB \sum@ \slimits@ _j ( \protect \tr (L_j)^{*} L_j - \protect \tr (L_j)
  L_j^\dagger ), \end {align*} where $\tau = I/d$ is the maximally mixed state.
  This transformation leaves the Lindbladian $\protect \mathcal {L}_\tau $
  unchanged while enforcing $\protect \tr (L_j) = 0$.}\BibitemShut {Stop}%
\bibitem [{\citenamefont {Baumgartner}\ and\ \citenamefont
  {Narnhofer}(2008)}]{Baumgartner2008-nb}%
  \BibitemOpen
  \bibfield  {author} {\bibinfo {author} {\bibfnamefont {B.}~\bibnamefont
  {Baumgartner}}\ and\ \bibinfo {author} {\bibfnamefont {H.}~\bibnamefont
  {Narnhofer}},\ }\bibfield  {title} {\bibinfo {title} {Analysis of quantum
  semigroups with {GKS--Lindblad} generators: {II}. general},\ }\href
  {https://doi.org/10.1088/1751-8113/41/39/395303} {\bibfield  {journal}
  {\bibinfo  {journal} {J. Phys. A Math. Theor.}\ }\textbf {\bibinfo {volume}
  {41}},\ \bibinfo {pages} {395303} (\bibinfo {year} {2008})}\BibitemShut
  {NoStop}%
\bibitem [{\citenamefont {Anton}(1998)}]{Anton1998}%
  \BibitemOpen
  \bibfield  {author} {\bibinfo {author} {\bibfnamefont {H.}~\bibnamefont
  {Anton}},\ }\href@noop {} {{\selectlanguage {english}\emph {\bibinfo {title}
  {Calculus: Manual}}}},\ \bibinfo {edition} {6th}\ ed.\ (\bibinfo  {publisher}
  {John Wiley and Sons (WIE)},\ \bibinfo {address} {Brisbane, QLD, Australia},\
  \bibinfo {year} {1998})\BibitemShut {NoStop}%
\bibitem [{\citenamefont {Apostol}(1967)}]{Apostol1967}%
  \BibitemOpen
  \bibfield  {author} {\bibinfo {author} {\bibfnamefont {T.~M.}\ \bibnamefont
  {Apostol}},\ }\href@noop {} {\emph {\bibinfo {title} {Calculus}}},\ \bibinfo
  {edition} {2nd}\ ed.\ (\bibinfo  {publisher} {John Wiley \& Sons},\ \bibinfo
  {address} {Nashville, TN},\ \bibinfo {year} {1967})\BibitemShut {NoStop}%
\bibitem [{\citenamefont {Watrous}(2018)}]{Watrous2018}%
  \BibitemOpen
  \bibfield  {author} {\bibinfo {author} {\bibfnamefont {J.}~\bibnamefont
  {Watrous}},\ }\href {https://doi.org/10.1017/9781316848142} {{\selectlanguage
  {english}\emph {\bibinfo {title} {The theory of quantum information}}}}\
  (\bibinfo  {publisher} {Cambridge University Press},\ \bibinfo {address}
  {Cambridge, England},\ \bibinfo {year} {2018})\BibitemShut {NoStop}%
\bibitem [{\citenamefont {Edmonds}(1967)}]{Edmonds1967-dk}%
  \BibitemOpen
  \bibfield  {author} {\bibinfo {author} {\bibfnamefont {J.}~\bibnamefont
  {Edmonds}},\ }\bibfield  {title} {\bibinfo {title} {Optimum branchings},\
  }\href@noop {} {\bibfield  {journal} {\bibinfo  {journal} {J. Res. Natl. Bur.
  Stand.}\ }\textbf {\bibinfo {volume} {71B}},\ \bibinfo {pages} {233}
  (\bibinfo {year} {1967})}\BibitemShut {NoStop}%
\bibitem [{\citenamefont {Harary}(1969)}]{Harary1969-lq}%
  \BibitemOpen
  \bibfield  {author} {\bibinfo {author} {\bibfnamefont {F.}~\bibnamefont
  {Harary}},\ }\href@noop {} {{\selectlanguage {english}\emph {\bibinfo {title}
  {Graph Theory}}}},\ Addison-Wesley series in mathematics\ (\bibinfo
  {publisher} {Addison Wesley Longman Publishing},\ \bibinfo {address} {New
  York, NY},\ \bibinfo {year} {1969})\BibitemShut {NoStop}%
\bibitem [{\citenamefont {Bang-Jensen}\ and\ \citenamefont
  {Gutin}(2001)}]{Bang-Jensen2001-hu}%
  \BibitemOpen
  \bibfield  {author} {\bibinfo {author} {\bibfnamefont {J.}~\bibnamefont
  {Bang-Jensen}}\ and\ \bibinfo {author} {\bibfnamefont {G.~Z.}\ \bibnamefont
  {Gutin}},\ }\href@noop {} {{\selectlanguage {english}\emph {\bibinfo {title}
  {Digraphs}}}},\ \bibinfo {edition} {1st}\ ed.,\ Springer Monographs in
  Mathematics\ (\bibinfo  {publisher} {Springer},\ \bibinfo {address}
  {Guildford, England},\ \bibinfo {year} {2001})\BibitemShut {NoStop}%
\bibitem [{\citenamefont {Tak{\'a}cs}(1990)}]{Takacs1990-bc}%
  \BibitemOpen
  \bibfield  {author} {\bibinfo {author} {\bibfnamefont {L.}~\bibnamefont
  {Tak{\'a}cs}},\ }\bibfield  {title} {{\selectlanguage {english}\bibinfo
  {title} {On cayley's formula for counting forests}},\ }\href
  {https://doi.org/10.1016/0097-3165(90)90064-4} {\bibfield  {journal}
  {\bibinfo  {journal} {J. Comb. Theory Ser. A.}\ }\textbf {\bibinfo {volume}
  {53}},\ \bibinfo {pages} {321} (\bibinfo {year} {1990})}\BibitemShut
  {NoStop}%
\bibitem [{\citenamefont {Kato}(1995)}]{Kato1995}%
  \BibitemOpen
  \bibfield  {author} {\bibinfo {author} {\bibfnamefont {T.}~\bibnamefont
  {Kato}},\ }\href {https://doi.org/10.1007/978-3-642-66282-9}
  {{\selectlanguage {english}\emph {\bibinfo {title} {Perturbation theory for
  linear operators}}}},\ \bibinfo {edition} {2nd}\ ed.,\ Classics in
  Mathematics\ (\bibinfo  {publisher} {Springer},\ \bibinfo {address} {Berlin,
  Germany},\ \bibinfo {year} {1995})\BibitemShut {NoStop}%
\bibitem [{\citenamefont {Bartle}(1995)}]{Bartle1995-hg}%
  \BibitemOpen
  \bibfield  {author} {\bibinfo {author} {\bibfnamefont {R.~G.}\ \bibnamefont
  {Bartle}},\ }\href@noop {} {{\selectlanguage {english}\emph {\bibinfo {title}
  {The elements of integration and Lebesgue measure}}}},\ Wiley Classics
  Library\ (\bibinfo  {publisher} {John Wiley \& Sons},\ \bibinfo {address}
  {Nashville, TN},\ \bibinfo {year} {1995})\BibitemShut {NoStop}%
\bibitem [{\citenamefont {Burnside}(1905)}]{Burnside1905}%
  \BibitemOpen
  \bibfield  {author} {\bibinfo {author} {\bibfnamefont {W.}~\bibnamefont
  {Burnside}},\ }\bibfield  {title} {{\selectlanguage {english}\bibinfo {title}
  {On the condition of reducibility of any group of linear substitutions}},\
  }\href {https://doi.org/10.1112/plms/s2-3.1.430} {\bibfield  {journal}
  {\bibinfo  {journal} {Proc. Lond. Math. Soc. (3)}\ }\textbf {\bibinfo
  {volume} {s2-3}},\ \bibinfo {pages} {430} (\bibinfo {year}
  {1905})}\BibitemShut {NoStop}%
\bibitem [{\citenamefont {Halperin}\ and\ \citenamefont
  {Rosenthal}(1980)}]{Halperin1980}%
  \BibitemOpen
  \bibfield  {author} {\bibinfo {author} {\bibfnamefont {I.}~\bibnamefont
  {Halperin}}\ and\ \bibinfo {author} {\bibfnamefont {P.}~\bibnamefont
  {Rosenthal}},\ }\bibfield  {title} {{\selectlanguage {english}\bibinfo
  {title} {Burnside's theorem on algebras of matrices}},\ }\href
  {https://doi.org/10.1080/00029890.1980.11995157} {\bibfield  {journal}
  {\bibinfo  {journal} {Am. Math. Mon.}\ }\textbf {\bibinfo {volume} {87}},\
  \bibinfo {pages} {810} (\bibinfo {year} {1980})}\BibitemShut {NoStop}%
\bibitem [{\citenamefont {Lomonosov}\ and\ \citenamefont
  {Rosenthal}(2004)}]{Lomonosov2004}%
  \BibitemOpen
  \bibfield  {author} {\bibinfo {author} {\bibfnamefont {V.}~\bibnamefont
  {Lomonosov}}\ and\ \bibinfo {author} {\bibfnamefont {P.}~\bibnamefont
  {Rosenthal}},\ }\bibfield  {title} {{\selectlanguage {english}\bibinfo
  {title} {The simplest proof of burnside's theorem on matrix algebras}},\
  }\href {https://doi.org/10.1016/j.laa.2003.08.012} {\bibfield  {journal}
  {\bibinfo  {journal} {Linear Algebra Appl.}\ }\textbf {\bibinfo {volume}
  {383}},\ \bibinfo {pages} {45} (\bibinfo {year} {2004})}\BibitemShut
  {NoStop}%
\bibitem [{\citenamefont {Drake}(2007)}]{drake2007springer}%
  \BibitemOpen
  \bibfield  {author} {\bibinfo {author} {\bibfnamefont {G.~W.}\ \bibnamefont
  {Drake}},\ }\href {https://doi.org/10.1007/978-0-387-26308-3} {\emph
  {\bibinfo {title} {Springer handbook of atomic, molecular, and optical
  physics}}}\ (\bibinfo  {publisher} {Springer Science \& Business Media},\
  \bibinfo {year} {2007})\BibitemShut {NoStop}%
\bibitem [{\citenamefont {Andreev}\ \emph {et~al.}(1980)\citenamefont
  {Andreev}, \citenamefont {Emelyanov},\ and\ \citenamefont
  {Ilinski}}]{andreev1980collective}%
  \BibitemOpen
  \bibfield  {author} {\bibinfo {author} {\bibfnamefont {A.~V.}\ \bibnamefont
  {Andreev}}, \bibinfo {author} {\bibfnamefont {V.~I.}\ \bibnamefont
  {Emelyanov}},\ and\ \bibinfo {author} {\bibfnamefont {Y.~A.}\ \bibnamefont
  {Ilinski}},\ }\bibfield  {title} {\bibinfo {title} {Collective spontaneous
  emission (dicke superradiance)},\ }\href
  {https://doi.org/10.1070/PU1980v023n08ABEH005024} {\bibfield  {journal}
  {\bibinfo  {journal} {Soviet Physics Uspekhi}\ }\textbf {\bibinfo {volume}
  {23}},\ \bibinfo {pages} {493} (\bibinfo {year} {1980})}\BibitemShut
  {NoStop}%
\bibitem [{\citenamefont {Dicke}(1954)}]{dicke1954coherence}%
  \BibitemOpen
  \bibfield  {author} {\bibinfo {author} {\bibfnamefont {R.~H.}\ \bibnamefont
  {Dicke}},\ }\bibfield  {title} {\bibinfo {title} {Coherence in spontaneous
  radiation processes},\ }\href {https://doi.org/10.1103/PhysRev.93.99}
  {\bibfield  {journal} {\bibinfo  {journal} {Physical review}\ }\textbf
  {\bibinfo {volume} {93}},\ \bibinfo {pages} {99} (\bibinfo {year}
  {1954})}\BibitemShut {NoStop}%
\end{thebibliography}%

\end{document}